\documentclass[12pt,draftcls,journal,onecolumn]{IEEEtran}

\usepackage{graphicx}
\usepackage[cmex10]{amsmath}
\usepackage{amsfonts}
\usepackage{amssymb}
\usepackage{mathrsfs}
\usepackage{bm}
\usepackage{algorithmic}
\usepackage{algorithm}
\usepackage{array}
\usepackage{mdwmath}
\usepackage{mdwtab}
\usepackage[caption=false,font=footnotesize]{subfig}
\usepackage{fixltx2e}
\usepackage{stfloats}
\usepackage{footnote}
\usepackage{tabularx}
\usepackage{multirow}
\usepackage{color}
\usepackage{textcomp}

\newtheorem{definition}{Definition}
\newtheorem{proposition}{Proposition}
\newtheorem{theorem}{Theorem}

\begin{document}

\title{Socially-Optimal Design of Service Exchange Platforms with Imperfect Monitoring}

\author{Yuanzhang~Xiao and~Mihaela~van~der~Schaar
\thanks{We would like to thank Yu Zhang for very insightful initial discussions on rating mechanisms,
and Prof. William Zame and Simpson Zhang (Department of Economics, UCLA) for their helpful comments that lead to the improvement of this paper.}}

\maketitle

\begin{abstract}
We study the design of service exchange platforms in which long-lived anonymous users exchange services with each other. The users are randomly
matched into pairs of clients and servers repeatedly, and each server can choose whether to provide high-quality or low-quality services to the
client with whom it is matched. Since the users are anonymous and incur high costs (e.g. exert high effort) in providing high-quality services, it is
crucial that the platform incentivizes users to provide high-quality services. Rating mechanisms have been shown to work effectively as incentive
schemes in such platforms. A rating mechanism labels each user by a rating, which summarizes the user's past behaviors, recommends a desirable
behavior to each server (e.g., provide higher-quality services for clients with higher ratings), and updates each server's rating based on the
recommendation and its client's report on the service quality. Based on this recommendation, a low-rating user is less likely to obtain high-quality
services, thereby providing users with incentives to obtain high ratings by providing high-quality services.

However, if monitoring or reporting is imperfect -- clients do not perfectly assess the quality or the reports are lost -- a user's rating may not be
updated correctly. In the presence of such errors, existing rating mechanisms cannot achieve the social optimum. In this paper, we propose the first
rating mechanism that does achieve the social optimum, even in the presence of monitoring or reporting errors. On one hand, the socially-optimal
rating mechanism needs to be complicated enough, because the optimal recommended behavior depends not only on the current rating distribution, but
also (necessarily) on the history of past rating distributions in the platform. On the other hand, we prove that the social optimum can be achieved
by ``simple'' rating mechanisms that use \emph{binary} rating labels and a \emph{small} set of (three) recommended behaviors. We provide design
guidelines of socially-optimal rating mechanisms, and a low-complexity online algorithm for the rating mechanism to determine the optimal recommended
behavior.
\end{abstract}


\section{Introduction}\label{sec:intro}
Service exchange platforms have proliferated as the medium that allows the users to exchange services valuable to each other. For instance, emerging
new service exchange platforms include crowdsourcing systems (e.g. in Amazon Mechanical Turk and CrowdSource) in which the users exchange labor
\cite{Kittur_08_Crowdsourcing}\cite{YuMihaela_TEAC}, online question-and-answer websites (e.g. in Yahoo! Answers and Quora) in which the users
exchange knowledge \cite{YuMihaela_TEAC}, peer-to-peer (P2P) networks in which the users exchange files/packets
\cite{YuMihaela_P2P}\cite{Blanc}\cite{Feldman_04_P2P}, and online trading platforms (e.g. eBay) where the users exchange goods \cite{Dellarocas}. In
a typical service exchange platform, a user plays a dual role: as a \emph{client}, who requests services, and as a \emph{server}, who chooses to
provide high-quality or low-quality services. Common features of many service exchange platforms are: the user population is large and users are
anonymous. In other words, each user interacts with a randomly-matched partner without knowing its partner's identity (However, the platform does
know the identify of the interacting users.). The absence of a fixed partner and the anonymity of the users create incentive problems -- namely the
users tend to ``free-ride'' (i.e., receive high-quality services from others as a client, while providing low-quality services as a server). In
addition, a user generally may not be able to perfectly monitor\footnote{The monitoring discussed throughout this paper is a user's observation on
its current partner's actions. Each user knows nothing about the ongoing interactions among the other pairs of users.} its partner's action, which
makes it even harder to incentivize the users to provide high-quality services.

An important class of incentive mechanisms for service exchange platforms are the rating mechanisms\footnote{Note that the rating mechanisms studied
in this paper focus on dealing with moral hazard problems, namely the server's quality of service is not perfectly observable. They are different
from the rating mechanisms dealing with adverse selection problems, namely the problems of identifying the users' types. See
\cite[Sec.~I]{Dellarocas} for detailed discussions on the above two classes of rating mechanisms.} \cite{YuMihaela_TEAC}--\cite{Ellison_SocialNorm},
in which each user is labeled with a rating based on its past behaviors in the system. A rating mechanism consists of a rating update rule and a
recommended strategy\footnote{Different terminologies have been used in the existing literature. For example,
\cite{Dellarocas}\cite{Kandori_SocialNorm} used ``reputation'' for ``rating'', and \cite{Kandori_SocialNorm} used ``social norm'' for ``recommended
strategy''.}. The recommended strategy specifies what is the desirable behavior under the current system state (e.g. the current rating profile of
the users or the current rating distribution). For example, the rating mechanism may recommend providing high-quality services for all the users when
the majority of users have high ratings, while recommending to provide high-quality services only to high-rating users when the majority have low
ratings. Then, based on each client's report on the quality of service, the rating mechanism revises each server's rating according to the rating
update rule. Generally speaking, the ratings of the users who comply with (resp. deviate from) the recommended behaviors go up (resp. down). Hence,
each user's rating summarizes its past behavior in the system. By keeping track of all the users' ratings and recommending them to reward (resp.
punish) the users with high (resp. low) ratings, the rating mechanism gives incentives to the users to obtain high ratings by rewarding them
indirectly, through recommending other users to provide them with high-quality services.

Existing rating mechanisms have been shown to work well when monitoring and reporting are perfect. However, when monitoring and reporting are subject
to errors, existing rating mechanisms cannot achieve the social optimum \cite{YuMihaela_TEAC}--\cite{Ellison_SocialNorm}. The errors, which are often
encountered in practice, may arise either from the client's own incapability of accurate assessment (for instance, the client, who wants to translate
some sentences into a foreign language, cannot accurately evaluate the server's translation), or from some system errors (for example, the client's
report on the server's service quality is missing due to network errors)\footnote{Note that the errors in this paper are not caused by the strategic
behaviors of the users. In other words, the clients report the service quality truthfully, and do not misreport intentionally to manipulate the
rating mechanism for their own interests. If the clients may report strategically, the mechanism can let the platform to assess the service quality
(still, with errors) to avoid strategic reporting.}. In the presence of errors, the server's rating may be wrongly updated. Hence, even if the users
follow the recommended desirable behavior, the platform may still fall into some ``bad'' states in which many users have low ratings due to erroneous
rating updates. In these bad states, the users with low ratings receive low-quality services, resulting in large performance loss compared to the
social optimum. This performance loss in the bad states is the major reason for the inefficiency of the existing rating mechanisms.

In this paper, we propose the first rating mechanisms that can achieve the social optimum even under imperfect monitoring. A key feature of the
proposed rating mechanism is the \emph{nonstationary} recommended strategy, which recommends different behaviors under the same system state,
depending on when this state occurs (for example, the rating mechanism may not always recommend punishing users with low ratings in the bad states).
Note, importantly, that the rating mechanism does not just randomize over different behaviors with a fixed probability in a state. Instead, it
recommends different behaviors in the current state based on the history of past states. We design the recommended strategy carefully, such that the
punishments happen frequently enough to provide sufficient incentives for the users, but not too frequently to reduce the performance loss incurred
in the bad states. The more patient the users are (i.e. the larger discount factor they have), the less frequent are the punishments. As a result,
the designed rating mechanism can asymptotically achieve the social optimum as the users become increasingly patient (i.e. as the discount factor
approaches $1$). This is in contrast with the existing rating mechanisms with \emph{stationary} recommended strategies, whose performance loss does
not vanish even as the users' patience increases. Another key feature of the proposed rating mechanism is the use of differential punishments that
punish users with different ratings differently. In Section~\ref{sec:Inefficiency}, we show that the absence of any one of these two features in our
mechanism will result in performance loss that does not vanish even when the users are arbitrarily patient.

We prove that the social optimum can be achieved by simple rating mechanisms, which assign \emph{binary} ratings to the users and recommend a small
set of \emph{three} recommended behaviors. We provide design guidelines of the rating update rules in socially-optimal rating mechanisms, and a
low-complexity online algorithm to construct the nonstationary recommended strategies. The algorithm essentially solves a system of two linear
equations with two variables in each period, and can be implemented with a memory of a fixed size (although by the definition of nonstationary
strategies, it appears that we may need a memory growing with time to store the history of past states), because we can appropriately summarize the
history of past states (by the solution to the above linear equations).

The rest of the paper is organized as follows. In Section~\ref{sec:related}, we discuss the differences between our work and related works. In
Section~\ref{sec:model}, we describe the model of service exchange systems with rating mechanisms. Then we design the optimal rating mechanisms in
Section~\ref{sec:OptimalMechanism}. Simulation results in Section~\ref{sec:Simulation} demonstrate the performance improvement of the proposed rating
mechanism. Finally, Section~\ref{sec:Conclusion} concludes the paper.

\section{Related Works}\label{sec:related}
\subsection{Related Works on Rating Protocols}
\begin{table}
\centering
\renewcommand{\arraystretch}{1.0}
\caption{Related Works on Rating Protocols.} \label{table:RelatedWork}
\begin{tabular}{|c|c|c|c|c|}
\hline
 & Rating update error & Recommended strategy & Discount factor & Performance loss \\
\hline
\cite{YuMihaela_TEAC}\cite{YuMihaela_P2P} & $\rightarrow0$ & Stationary & $<1$ & Yes \\
\hline
\cite{Blanc}\cite{Feldman_04_P2P} & $>0$ & Stationary & $<1$ & Yes \\
\hline
\cite{Dellarocas} & $>0$ & Stationary/Nonstationary & $<1$ & Yes \\
\hline
\cite{Kandori_SocialNorm}--\cite{Deb} & $=0$ & Stationary & $\rightarrow1$ & Yes \\
\hline
\cite{Ellison_SocialNorm} & $\rightarrow0$ & Stationary & $\rightarrow1$ & Yes \\
\hline
This work & $>0$ & Nonstationary & $<1$ & No \\
\hline
\end{tabular}
\end{table}

Rating mechanisms were originally proposed by \cite{Kandori_SocialNorm} for a large anonymous society, in which users are repeatedly randomly matched
to play the Prisoners' dilemma game. Assuming perfect monitoring, \cite{Kandori_SocialNorm} proposed a simple rating mechanism that can achieve the
social optimum: any user who has defected will be assigned with the lowest rating forever and will be punished by its future partners. Subsequent
research has been focusing on extending the results to more general games (see \cite{Postlewaite}\cite{DalBo_2007}\cite{Hasker_2007}\cite{Deb}), or
on discovering alternative mechanisms (for example, \cite{Takahashi_2010} showed that cooperation can be sustained if each user can observe its
partner's past actions). However, all these works assumed perfect monitoring and were aimed at dealing with the incentive problems caused by the
anonymity of users and the lack of fixed partnership; they did not study the impact of imperfect monitoring. Under imperfect observation/reporting,
the system will collapse under their rating mechanisms because all the users will eventually end up with having low ratings forever due to errors.

Some works \cite{YuMihaela_TEAC}\cite{YuMihaela_P2P}\cite{Ellison_SocialNorm} assumed imperfect monitoring, but focused on the limit case when the
monitoring tends to be perfect. The conclusion of these works is that the social optimum can be achieved in the limit case when the monitoring
becomes ``almost perfect'' (i.e., when the rating update error goes to zero).

Only a few works \cite{Blanc}--\cite{Dellarocas} analyzed rating mechanisms under imperfect monitoring with \emph{fixed nonzero} monitoring errors.
For a variety of rating mechanisms studied in \cite{Blanc}--\cite{Dellarocas}, the performance loss with respect to the social optimum is quantified
in terms of the rating update error. These results confirm that existing rating mechanisms suffer from (severe) performance loss under rating update
errors. Note that the model in \cite{Dellarocas} is fundamentally different than ours. In \cite{Dellarocas}, there is only a single long-lived seller
(server), while all the buyers (clients) are short-lived. Under this model, it is shown in \cite{Dellarocas} that the rating mechanism is bounded
away from social optimum even when nonstationary strategies are used. In contrast, we show that under our model with long-lived servers and clients,
we can achieve the social optimum by nonstationary strategies with differential punishments. In the following, we discuss the intuitions of how to
achieve the social optimum under our model.

There are two sources of inefficiency. One source of inefficiency comes from the stationary recommended strategies, which recommends the same
behavior under the same state \cite{YuMihaela_TEAC}--\cite{Feldman_04_P2P}\cite{Kandori_SocialNorm}--\cite{Ellison_SocialNorm}. As we have discussed
earlier, the inefficiency of the existing rating mechanisms comes from the punishments triggered in the ``bad'' states. Specifically, to give
incentives for the users to provide high-quality services, the rating mechanism \emph{must} punish the low-rating users under certain rating
distributions (i.e. under certain ``bad'' states). When the users are punished (i.e. they are provided with low-quality services), the average
payoffs in these states are far below the social optimum. In the presence of rating update errors, the bad states happen with a probability bounded
above zero (the lower bound depends only on the rating update error). As a result, the low payoffs occur with a frequency bounded above zero, which
incurs an efficiency loss that cannot vanish unless the rating update error goes to zero.

Another source of inefficiency is the lack of differential punishments. As will be proved in Section~\ref{sec:Inefficiency}, the rating mechanisms
with no differential punishment have performance loss even when nonstationary recommended strategies are used.

This paper is the first to propose a class of rating mechanisms that achieve the social optimum even when update errors do not tend to zero. Our
mechanisms rely on (explicitly-constructed) \emph{nonstationary} strategies with differential punishments. The key intuitions of why the proposed
mechanism achieves social optimum are as follows. First, nonstationary strategies punish the users in the bad states only when necessary, depending
on the history of past states. In this way, nonstationary strategies can lower the frequency of punishment in the bad states to a level just enough
to provide sufficient incentives for the users to provide high-quality services. In addition, differential punishment further reduces the loss in
social welfare by transferring payoffs from low-rating users to high-rating users, instead of lowering everyone's payoff with non-differential
punishment.

In Table~\ref{table:RelatedWork}, we compare the proposed work with existing rating mechanisms.

\subsection{Related Works in Game Theory Literature}
Our results are related to folk theorem results for repeated games \cite{FLM1994} and stochastic games \cite{HornerSugayaTakahashi}. However, these
existing folk theorem results \cite{FLM1994}\cite{HornerSugayaTakahashi} cannot be directly applied to our model. First, the results in
\cite{FLM1994} are derived for repeated games, in which every stage game is the same. Our system is modeled as a stochastic game, in which the stage
games may be different because of the rating distributions.

Second, there do exist folk theorems for stochastic games \cite{HornerSugayaTakahashi}, but they also do not apply to our model. The folk theorems
\cite{HornerSugayaTakahashi} apply to standard stochastic games, in which the state must satisfy the following properties: 1) the state, together
with the plan profile, uniquely determines the stage-game payoff, and 2) the state is known to all the users. In our model, since each user's stage
game payoff depends on its own rating, each user's rating must be included in the state and be known to all the users. In other words, if we model
the system as a standard stochastic game in order to apply the folk theorems, we need to define the state as the rating profile of all the users (not
just the rating distribution). Then, the folk theorem states that the social optimum can be asymptotically achieved by strategies that depend on the
history of rating profiles. However, in our model, the players do not know the full rating profile, but only know the rating distribution. Hence, the
strategy can use only the information of rating distributions.\footnote{We insist on restricting to strategies that depend only on the history of
rating distributions because in practice, 1) the platform may not publish the full rating profile due to informational and privacy constraints, and
2) even if the platform does publish such information, it is impractical to assume that the users can keep track of it.} Whether such strategies can
achieve the social optimum is not known according to the folk theorems; we need to prove the existence of socially optimal strategies that use only
the information of rating distributions.

In addition, our results are fundamentally different from the folk theorem results \cite{FLM1994}\cite{HornerSugayaTakahashi} in nature. First,
\cite{FLM1994}\cite{HornerSugayaTakahashi} focus on the limit case when the discount factor goes to one, which is not realistic because the users are
not sufficiently patient. More importantly, the results in \cite{FLM1994}\cite{HornerSugayaTakahashi} are not constructive. They focus on \emph{what}
payoff profiles are achievable, but cannot show \emph{how} to achieve those payoff profiles. They do not determine a lower bound on discount factors
that admit equilibrium strategy profiles yielding the target payoff profile, and hence cannot construct equilibrium strategy profiles. By contrast,
we do determine a lower bound on discount factors that admit equilibrium strategy profiles yielding the target payoff profile, and do construct
equilibrium strategy profiles.

\subsection{Related Mathematical Frameworks}

Rating mechanisms with stationary recommended strategies can be designed by extending Markov decision processes (MDPs) in two important and
non-trivial ways \cite{YuMihaela_TEAC}\cite{YuMihaela_P2P}\cite{Izhak-RatzinParkvanderSchaar}\cite{ParkvanderSchaar_TSP}: 1) since there are multiple
users, the value of each state is a vector of all the users' values, instead of a scalar in standard MDPs, and 2) the incentive compatibility
constraints of self-interested users need to be fulfilled (e.g., the values of ``good'' states, in which most users have high ratings, should be
sufficiently larger than those of ``bad'' states, such that users are incentivized to obtain high ratings), while standard MDPs do not impose such
constraints.

In this paper, we make a significant step forward with respect to the state-of-the-art rating mechanisms with stationary strategies: we design rating
mechanisms where the recommended strategies can be nonstationary. The proposed design leads to significant performance improvements, but is also
significantly more challenging from a theoretical perspective. The key challenge is that nonstationary strategies may choose different actions under
the same state, resulting in possibly different current payoffs in the same state. Hence, the value function under nonstationary strategies are
\emph{set-valued}, which significantly complicates the analysis, compared to \emph{single-valued} value functions under stationary
strategies\footnote{In randomized stationary strategies, although different actions may be taken in the same state \emph{after randomization}, the
probability of actions chosen is fixed. In the Bellman equation, we need to use the \emph{expected payoff before randomization} which is fixed in the
same state, instead of the realized payoffs after randomization. Hence, the value function is still single-valued.}.
\begin{table}
\centering
\renewcommand{\arraystretch}{1.0}
\caption{Related Mathematical Frameworks.} \label{table:RelatedWork_MathematicalFramework}
\begin{tabular}{|c|c|c|c|c|}
\hline
 & Standard MDP & Extended MDP \cite{YuMihaela_TEAC}\cite{YuMihaela_P2P}\cite{Izhak-RatzinParkvanderSchaar}\cite{ParkvanderSchaar_TSP} & Self-generating sets \cite{APS1990}--\cite{HornerSugayaTakahashi} & This work \\
\hline
\# of users & Single & Multiple & Multiple & Multiple \\
\hline
Value function & Single-valued & Single-valued & Set-valued & Set-valued \\
\hline
Incentive constraints & No & Yes & Yes & Yes \\
\hline
Strategies & Stationary & Stationary & Nonstationary & Nonstationary \\
\hline
Discount factor & $<1$ & $<1$ & $\rightarrow1$ & $<1$ \\
\hline
Constructive & Yes & Yes & No & Yes \\
\hline
\end{tabular}
\end{table}

The mathematical framework of analyzing nonstationary strategies with set-valued value functions was proposed as a theory of self-generating sets in
\cite{APS1990}. It was widely used in game theory to prove folk theorems in repeated games \cite{FLM1994} and stochastic games
\cite{HornerSugayaTakahashi}. We have discussed our differences from the folk theorem results \cite{FLM1994}\cite{HornerSugayaTakahashi} in the
previous subsection.

In Table~\ref{table:RelatedWork_MathematicalFramework}, we compare our work with existing mathematical frameworks.

\section{System Model and Problem Formulation}\label{sec:model}

\begin{figure}
\centering
\includegraphics[width =5.0in]{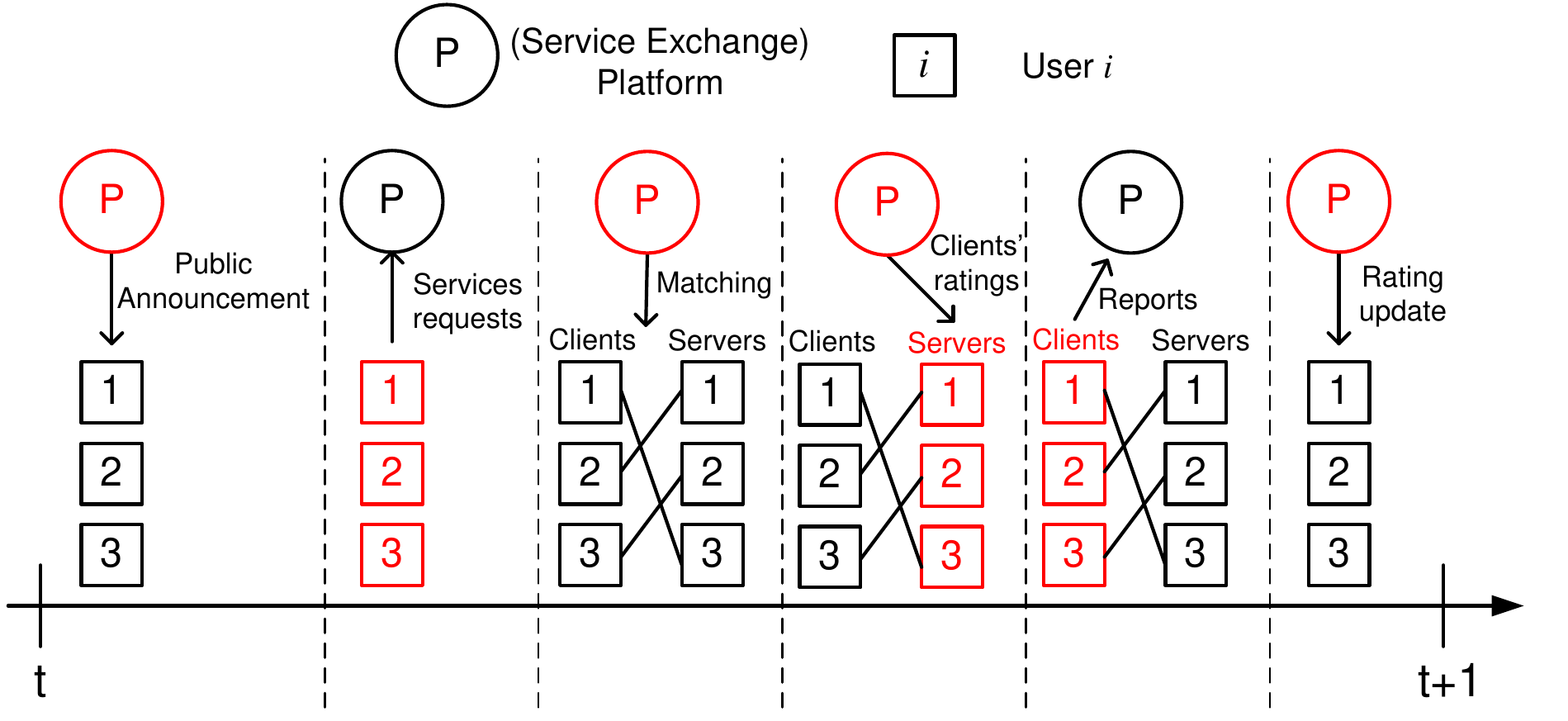}
\caption{Illustration of the rating mechanism in one period.} \label{fig:RatingProtocol}
\end{figure}

\subsection{System Model}
\subsubsection{The Rating Mechanism}
We consider a service exchange platform with a set of $N$ users, denoted by $\mathcal{N}=\{1,\ldots,N\}$. Each user can provide some services (e.g.
data in P2P networks, labor in Amazon Mechanic Turk) valuable to the other users. The rating mechanism assigns each user $i$ a binary label
$\theta_i\in\Theta\triangleq\{0,1\}$, and keep record of the \emph{rating profile} $\bm{\theta}=(\theta_1,\ldots,\theta_N)$. Since the users usually
stay in the platform for a long period of time, we divide time into periods indexed by $t=0,1,2,\ldots$. In each period, the rating mechanism
operates as illustrated in Fig.~\ref{fig:RatingProtocol}, which can be roughly described as follows:
\begin{itemize}
\item Each user requests services as a \emph{client}.
\item Each user, as a \emph{server}, is matched to another user (its client) based on a matching rule.
\item Each server chooses to provide high-quality or low-quality services.
\item Each client reports its assessment of the service quality to the rating mechanism, who will update the server's rating based on the report.
\end{itemize}
Next, we describe the key components in the rating mechanism in details.

\emph{Public announcement:} At the beginning of each period, the platform makes public announcement to the users. The public announcement includes
the rating distribution and the recommended plan in this period. The \emph{rating distribution} indicates how many users have rating 1 and rating 0,
respectively. Denote the rating distribution by $\bm{s}(\bm{\theta})=(s_0(\bm{\theta}),s_1(\bm{\theta}))$, where
$s_1(\bm{\theta})=\sum_{i\in\mathcal{N}} \theta_i$ is the number of users with rating $1$, and $s_0(\bm{\theta})=N-s_1(\bm{\theta})$ is the number of
users with rating $0$. Denote the set of all possible rating distributions by $S$. Note that the platform does not disclose the rating profile
$\bm{\theta}$ for privacy concerns. The platform also recommends a desired behavior in this period, called \emph{recommended plan}. The recommended
plan is a contingent plan of which service quality the server should choose based on its own rating and its client's rating. Formally, the
recommended plan, denoted by $\alpha_0$, is a mapping $\alpha_0:\Theta\times\Theta\rightarrow \{0,1\}$, where $0$ and $1$ represent ``low-quality
service'' and ``high-quality service'', respectively. Then $\alpha_0(\theta_c,\theta_s)$ denotes the recommended service quality for a server with
rating $\theta_s$ when it is matched to a client with rating $\theta_c$. We write the set of recommended plans as
$\mathcal{A}=\{\alpha|\alpha:\Theta\times\Theta\rightarrow \{0,1\}\}$. We are particularly interested in the following three plans. The plan
$\alpha^{\rm a}$ is the \emph{altruistic} plan:
\begin{eqnarray}
\alpha^{\rm a}(\theta_c,\theta_s)=1, \forall \theta_c,\theta_s\in\{0,1\},
\end{eqnarray}
where the server provides high-quality service regardless of its own and its client's ratings. The plan $\alpha^{\rm f}$ is the \emph{fair} plan:
\begin{eqnarray}
\alpha^{\rm f}(\theta_c,\theta_s)=\left\{\begin{array}{ll} 0 & \theta_s>\theta_c \\ 1 & \theta_s\leq\theta_c \end{array}\right.,
\end{eqnarray}
where the server provides high-quality service when its client has higher or equal ratings. The plan $\alpha^{\rm s}$ is the \emph{selfish} plan:
\begin{eqnarray}
\alpha^{\rm s}(\theta_c,\theta_s)=0, \forall \theta_c,\theta_s\in\{0,1\},
\end{eqnarray}
where the server provides low-quality service regardless of its own and its client's ratings. Note that we can consider the selfish plan as a
non-differential punishment in which everyone receives low-quality services, and consider the fair plan as a differential punishment in which users
with different ratings receive different services.

\emph{Service requests:} The platform receives service requests from the users. We assume that there is no cost in requesting services, and that each
user always have demands for services. Hence, all the users will request services.

\emph{Matching:} The platform matches each user $i$, as a client, to another user $m(i)$ who will serve $i$, where $m$ is a matching
$m:\mathcal{N}\rightarrow\mathcal{N}$. Since the platform cannot match a user to itself, we write the set of all possible matchings as
$M=\left\{m:m~{\rm bijective},~m(i)\neq i,\forall i\in\mathcal{N}\right\}$. The mechanism defines a random matching rule, which is a probability
distribution $\mu$ on the set of all possible matchings $M$. In this paper, we focus on the uniformly random matching rule, which chooses every
possible matching $m\in M$ with the same probability. The analysis can be easily generalized to the cases with non-uniformly random matching rules,
as long as the matching rules do not distinguish users with the same rating.

\emph{Clients' ratings:} The platform will inform each server of its client's rating, such that each server can choose its service quality based on
its own and its client's ratings.

\emph{Reports:} After the servers serve their clients, the platform elicits reports from the clients about their service quality. However, the report
is \emph{inaccurate}, either by the client's incapability of accurate assessment (for instance, the client, who wants to translate some sentences
into a foreign language, cannot accurately evaluate the server's translation) or by some system error (for example, the data/file sent by the server
is missing due to network errors). We characterize the erroneous report by a mapping $R:\{0,1\}\rightarrow\Delta(\{0,1\})$, where $\Delta(\{0,1\})$
is the probability distribution over $\{0,1\}$. For example, $R(1|q)$ is the probability that the client reports ``high quality'' given the server's
actual service quality $q$. In this paper, we focus on reports of the following form
\begin{eqnarray}\label{eqn:Report}
R(r|q)=\left\{\begin{array}{ll}1-\varepsilon, & r=q \\ \varepsilon, & r\neq q \end{array} \right.,~\forall r,q\in\{0,1\},
\end{eqnarray}
where $\varepsilon\in[0,0.5)$ is the report error probability.\footnote{We confine the report error probability $\varepsilon$ to be smaller than
$0.5$. If the error probability $\varepsilon$ is $0.5$, the report contains no useful information about the service quality. If the error probability
is larger than $0.5$, the rating mechanism can use the opposite of the report as an indication of the service quality, which is equivalent to the
case with the error probability smaller than $0.5$.} Note, however, that reporting is not strategic: the client reports truthfully, but with errors.
If the clients report strategically, the mechanism can let the platform to assess the service quality (still, with errors) to avoid strategic
reporting. For simplicity, we assume that the report error is symmetric, in the sense that reporting high and low qualities have the same error
probability. Extension to asymmetric report errors is straightforward.

\emph{Rating update:} Given the clients' reports, the platform updates the servers' ratings according to the \emph{rating update rule}, which is
defined as a mapping $\tau:\Theta\times\Theta\times \{0,1\} \times \mathcal{A}\rightarrow\Delta(\Theta)$. For example,
$\tau(\theta_s^\prime|\theta_c,\theta_s,r,\alpha_0)$ is the probability of the server's updated rating being $\theta_s^\prime$, given the client's
rating $\theta_c$, the server's own rating $\theta_s$, the client's report $r$, and the recommended plan $\alpha_0$. We focus on the following class
of rating update rules (see Fig.~\ref{fig:Automata_RatingUpdateRule} for illustration):
\begin{eqnarray}\label{eqn:ReputationUpdateRule}
\tau(\theta_s^\prime|\theta_c,\theta_s,r,\alpha_0)
= \left\{\begin{array}{ll} \beta_{\theta_s}^+, & \theta_s^\prime=1, r\geq\alpha_0(\theta_c,\theta_s) \\
                             1-\beta_{\theta_s}^+, & \theta_s^\prime=0, r\geq\alpha_0(\theta_c,\theta_s) \\
                             1-\beta_{\theta_s}^-, & \theta_s^\prime=1, r<\alpha_0(\theta_c,\theta_s) \\
                             \beta_{\theta_s}^-, & \theta_s^\prime=0, r<\alpha_0(\theta_c,\theta_s) \\ \end{array}\right..
\nonumber
\end{eqnarray}
In the above rating update rule, if the reported service quality is not lower than the recommended service quality, a server with rating $\theta_s$
will have rating $1$ with probability $\beta_{\theta_s}^+$; otherwise, it will have rating $0$ with probability $\beta_{\theta_s}^-$. Other more
elaborate rating update rules may be considered. But we show that this simple one is good enough to achieve the social optimum.

\begin{figure}
\centering
\includegraphics[width =3.0in]{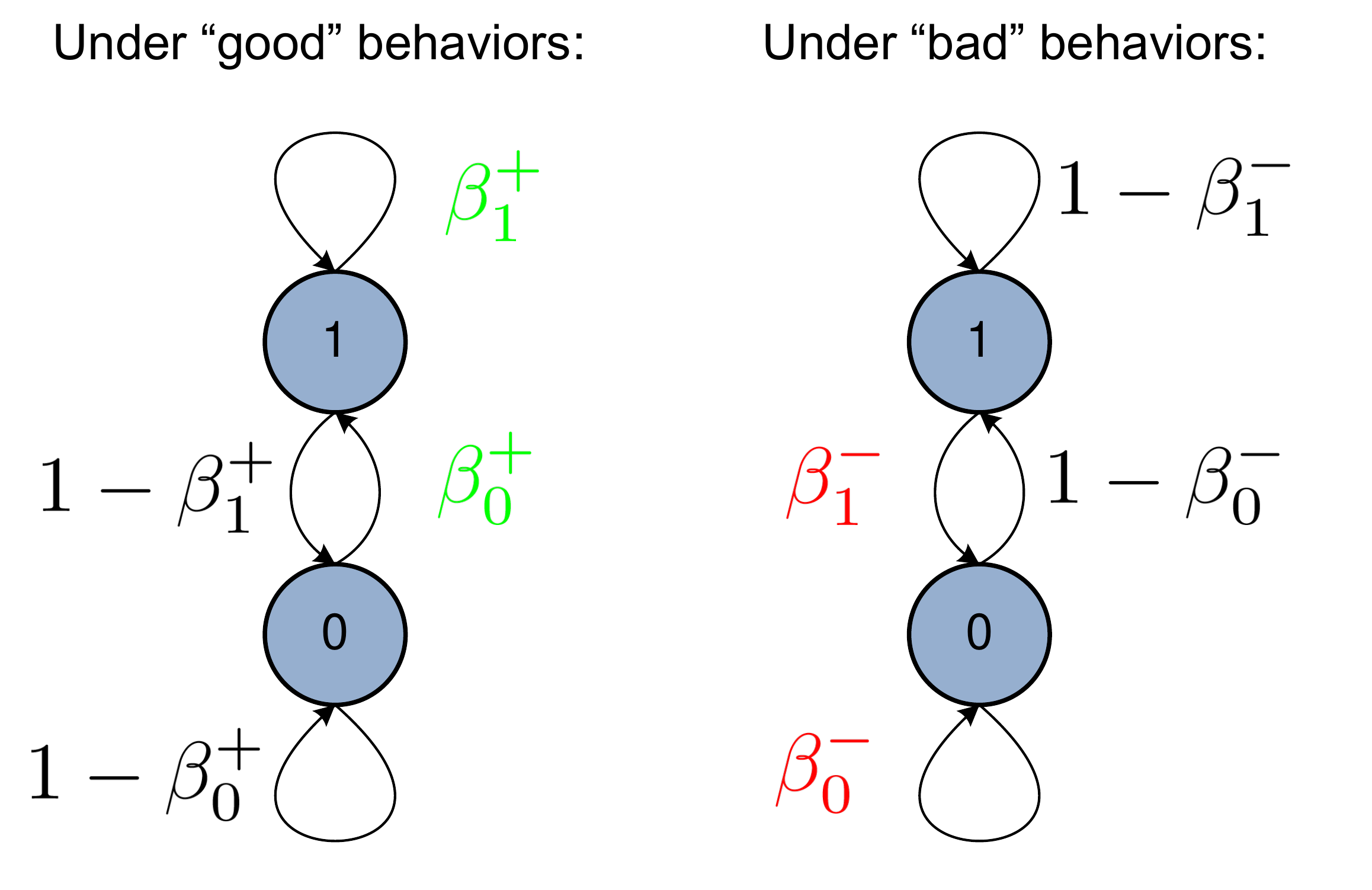}
\caption{Illustration of the rating update rule. The circle denotes the rating, and the arrow denotes the rating update with corresponding
probabilities.} \label{fig:Automata_RatingUpdateRule}
\end{figure}

\emph{Recommended strategy:} The final key component of the rating mechanism is the recommended \emph{strategy}, which determines what recommended
plan should be announced in each period. In each period $t$, the mechanism keeps track of the history of rating distributions, denoted by
$\bm{h}^t=(\bm{s}^0,\ldots,\bm{s}^t)\in S^{t+1}$, and chooses the recommended plan based on $\bm{h}^t$. In other words, the recommended strategy is a
mapping from the set of histories to its plan set, denoted by $\pi_0:\cup_{t=0}^\infty S^{t+1} \rightarrow \mathcal{A}$. Denote the set of all
recommended strategies by $\Pi$. Note that although the rating mechanism knows the rating profile, it determines the recommended plan based on the
history of rating distributions, because 1) this reduces the computational and memory complexity of the protocol, and 2) it is easy for the users to
follow since they do not know the rating profile. Moreover, since the plan set $\mathcal{A}$ has $16$ elements, the complexity of choosing the plan
is large. Hence, we consider the strategies that choose plans from a subset $\mathcal{B}\subseteq \mathcal{A}$, and define $\Pi(\mathcal{B})$ as the
set of strategies restricted on the subset $\mathcal{B}$ of plans.

In summary, the rating mechanism can be represented by the design parameters: the rating update rule and the recommended strategy, and can be denoted
by the tuple $(\tau,\pi_0)$.

\subsubsection{Payoffs} Once a server and a client are matched, they play the gift-giving game in Table~\ref{table:GiftGivingGame}
\cite{YuMihaela_TEAC}--\cite{Kandori_SocialNorm}\cite{Takahashi_2010}\cite{Ellison_SocialNorm}, where the row player is the client and the column
player is the server. We normalize the payoffs received by the client and the server when a server provides low-quality services to $0$. When a
server provides high-quality services, the client gets a benefit of $b>0$ and the worker incurs a cost of $c\in(0,b)$. In the unique Nash equilibrium
of the gift-giving game, the server will provide low-quality services, which results in a zero payoff for both the client and the server. Note that
as in \cite{YuMihaela_TEAC}--\cite{Kandori_SocialNorm}\cite{Takahashi_2010}\cite{Ellison_SocialNorm}, we assume that the same gift-giving game is
played for all the client-server pairs. This assumption is reasonable when the number of users is large. Since $b$ can be considered as a user's
expected benefit across different servers, and $c$ as its expected cost of high-quality service across different clients, the users' expected
benefits/costs should be approximately the same when the number of users is large. This assumption is also valid when the users have different
realized benefit and cost in each period but the same expected benefit $b$ and expected cost $c$ across different periods.

\begin{table}
\renewcommand{\arraystretch}{1.3}\centering
\caption{Gift-Giving Game Between A Client and A Server.}{
\begin{tabular}{r|c|c|}
\multicolumn{1}{r}{}
 &  \multicolumn{1}{c}{high-quality}
 & \multicolumn{1}{c}{low-quality} \\
\cline{2-3}
request & $(b, -c)$ & $(0, 0)$ \\
\cline{2-3}
\end{tabular}}
\label{table:GiftGivingGame}
\end{table}

\emph{Expected payoff in one period:} Based on the gift-giving game, we can calculate each user's expected payoff obtained in one period. A user's
expected payoff in one period depends on its own rating, the rating distribution, and the users' plans. We write user $i$'s plan as
$\alpha_i\in\mathcal{A}$, and the plan profile of all the users as $\bm{\alpha}=(\alpha_1,\ldots,\alpha_N)$. Then user $i$'s expected payoff in one
period is $u_i(\theta_i,\bm{s},\bm{\alpha})$. For illustration, we calculate the users' expected payoffs under several important scenarios, assuming
that all the users follow the recommended plan (i.e. $\alpha_i=\alpha_0,~\forall i\in\mathcal{N}$). When the altruistic plan $\alpha^{\rm a}$ is
recommended, all the users receive the same expected payoff in one period as
$$
u_i(\theta_i,\bm{s},\alpha^{\rm a}\cdot\bm{1}_N)=b-c,~\forall i,\theta_i,\bm{s},
$$
where $\alpha\cdot\bm{1}_N$ is the plan profile in which every user chooses plan $\alpha$. Similarly, when the selfish plan
$\alpha^{\rm s}$ is recommended, all the users receive zero expected payoff in one period, namely
$$
u_i(\theta_i,\bm{s},\alpha^{\rm s}\cdot\bm{1}_N)=0,~\forall i,\theta_i,\bm{s}.
$$
When the fair plan $\alpha^{\rm f}$ is recommended, the users receive expected payoffs in one period
as follows
\begin{eqnarray}
u_i(\theta_i,\bm{s},\alpha^{\rm f}\cdot\bm{1}_N) = \left\{\begin{array}{ll} \frac{s_0-1}{N-1}\cdot b-c, & \theta_i=0 \\
 b-\frac{s_1-1}{N-1}\cdot c, & \theta_i=1\end{array}\right..
\end{eqnarray}
Under the fair plan, the users with rating 1 receive a payoff higher than $b-c$, because they get high-quality services from everyone and provide
high-quality services only when matched to clients with rating 1. In contrast, the users with rating 0 receive a payoff lower than $b-c$. Hence, the
fair plan $\alpha^{\rm f}$ can be considered as a differential punishment.

\emph{Discounted average payoff:} Each user $i$ has its own strategy $\pi_i\in\Pi$. Write the joint strategy profile of all the users as
$\bm{\pi}=(\pi_1,\ldots,\pi_N)$. Then given the initial rating profile $\bm{\theta}^0$, the recommended strategy $\pi_0$ and the joint strategy
profile $\bm{\pi}$ induce a probability distribution over the sequence of rating profiles $\bm{\theta}^1,\bm{\theta}^2,\ldots$. Taking the
expectation with respect to this probability distribution, each user $i$ receives a discounted average payoff $U_i(\bm{\theta}^0,\pi_0,\bm{\pi})$
calculated as
\begin{eqnarray}\label{eqn:RepeatedGamePayoff}
U_i(\bm{\theta}^0,\pi_0,\bm{\pi})=\mathbb{E}_{\bm{\theta}^1,\bm{\theta}^2,\ldots} \left\{(1-\delta)\sum_{t=0}^\infty \delta^t
u_i\left(\theta_i^t,\bm{s}(\bm{\theta}^t),\bm{\pi}(\bm{s}(\bm{\theta}^0),\ldots,\bm{s}(\bm{\theta}^t)\right)\right\} \nonumber
\end{eqnarray}
where $\delta\in[0,1)$ is the common discount factor of all the users. The discount factor $\delta$ is the rate at which the users discount future
payoffs, and reflects the patience of the users. A more patient user has a larger discount factor. Note that the recommended strategy $\pi_0$ does
affect the users' discounted average payoffs by affecting the evolution of the rating profile (i.e. by affecting the expectation operator
$\mathbb{E}_{\bm{\theta}^1,\bm{\theta}^2,\ldots}$).

\subsubsection{Definition of The Equilibrium} The platform adopts \emph{sustainable rating mechanisms}, which specifies a tuple of
rating update rule and recommended strategy $(\tau,\pi_0)$, such that the users find it in their self-interests to follow the recommended strategy.
In other words, the recommended strategy should be an equilibrium.

Note that the interaction among the users is neither a repeated game \cite{FLM1994} nor a standard stochastic game \cite{HornerSugayaTakahashi}. In a
repeated game, every stage game is the same, which is clearly not true in the platform because users' stage-game payoff
$u_i(\theta_i,\bm{s},\bm{\alpha})$ depends on the rating distribution $\bm{s}$. In a standard stochastic game, the state must satisfy: 1) the state
and the plan profile uniquely determines the stage-game payoff, and 2) the state is known to all the users. In the platform, the user's stage-game
payoff $u_i(\theta_i,\bm{s},\bm{\alpha})$ depends on its own rating $\theta_i$, which should be included in the state and be known to all the users.
Hence, if we were to model the interaction as a standard stochastic game, we need to define the state as the rating profile $\bm{\theta}$. However,
the rating profile is not known to the users in our formulation.

To reflect our restriction on recommended strategies that depend only on rating distributions, we define the equilibrium as public announcement
equilibrium (PAE), since the strategy depends on the publicly announced rating distributions. Before we define PAE, we need to define the
continuation strategy $\pi|_{\bm{h}^t}$, which is a mapping $\pi|_{\bm{h}^t}:\cup_{k=0}^\infty \mathcal{H}^k\rightarrow A$ with
$\pi|_{\bm{h}^t}(\bm{h}^k)=\pi(\bm{h}^t\bm{h}^k)$, where $\bm{h}^t\bm{h}^k$ is the concatenation of $\bm{h}^t$ and $\bm{h}^k$.

\begin{definition}\label{def:PPE}
A pair of a recommended strategy and a symmetric strategy profile $(\pi_0,\pi_0\cdot\bm{1}_N)$ is a PAE, if for all $t\geq0$, for all
$\bm{\tilde{h}}^t\in\mathcal{H}^t$, and for all $i\in\mathcal{N}$, we have
\begin{eqnarray}
U_i(\bm{\tilde{\theta}}^t,\pi_0|_{\bm{\tilde{h}}^t},\pi_0|_{\bm{\tilde{h}}^t}\cdot\bm{1}_N) \geq
U_i(\bm{\tilde{\theta}}^t,\pi_0|_{\bm{\tilde{h}}^t},(\pi_i|_{\bm{\tilde{h}}^t},\pi_0|_{\bm{\tilde{h}}^t}\cdot\bm{1}_{N-1})),~\forall
\pi_i|_{\bm{\tilde{h}}^t}\in\Pi, \nonumber
\end{eqnarray}
where $(\pi_i|_{\bm{\tilde{h}}^t},\pi_0|_{\bm{\tilde{h}}^t}\cdot\bm{1}_{N-1})$ is the continuation strategy profile in which user $i$ deviates to
$\pi_i|_{\bm{\tilde{h}}^t}$ and the other users follow the strategy $\pi_0|_{\bm{\tilde{h}}^t}$.
\end{definition}

Note that in the definition, we allow a user to deviate to any strategy $\pi_i\in\Pi$, even if the recommended strategy $\pi_0$ is restricted to a
subset $\mathcal{B}$ of plans. Hence, the rating mechanism is robust, in the sense that a user cannot gain even when it uses more complicated
strategies. Note also that although a rating mechanism can choose the initial rating profile $\bm{\theta}^0$, we require a recommended strategy to
fulfill the incentive constraints under all the initial rating profiles. This adds to the flexibility in choosing the initial rating profile.

PAE is stronger than the Nash equilibrium (NE), because PAE requires the users to not deviate following \emph{any} history, while NE requires the
users to not deviate following the histories that happen in the equilibrium. In this sense, PAE can be considered as a special case of public perfect
equilibrium (PPE) in standard repeated and stochastic games, where the strategies depend only on the rating distribution.

\subsection{The Rating Protocol Design Problem}
The goal of the rating mechanism designer is to choose a rating mechanism $(\tau,\pi_0)$, such that the social welfare at the equilibrium is
maximized in the worst case (with respect to different initial rating profiles). Maximizing the worst-case performance gives us a much stronger
performance guarantee than maximizing the performance under a given initial rating profile. Given the rating update error $\varepsilon$, the discount
factor $\delta$, and the subset $\mathcal{B}$ of plans, the rating mechanism design problem is formulated as:
\begin{eqnarray}\label{eqn:DesignProblem}
W(\varepsilon,\delta,\mathcal{B})=&\displaystyle\max_{\tau,\pi_0\in\Pi(\mathcal{B})}& \min_{\bm{\theta}^0\in \Theta^N} \frac{1}{N} \sum_{i\in\mathcal{N}} U_i(\bm{\theta}^0,\pi_0,\pi_0\cdot\bm{1}_{N}) \nonumber\\
& s.t. & (\pi_0,\pi_0\cdot\bm{1}_N)~\mathrm{is~a~PAE}.
\end{eqnarray}

Note that $W(\varepsilon,\delta,\mathcal{B})$ is strictly smaller than the social optimum $b-c$ for any $\varepsilon$, $\delta$, and $\mathcal{B}$.
This is because to exactly achieve $b-c$, the protocol must recommend the altruistic plan $\alpha^{\rm a}$ all the time (even when someone shirks),
which cannot be an equilibrium. However, we can design rating mechanisms such that for any fixed $\varepsilon\in[0,0.5)$,
$W(\varepsilon,\delta,\mathcal{B})$ can be arbitrarily close to the social optimum. In particular, such rating mechanisms can be simple, in that
$\mathcal{B}$ can be a small subset of three plans (i.e. $\mathcal{B}=A^{\rm afs}\triangleq\{\alpha^{\rm a},\alpha^{\rm f},\alpha^{\rm s}\}$).

\section{Sources of Inefficiency}\label{sec:Inefficiency}
To illustrate the importance of designing optimal, yet simple rating schemes, as well as the challenges associated with determining such a design, in
this section, we discuss several simple rating mechanisms that appear to work well intuitively, and show that they are actually bounded away from the
social optimum even when the users are arbitrarily patient. We will illustrate why they are inefficient, which gives us some insights on how to
design socially-optimal rating mechanisms.

\subsection{Stationary Recommended Strategies}
\subsubsection{Analysis} We first consider rating mechanisms with stationary recommended strategies, which determine the recommended plan solely based on the
current rating distribution. Since the game is infinitely-repeatedly played, given the same rating distribution, the continuation game is the same
regardless of when the rating distribution occurs. Hence, similar to MDP, we can assign value functions $V_\theta^{\pi_0}: S\rightarrow
\mathbb{R},~\forall \theta$ for a stationary strategy $\pi_0$, with $V_\theta^{\pi_0}(\bm{s})$ being the continuation payoff of a user with rating
$\theta$ at the rating distribution $\bm{s}$. Then, we have the following set of equalities that the value function needs to satisfy:
\begin{eqnarray}\label{eqn:ValueIteration}
V_{\theta_i}^{\pi_0}(\bm{s}) &=& (1-\delta)\cdot u_i(\pi_0(\bm{s}),\pi_0(\bm{s})\cdot \bm{1}_N) \\
& + &\delta\cdot\sum_{\theta_i^\prime,\bm{s}^\prime}\Pr(\theta_i^\prime,\bm{s}^\prime|\theta_i,\bm{s},\pi_0(\bm{s}),\pi_0(\bm{s})\cdot\bm{1}_N)\cdot
V_{\theta_i^\prime}^{\pi_0}(\bm{s}^\prime), \forall i\in\mathcal{N}, \nonumber
\end{eqnarray}
where $\Pr(\theta_i^\prime,\bm{s}^\prime|\theta_i,\bm{s},\pi_0(\bm{s}),\pi_0(\bm{s})\cdot\bm{1}_N)$ is the transition probability. We can solve for
the value function from the above set of equalities, which are similar to the Bellman equations in MDP. However, note that obtaining the value
function is not the final step. We also need to check the incentive compatibility (IC) constraints. For example, to prevent user $i$ from deviating
to plan $\alpha^\prime$, the following inequality has to be satisfied:
\begin{eqnarray}\label{eqn:Stationary_IC}
V_{\theta_i}^{\pi_0}(\bm{s}) &\geq& (1-\delta)\cdot u_i(\pi_0(\bm{s}),(\alpha^\prime, \pi_0(\bm{s})\cdot \bm{1}_{N-1})) \\
& + &\delta\cdot\sum_{\theta_i^\prime,\bm{s}^\prime}\Pr(\theta_i^\prime,\bm{s}^\prime|\theta_i,\bm{s},\pi_0(\bm{s}),(\alpha^\prime,
\pi_0(\bm{s})\cdot \bm{1}_{N-1}))\cdot V_{\theta_i^\prime}^{\pi_0}(\bm{s}^\prime), \forall i\in\mathcal{N}. \nonumber
\end{eqnarray}

Given a rating mechanism with a stationary recommended strategy $\pi_0$, if its value function satisfies all the IC constraints, we can determine the
social welfare of the rating mechanism. For example, suppose that all the users have an initial rating of 1. Then, all of them achieve the expected
payoff $V_1^{\pi_0}(0,N)$, which is the social welfare achieved under this rating mechanism.

Note that given a recommended strategy $\pi_0$, it is not difficult to compute the value function by solving the set of linear equations in
\eqref{eqn:ValueIteration} and check the IC constraints according to the set of linear inequalities in \eqref{eqn:Stationary_IC}. However, it is
difficult to derive structural results on the value function (e.g. whether the state with more rating-1 users has a higher value), and thus difficult
to know the structures of the optimal recommended strategy (e.g. whether the optimal recommended strategy is a threshold strategy). The difficulty
mainly comes from the complexity of the transition probabilities
$\Pr(\theta_i^\prime,\bm{s}^\prime|\theta_i,\bm{s},\pi_0(\bm{s}),\pi_0(\bm{s})\cdot\bm{1}_N)$. For example, assuming $\pi_0(\bm{s})=\alpha^{\rm a}$,
we have
\begin{eqnarray}
\begin{array}{l} \Pr(1,\bm{s}^\prime|1,\alpha^{\rm a},\alpha^{\rm a}\cdot\bm{1}_N) = x_1^+ \cdot \\
\sum_{k=\max\{0,s_1^\prime-1-(N-s_1)\}}^{\min\{s_1-1,s_1^\prime-1\}} {s_1-1 \choose k} (x_1^+)^k (1-x_1^+)^{s_1-1-k} {N-s_1 \choose s_1^\prime-1-k}
(x_0^+)^{s_1^\prime-1-k}(1-x_0^+)^{N-s_1-(s_1^\prime-1-k)}\end{array}, \nonumber
\end{eqnarray}
where $x_1^+ \triangleq (1-\varepsilon)\beta_1^+ + \varepsilon(1-\beta_1^-)$ is the probability that a rating-1 user's rating remains to be 1, and
$x_0^+ \triangleq (1-\varepsilon)\beta_0^+ + \varepsilon(1-\beta_0^-)$ is the probability that a rating-0 user's rating goes up to 1. We can see that
the transition probability has combinatorial numbers in it and is complicated. Hence, although the stationary strategies themselves are simpler than
the nonstationary strategies, they are \emph{harder to compute}, in the sense that it is difficult to derive structural results for rating mechanisms
with stationary recommended strategy. In contrast, we are able to develop a unified design framework for socially-optimal rating mechanisms with
nonstationary recommended strategies.

\subsubsection{Inefficiency}
We measure the efficiency of the rating mechanisms with stationary recommended strategies using the ``price of stationarity'' ($\mathrm{PoStat}$),
defined as
\begin{eqnarray}
\mathrm{PoStat}(\varepsilon,\mathcal{B}) = \frac{\lim_{\delta\rightarrow1} W^s(\varepsilon,\delta,\mathcal{B})}{b-c},
\end{eqnarray}
where $W^s(\varepsilon,\delta,\mathcal{B})$ is the optimal value of a modified optimization problem \eqref{eqn:DesignProblem} with an additional
constraint that $\pi_0$ is stationary.

Note that $\mathrm{PoStat}(\varepsilon,\mathcal{B})$ measures the efficiency of a class of rating mechanisms (not a specific rating mechanism),
because we optimize over all the rating update rules and stationary recommended strategies restricted on $\mathcal{B}$. $\mathrm{PoStat}$ is a number
between 0 and 1. A small $\mathrm{PoStat}$ indicates a low efficiency.

Through simulation, we can compute $\mathrm{PoStat}(0.1,A^{\rm afs})=0.720$. In other words, even with differential punishment $\alpha^{\rm f}$, the
performance of stationary strategies is bounded away from social optimum. We compute the $\mathrm{PoStat}$ in a platform with $N=5$ users, the
benefit $b=3$, the cost $c=1$, and $\varepsilon=0.1$. Under each discount factor $\delta$, we assign values between 0 and 1 with a 0.1 grid to
$\beta_\theta^+, \beta_\theta^-$ in the rating update rule, namely we try $11^4$ rating update rules to select the optimal one. For each rating
update rule, we try all the $3^{N+1}=729$ stationary recommended strategies restricted on $A^{\rm afs}$. In Table~\ref{table:ThresholdSocialWelfare},
we list normalized social welfare under different discount factors.

\begin{table}
\renewcommand{\arraystretch}{1.1} \centering
\caption{Normalized social welfare of stationary strategies restricted on $A^{\rm afs}$.} \label{table:ThresholdSocialWelfare}
\begin{tabular}{|c|c|c|c|c|c|c|}
\hline
$\delta$ & $0.7$ & $0.8$ & $0.9$ & $0.99$ & $0.999$ & $0.9999$ \\
\hline
Normalized welfare & 0.690 & 0.700 & 0.715 & 0.720 & 0.720 & 0.720 \\
\hline
\end{tabular}
\end{table}

As mentioned before, the inefficiency of stationary strategies is due to the punishment exerted under certain rating distributions. For example, the
optimal recommended strategies discussed above recommend the selfish or fair plan when at least one user has rating 0, resulting in performance loss.
One may think that when the users are more patient (i.e. when the discount factor is larger), we can use milder punishments by lowering the
punishment probabilities $\beta_1^-$ and $\beta_0^-$, such that the rating distributions with many low-rating users happen with less frequency.
However, simulations on the above strategies show that, to fulfill the IC constraints, the punishment probabilities cannot be made arbitrarily small
even when the discount factor is large. For example, Table~\ref{table:ThresholdMinimumPunishment} shows the minimum punishment probability
$\beta_1^-$ (which is smaller than $\beta_0^-$) of rating mechanisms restricted on $A^{\rm afs}$ under different discount factors. In other words,
the rating distributions with many low-rating users will happen with some probabilities bounded above zero, with a bound independent of the discount
factor. Hence, the performance loss is bounded above zero regardless of the discount factor. Note that in a nonstationary strategy, we could choose
whether to punish in rating distributions with many low-rating users, depending on the history of past rating distributions. This adaptive adjustment
of punishments allows nonstationary strategies to potentially achieve the social optimum.

\begin{table}
\renewcommand{\arraystretch}{1.1} \centering
\caption{Minimum punishment probabilities of rating mechanisms restricted on $A^{\rm afs}$ when $\varepsilon=0.1$.}
\label{table:ThresholdMinimumPunishment}
\begin{tabular}{|c|c|c|c|c|c|c|c|}
\hline
$\delta$ & $0.7$ & $0.8$ & $0.9$ & $0.99$ & $0.999$ & $0.9999$ & $0.99999$ \\
\hline
Minimum $\beta_1^-$ & 0.8 & 0.8 & 0.6 & 0.6 & 0.3 & 0.3 & 0.3 \\
\hline
\end{tabular}
\end{table}

\subsection{Lack of Differential Punishments}\label{subsec:Analysis_InefficientStrategies}
We have discussed in the previous subsection the inefficiency of stationary strategies. Now we consider a class of nonstationary strategies
restricted on the subset of plans $A^{\rm as}$. Under this class of strategies, all the users are rewarded (by choosing $\alpha^{\rm a}$) or punished
(by choosing $\alpha^{\rm s}$) simultaneously. In other words, there is no differential punishment that can ``transfer'' some payoff from low-rating
users to high-rating users. We quantify the performance loss of this class of nonstationary strategies restricted on $A^{\rm as}$ as follows.

\begin{proposition}\label{proposition:RestrictedStrategies}
For any $\varepsilon>0$, we have
\begin{eqnarray}\label{eqn:BoundedAway_StronglySymmetric}
\lim_{\delta\rightarrow1} W(\varepsilon,\delta,A^{\rm as})\leq b-c-\zeta(\varepsilon),
\end{eqnarray}
where $\zeta(\varepsilon)>0$ for any $\varepsilon>0$.
\end{proposition}
\begin{proof}
The proof is similar to the proof of \cite[Proposition~6]{Dellarocas}; see Appendix~\ref{Proof:RestrictedStrategies}.
\end{proof}

The above proposition shows that the maximum social welfare achievable by $(\pi_0,\pi\cdot\bm{1}_N)\in \Pi(A^{\rm as})\times\Pi^N(A^{\rm as})$ at the
equilibrium is bounded away from the social optimum $b-c$, unless there is no rating update error. Note that the performance loss is independent of
the discount factor. In contrast, we will show later that, if we can use the fair plan $\alpha^{\rm f}$, the social optimum can be asymptotically
achieve when the discount factor goes to $1$. Hence, the differential punishment introduced by the fair plan is crucial for achieving the social
optimum.

\section{Socially Optimal Design}\label{sec:OptimalMechanism}
In this section, we design rating mechanisms that asymptotically achieve the social optimum at the equilibrium, even when the rating update rule
$\varepsilon>0$. In our design, we use the APS technique, named after the authors of the seminal paper \cite{APS1990}, which is also used to prove
the folk theorem for repeated games in \cite{FLM1994} and for stochastic games in \cite{HornerSugayaTakahashi}. We will briefly introduce the APS
technique first. Meanwhile, more importantly, we will illustrate why we \emph{cannot} use APS in our setting in the same way as \cite{FLM1994} and
\cite{HornerSugayaTakahashi} did. Then, we will show how we use APS in a different way in our setting, in order to design the optimal rating
mechanism and to construct the equilibrium strategy. Finally, we analyze the performance of a class of simple but suboptimal strategies, which sheds
light on why the proposed strategy can achieve the social optimum.

\subsection{The APS Technique}\label{subsec:APS}
APS \cite{APS1990} provides a characterization of the set of PPE payoffs. It builds on the idea of \emph{self-generating sets} described as follows.
Note that APS is used for standard stochastic games, and recall from our discussion in Section~\ref{sec:related} that the state of the standard
stochastic game is the rating profile $\bm{\theta}$. Then define a set $\mathcal{W}^{\bm{\theta}}\subset \mathbb{R}^N$ for every state
$\bm{\theta}\in\Theta^N$, and write $(\mathcal{W}^{\bm{\theta}^\prime})_{\bm{\theta}^\prime\in\Theta^N}$ as the collection of these sets. Then we
have the following definitions \cite{APS1990}\cite{HornerSugayaTakahashi}\cite{MS2006}. First, we say a payoff profile
$\bm{v}(\bm{\theta})\in\mathbb{R}^N$ is \emph{decomposable} on $(\mathcal{W}^{\bm{\theta}^\prime})_{\bm{\theta}^\prime\in\Theta^N}$ given
$\bm{\theta}$, if there exists a recommended plan $\alpha_0$, an plan profile $\bm{\alpha}^*$, and a \emph{continuation payoff function}
$\bm{\gamma}:\Theta^N\rightarrow \cup_{\bm{\theta}^\prime\in\Theta^N} \mathcal{W}^{\bm{\theta}^\prime}$ with
$\bm{\gamma}(\bm{\theta}^\prime)\in\mathcal{W}^{\bm{\theta}^\prime}$, such that for all $i\in\mathcal{N}$ and for all $\alpha_i\in A$,
\begin{eqnarray}\label{eqn:Decomposability}
v_i &=& (1-\delta)u_i(\bm{\theta},\alpha_0,\bm{\alpha}^*) + \delta \sum_{\bm{\theta}^\prime} \gamma_i(\bm{\theta}^\prime)
q(\bm{\theta}^\prime|\bm{\theta},\alpha_0,\bm{\alpha}^*) \\
&\geq& (1-\delta)u_i(\bm{\theta},\alpha_0,\alpha_i,\bm{\alpha}_{-i}^*) + \delta \sum_{\bm{\theta}^\prime} \gamma_i(\bm{\theta}^\prime)
q(\bm{\theta}^\prime|\bm{\theta},\alpha_0,\alpha_i,\bm{\alpha}_{-i}^*). \nonumber
\end{eqnarray}
Then, we say a set $(\mathcal{W}^{\bm{\theta}})_{\bm{\theta}\in\Theta^N}$ is a self-generating set, if for any $\bm{\theta}$, every payoff profile
$\bm{v}(\bm{\theta})\in\mathcal{W}^{\bm{\theta}}$ is decomposable on $(\mathcal{W}^{\bm{\theta}^\prime})_{\bm{\theta}^\prime\in\Theta^N}$ given
$\bm{\theta}$. The important property of self-generating sets is that any self-generating set is a set of PPE payoffs
\cite{APS1990}\cite{HornerSugayaTakahashi}\cite{MS2006}.

Based on the idea of self-generating sets, \cite{FLM1994} and \cite{HornerSugayaTakahashi} proved the folk theorem for repeated games and stochastic
games, respectively. However, we cannot use APS in the same way as \cite{FLM1994} and \cite{HornerSugayaTakahashi} did for the following reason. We
assume that the users do not know the rating profile of every user, and restrict our attention on symmetric PA strategy profiles. This requires that
each user $i$ cannot use the continuation payoff function $\gamma_i(\bm{\theta})$ directly. Instead, each user $i$ should assign the same
continuation payoff for the rating profiles that have the same rating distribution, namely $\gamma_i(\bm{\theta})=\gamma_i(\bm{\theta}^\prime)$ for
all $\bm{\theta}$ and $\bm{\theta}^\prime$ such that $\bm{s}(\bm{\theta})=\bm{s}(\bm{\theta}^\prime)$.

\subsection{Socially Optimal Design}
As mentioned before, the social optimum $b-c$ can be exactly achieved only by servers providing high-quality service all the time, which is not an
equilibrium. Hence, we aim at achieving the social optimum $b-c$ asymptotically. We define the asymptotically socially optimal rating mechanisms as
follows.

\begin{definition}[Asymptotically Socially Optimal Rating Mechanisms]\label{def:AsymptoticallyOptimal}
Given a rating update error $\varepsilon\in[0,0.5)$, we say a rating mechanism $(\tau(\varepsilon),\pi_0(\varepsilon,\xi,\delta)\in\Pi)$ is
asymptotically socially optimal under $\varepsilon$, if for any small performance loss $\xi>0$, we can find a $\underline{\delta}(\xi)$, such that
for any discount factor $\delta>\underline{\delta}(\xi)$, we have
\begin{itemize}
\item $(\pi_0(\xi,\delta),\pi_0(\xi,\delta)\cdot\bm{1}_N)$ is a PAE;
\item $U_i(\bm{\theta}^0,\pi_0,\pi_0\cdot\bm{1}_N)\geq b-c-\xi,~\forall i\in\mathcal{N},~\forall \bm{\theta}^0$.
\end{itemize}
\end{definition}

Note that in the asymptotically socially optimal rating mechanism, the rating update rule depends only on the rating update error, and works for any
tolerated performance loss $\xi$ and for any the discount factor $\delta>\underline{\delta}$. The recommended strategy $\pi_0$ is a class of
strategies parameterized by $(\varepsilon, \xi, \delta)$, and works for any $\varepsilon\in[0,0.5)$, any $\xi>0$ and any discount factor
$\delta>\underline{\delta}$ under the rating update rule $\tau(\varepsilon)$.

First, we define a few auxiliary variables first for better exposition of the theorem. Define $\kappa_1\triangleq\frac{b}{\frac{N-2}{N-1}b-c}-1$ and
$\kappa_2\triangleq 1+\frac{c}{(N-1)b}$. In addition, we write the probability that a user with rating $1$ has its rating remain at $1$ if it follows
the recommended altruistic plan $\alpha^{\rm a}$ as:
$$
x_1^+\triangleq(1-\varepsilon)\beta_1^+ + \varepsilon(1-\beta_1^-).
$$
Write the probability that a user with rating $1$ has its rating remain at $1$ if it follows the recommended fair plan $\alpha^{\rm f}$ as:
$$
x_{s_1(\bm{\theta})}\triangleq \left[(1-\varepsilon)\frac{s_1(\bm{\theta})-1}{N-1}+\frac{N-s_1(\bm{\theta})}{N-1}\right] \beta_1^+ +
\left(\varepsilon\frac{s_1(\bm{\theta})-1}{N-1}\right) (1-\beta_1^-).
$$
Write the probability that a user with rating $0$ has its rating increase to $1$ if it follows the recommended plan $\alpha^{\rm a}$ or $\alpha^{\rm
f}$:
$$
x_0^+\triangleq(1-\varepsilon)\beta_0^+ + \varepsilon(1-\beta_0^-).
$$

\begin{theorem}\label{theorem:AchieveSocialOptimum}
Given any rating update error $\varepsilon\in[0,0.5)$,
\begin{itemize}
\item \emph{(Design rating update rules):} A rating update rule $\tau(\varepsilon)$ that satisfies
\begin{itemize}
\item Condition~1 (following the recommended plan leads to a higher rating): $$\beta_1^+>1-\beta_1^-~\mathrm{and}~\beta_0^+>1-\beta_0^-,$$
\item Condition~2 (Enough ``reward'' for users with rating 1): $$x_1^+=(1-\varepsilon)\beta_1^+ + \varepsilon(1-\beta_1^-) > \frac{1}{1+\frac{c}{(N-1)b}},$$
\item Condition~3 (Enough ``punishment'' for users with rating 0): $$x_0^+=(1-\varepsilon)\beta_0^+ + \varepsilon(1-\beta_0^-) < \frac{1-\beta_1^+}{\frac{c}{(N-1)b}},$$
\end{itemize}
can be the rating update rule in a asymptotically socially-optimal rating mechanism.
\item \emph{(Optimal recommended strategies):} Given the rating update rule $\tau(\varepsilon)$ that satisfies the above conditions, any small performance loss $\xi>0$,
and any discount factor $\delta\geq\underline{\delta}(\varepsilon,\xi)$ with $\underline{\delta}(\varepsilon,\xi)$ defined in
Appendix~\ref{appendix:LowerBoundDiscountFactor}, the recommended strategy $\pi_0(\varepsilon,\xi,\delta)\in\Pi(A^{\rm afs})$ constructed by
Table~\ref{table:EquilibriumStrategy} is the recommended strategy in a asymptotically socially-optimal rating mechanism.
\end{itemize}
\end{theorem}
\begin{proof}
See Appendix~\ref{Proof:AchieveSocialOptimum} for the entire proof. We provide a proof sketch here.

The proof builds on the theory of self-generating sets \cite{APS1990}, which can be considered as the extension of Bellman equations in dynamic
programming to the cases with multiple self-interested users using nonstationary strategies. We can decompose each user $i$'s discounted average
payoff into the current payoff and the continuation payoff as follows:
\begin{eqnarray}
&& U_i(\bm{\theta}^0,\pi_0,\bm{\pi}) \nonumber \\
&=& \mathbb{E}_{\bm{\theta}^1,\ldots} \!\left\{\!(1-\delta)\sum_{t=0}^\infty \delta^t
u_i\left(\theta_i^t,\bm{s}(\bm{\theta}^t),\bm{\pi}(\bm{s}(\bm{\theta}^0),\ldots,\bm{s}(\bm{\theta}^t)\right)\!\right\} \nonumber \\
         &=& (1-\delta) \!\cdot \!\underbrace{u_i\left(\theta_i^0,\bm{s}(\bm{\theta}^0),\bm{\pi}(\bm{s}(\bm{\theta}^0))\right)}_{\mathrm{current~payoff~at~}t=0} + \delta \cdot \underbrace{\mathbb{E}_{\bm{\theta}^2,\ldots} \!\!\left\{\!\!(1-\delta)\!\sum_{t=1}^\infty \delta^{t-1}
u_i\left(\theta_i^t,\bm{s}(\bm{\theta}^t),\bm{\pi}(\bm{s}(\bm{\theta}^0),\ldots,\bm{s}(\bm{\theta}^t))\right)\!\!\right\}}_{\mathrm{continuation~payoff~starting~from~}t=1}.
\nonumber
\end{eqnarray}
We can see that the continuation payoff starting from $t=1$ is the discounted average payoff as if the system starts from $t=1$. Suppose that the
users follow the recommended strategy. Since the recommended strategy and the rating update rule do not differentiate users with the same rating, we
can prove that the users with the same rating have the same continuation payoff starting from any point. Hence, given $\pi_0$ and
$\bm{\pi}=\pi_0\cdot\bm{1}_N$, the decomposition above can be simplified into
\begin{eqnarray}\label{eqn:DecompositionPayoff_Equality}
v_\theta^{\pi_0}(\bm{s}) = (1-\delta) \cdot u\left(\theta,\bm{s},\alpha_0\cdot\bm{1}_N\right) + \delta \cdot \sum_{\theta^\prime=0}^1
\sum_{\bm{s}^\prime} q(\theta^\prime,\bm{s}^\prime|\theta,\bm{s},\alpha_0\cdot\bm{1}_N) \cdot v_{\theta^\prime}^{\pi_0}(\bm{s}^\prime),
\end{eqnarray}
where $q(\theta^\prime,\bm{s}^\prime|\theta,\bm{s},\alpha\cdot\bm{1}_N)$ is the probability that the user has rating $\theta^\prime$ and the rating
distribution is $\bm{s}^\prime$ in the next period given the user's current rating $\theta$, the current rating distribution $\bm{s}$, and the action
profile $\alpha\cdot\bm{1}_N$, and $v_\theta^{\pi_0}(\bm{s})$ is the continuation payoff of the users with rating $\theta$ starting from the initial
rating distribution $\bm{s}$.

The differences between \eqref{eqn:DecompositionPayoff_Equality} and the Bellman equations are 1) the ``value'' of state $\bm{s}$ in
\eqref{eqn:DecompositionPayoff_Equality} is a vector comprised of rating-$1$ and rating-$0$ users' values, compared to scalar values in Bellman
equations, and 2) the value of state $\bm{s}$ is not fixed in \eqref{eqn:DecompositionPayoff_Equality}, because the action $\alpha_0$ taken under
state $\bm{s}$ is not fixed in a nonstationary strategy (this is also the key difference from the analysis of stationary strategies; see
\eqref{eqn:ValueIteration} where the action taken at state $\bm{s}$ is fixed to be $\pi_0(\bm{s})$). In addition, for an equilibrium recommended
strategy, the decomposition needs to satisfy the following incentive constraints: for all $\alpha$,
\begin{eqnarray}\label{eqn:DecompositionPayoff_IC}
v_\theta^{\pi_0}(\bm{s}) \geq (1-\delta) \cdot u\left(\theta,\bm{s},(\alpha, \alpha_0\cdot\bm{1}_{N-1})\right) + \delta \cdot
\sum_{\theta^\prime=0}^1 \sum_{\bm{s}^\prime} \rho(\theta^\prime,\bm{s}^\prime|\theta,\bm{s},(\alpha, \alpha_0\cdot\bm{1}_{N-1})) \cdot
v_{\theta^\prime}^{\pi_0}(\bm{s}^\prime).
\end{eqnarray}

\begin{figure}
\centering
\includegraphics[width =5.0in]{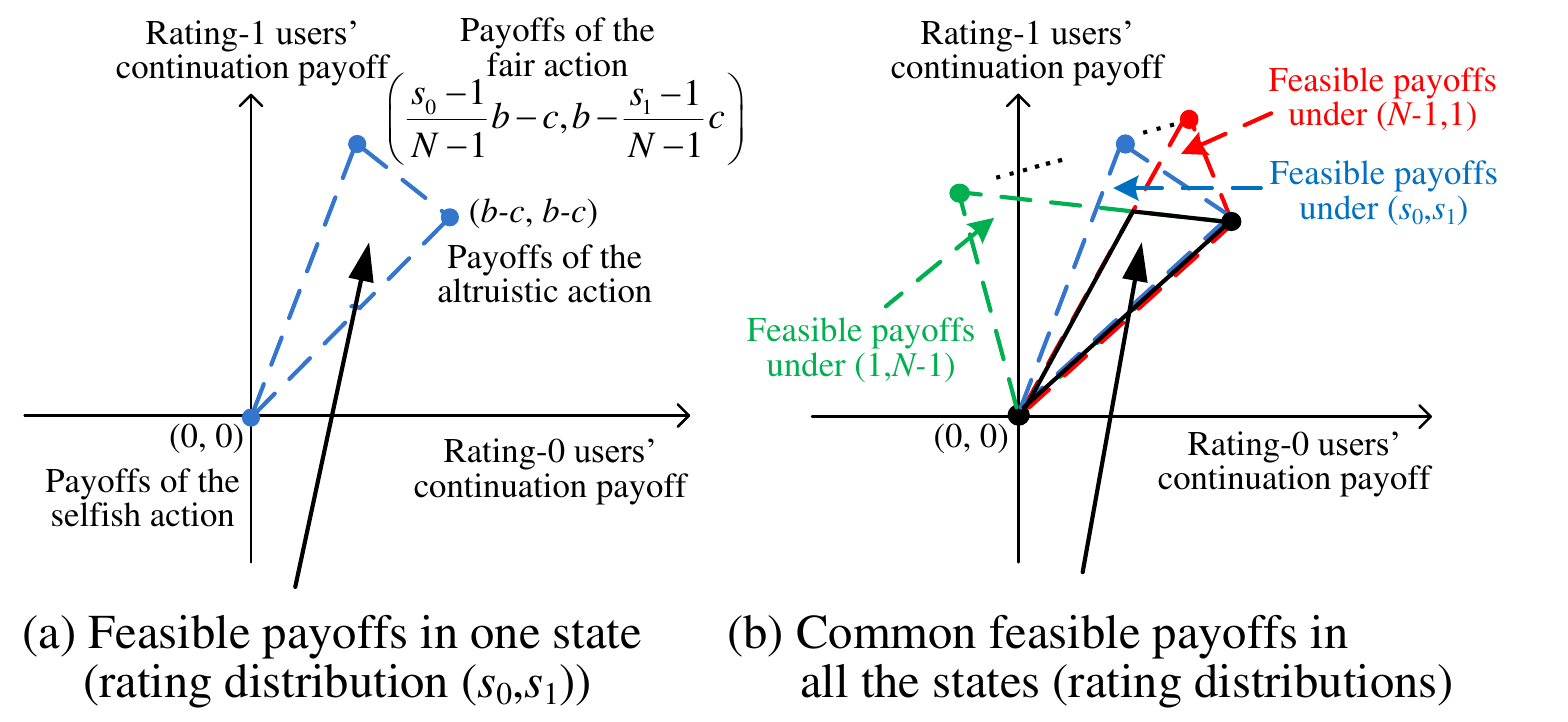}
\caption{Illustration of how to build the self-generating set. The left figure shows the set of feasible payoffs in one state (i.e. under the rating
distribution $(s_0,s_1)$). The right figure shows the sets of feasible payoffs in different states (i.e. rating distributions) and their
intersection, namely the set of common feasible payoffs in all the states (i.e. under all the rating distributions).}
\label{fig:BuildSelfGeneratingSet}
\end{figure}

\begin{figure}
\centering
\includegraphics[width =3.2in]{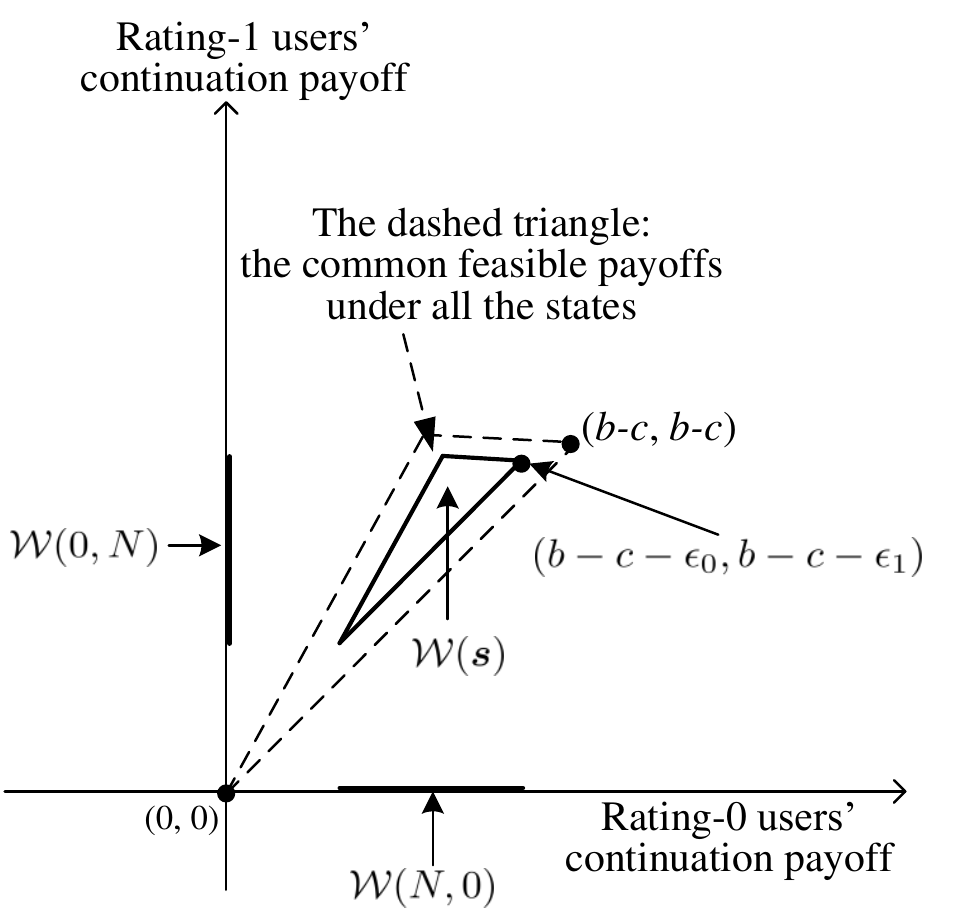}
\caption{Illustration of the self-generating set, which is a triangle within the set of common feasible payoffs in
Fig.~\ref{fig:BuildSelfGeneratingSet}.} \label{fig:SelfGeneratingSet}
\end{figure}

To analyze nonstationary strategies, we use the theory of self-generating sets. Note, however, that \cite{APS1990} does not tell us how to construct
a self-generating set, which is exactly the major difficulty to overcome in our proof. In our proof, we construct the following self-generating set.
First, since the strategies depend on rating distributions only, we let $\mathcal{W}(\bm{\theta})=\mathcal{W}(\bm{\theta}^\prime)$ for any
$\bm{\theta}$ and $\bm{\theta}$ that have the same rating distribution. Hence, in the rest of the proof sketch, we write the self-generating set as
$\{\mathcal{W}(\bm{s})\}_{\bm{s}}$, which is illustrated in Fig.~\ref{fig:BuildSelfGeneratingSet} and Fig.~\ref{fig:SelfGeneratingSet}.
Fig.~\ref{fig:BuildSelfGeneratingSet} shows how to construct the self-generating set. The left of Fig.~\ref{fig:BuildSelfGeneratingSet} shows the
feasible payoffs in one state, and the right shows the common feasible payoffs in all the states (we consider the common feasible payoffs such that
we can use the same $\mathcal{W}(\bm{s})$ under all the states $\bm{s}$). The self-generating set is a subset of the common feasible payoffs, as
illustrated in Fig.~\ref{fig:SelfGeneratingSet}. When the users have different ratings (i.e. $1\leq s_0\leq N-1$), the set $\mathcal{W}(\bm{s})$ is
the triangle shown in Fig.~\ref{fig:SelfGeneratingSet}, which is congruent to the triangle of common feasible payoffs (shown in dashed lines), and
has the upper right vertex at $(b-c-\epsilon_0,b-c-\epsilon_1)$ with $\epsilon_0,\epsilon_1\leq\xi$. We have the analytical expression for the
triangle in Appendix~\ref{Proof:AchieveSocialOptimum}. When all the users have the same rating (i.e. $s_0=0$ or $s_0=N$), only one component in
$\bm{v}(\bm{s})$ is relevant. Hence, the sets $\mathcal{W}((N,0))$ and $\mathcal{W}((0,N))$ are line segments determined by the ranges of $v_0$ and
$v_1$ in the triangle, respectively.

In addition, we simplify the decomposition \eqref{eqn:DecompositionPayoff_Equality} and \eqref{eqn:DecompositionPayoff_IC} by letting the
continuation payoffs $v_{\theta^\prime}^{\pi_0}(\bm{s}^\prime)=v_{\theta^\prime}^{\pi_0}$ for all $\bm{s}^\prime$. Hence, for a given $\bm{s}$ and a
payoff vector $\bm{v}(\bm{s})$, the continuation payoffs $\bm{v}^\prime=(v_0^\prime,v_1^\prime)$ can be determined by solving the following two
linear equations:
\begin{eqnarray}\label{eqn:DecompositionPayoffSimplified_Equality}
\!\left\{\begin{array}{l} v_0(\bm{s}) \!=\! (1-\delta) u\left(0,\bm{s},\alpha_0\bm{1}_N\right) \!+\! \delta \sum_{\theta^\prime=0}^1
q(\theta^\prime|0,\alpha_0\bm{1}_N) v_{\theta^\prime}^{\prime} \\
v_1(\bm{s}) \!= \!(1-\delta) u\left(1,\bm{s},\alpha_0\bm{1}_N\right) \!+\! \delta \sum_{\theta^\prime=0}^1 q(\theta^\prime|1,\alpha_0\bm{1}_N)
v_{\theta^\prime}^{\prime}\end{array}\right.\!\!\!\!\!\!\!\!
\end{eqnarray}
where $q(\theta^\prime|\theta,\alpha_0)$ is the probability that the next rating is $\theta^\prime$ for a user with rating $\theta$ under the plan
profile $\alpha_0\cdot\bm{1}_N$.

Based on the above simplification, the collection of sets $\{\mathcal{W}(\bm{s})\}_{\bm{s}}$ in Fig.~\ref{fig:SelfGeneratingSet} is a self-generating
set, if for any $\bm{s}$ and any payoff vector $\bm{v}(\bm{s})\in\mathcal{W}(\bm{s})$, we can find a plan $\alpha_0$ such that the continuation
payoffs $\bm{v}^\prime$ calculated from \eqref{eqn:DecompositionPayoffSimplified_Equality} lie in the triangle and satisfy the incentive constraints
in \eqref{eqn:DecompositionPayoff_IC}.

In summary, we can prove that the collection of sets $\{\mathcal{W}(\bm{s})\}_{\bm{s}}$ illustrated in Fig.~\ref{fig:SelfGeneratingSet} is a
self-generating set under certain conditions. Specifically, given a performance loss $\xi$, we construct the corresponding
$\{\mathcal{W}(\bm{s})\}_{\bm{s}}$, and prove that it is a self-generating set under the following conditions: 1) the discount factor
$\delta\geq\underline{\delta}(\varepsilon,\xi)$ with $\underline{\delta}(\varepsilon,\xi)$ defined in
Appendix~\ref{appendix:LowerBoundDiscountFactor}, and 2) the three conditions on the rating update rule in
Theorem~\ref{theorem:AchieveSocialOptimum}. This proves the first part of Theorem~\ref{theorem:AchieveSocialOptimum}.

\begin{figure}
\centering
\includegraphics[width =6.0in]{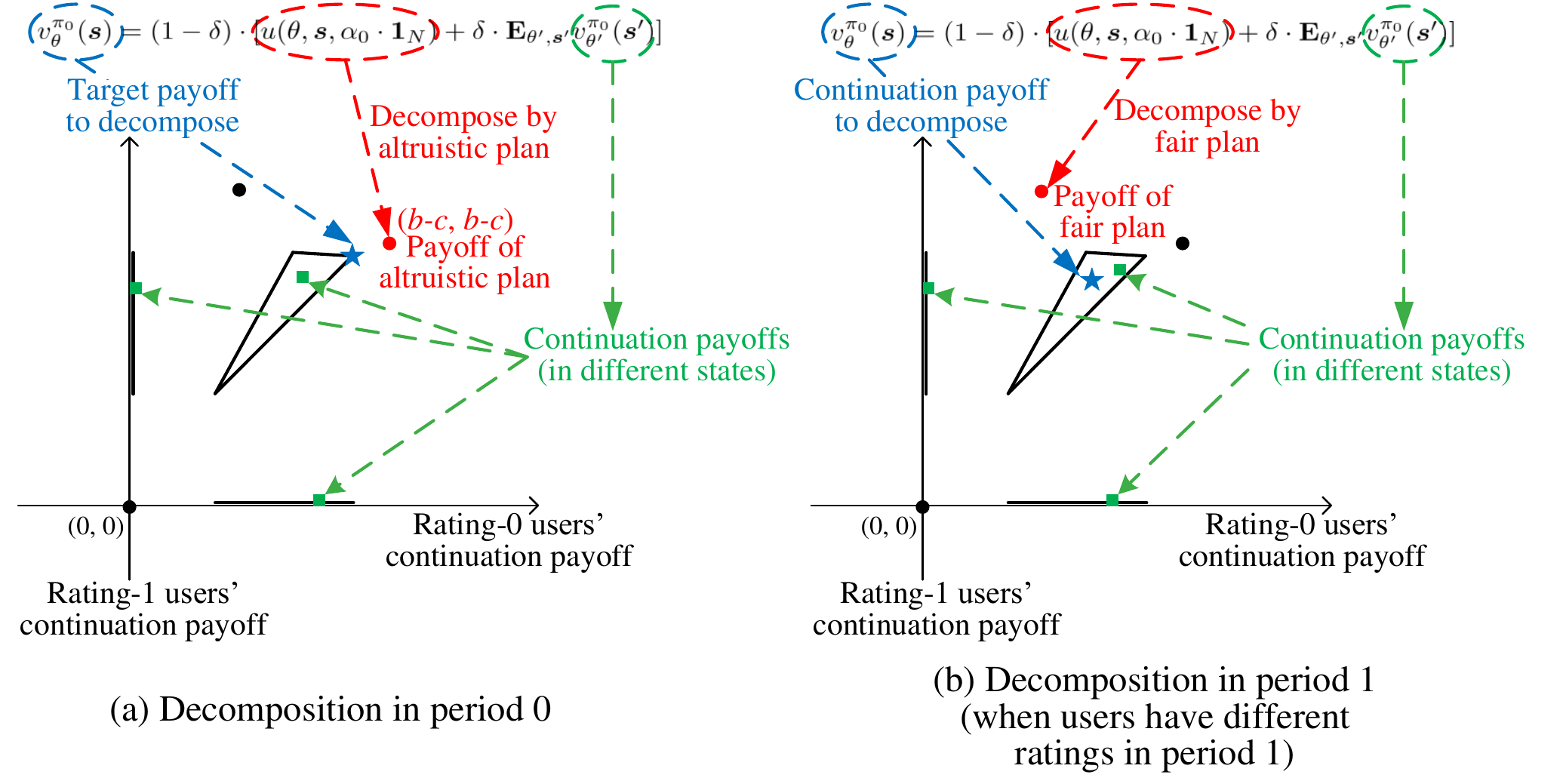}
\caption{The decomposition of payoffs. The left figure shows the decomposition in period $0$, when the payoff to decompose is the target payoff
$(b-c-\epsilon_0,b-c-\epsilon_1)$; the right figure shows the decomposition in period $1$, when the payoff to decompose is the continuation payoff
starting from period $1$ and when the users have different ratings.} \label{fig:Decomposition}
\end{figure}

The corresponding recommended strategy can be determined based on the decomposition of payoffs. Specifically, given the current rating distribution
$\bm{s}$ and the current expected payoffs $\bm{v}(\bm{s})\in\mathcal{W}(\bm{s})$, we find a plan $\alpha_0$ such that the continuation payoffs
$\bm{v}^\prime$ calculated from \eqref{eqn:DecompositionPayoffSimplified_Equality} lie in the triangle and satisfy the incentive constraints. The
decomposition is illustrated in Fig.~\ref{fig:Decomposition}. One important issue in the decomposition is which plan should be used to decompose the
payoff. We prove that we can determine the plan in the following way (illustrated in Fig.~\ref{fig:AlgorithmIllustration}). When the users have
different ratings, choose the altruistic plan $\alpha^{\rm a}$ when $\bm{v}(\bm{s})$ lies in the region marked by ``a'' in the triangle in
Fig.~\ref{fig:AlgorithmIllustration}-(b), and choose the fair plan $\alpha^{\rm f}$ otherwise. When the users have the same rating $0$ or $1$, choose
the altruistic plan $\alpha^{\rm a}$ when $v_0(\bm{s})$ or $v_1(\bm{s})$ lies in the region marked by ``a'' in the line segment in
Fig.~\ref{fig:AlgorithmIllustration}-(a) or Fig.~\ref{fig:AlgorithmIllustration}-(c), and choose the selfish plan $\alpha^{\rm s}$ otherwise. Note
that we can analytically determine the line that separates the two regions in the triangle and the threshold that separates the two regions in the
line segment (analytical expressions are omitted due to space limitation; see Appendix~\ref{EquilibriumStrategyComplete} for details). The above
decomposition is repeated, and is used to determine the recommended plan in each period based on the current rating distribution $\bm{s}$ and the
current expected payoffs to achieve $\bm{v}(\bm{s})$. The procedure described above is exactly the algorithm to construct the recommended strategy,
which is described in Table~\ref{table:EquilibriumStrategy}. Due to space limitation, Table~\ref{table:EquilibriumStrategy} is illustrative but not
specific. The detailed table that describes the algorithm can be found in Appendix~\ref{EquilibriumStrategyComplete}.
\end{proof}

\begin{figure}
\centering
\includegraphics[width =5.0in]{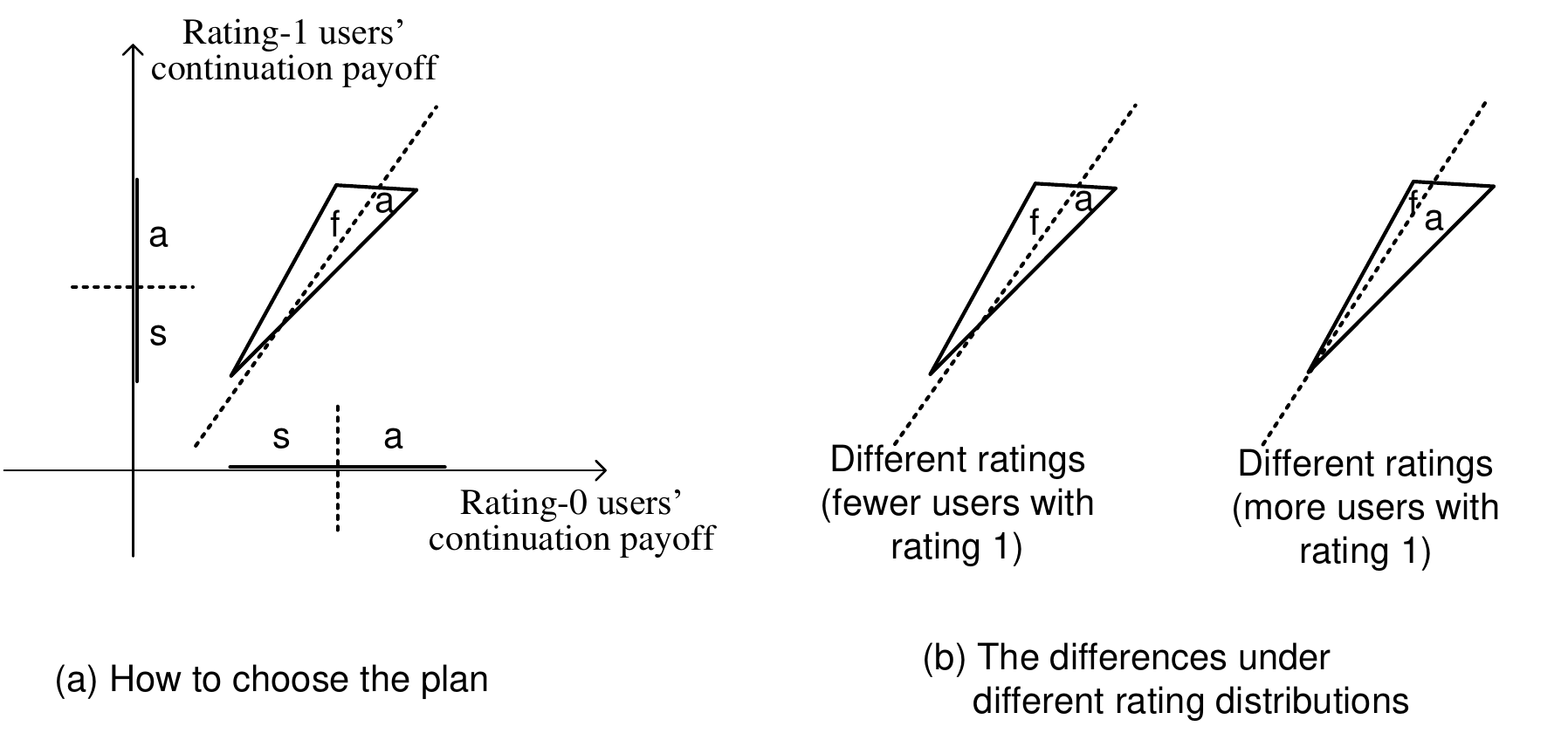}
\caption{Illustration of how to choose the plan in order to decompose a given payoff. Each self-generating set is partitioned into two parts. In each
period, a recommended plan (the altruistic plan ``a'', the fair plan ``f'', or the selfish plan ``s'') is chosen, depending on which part of the
self-generating set the expected payoffs fall into.} \label{fig:AlgorithmIllustration}
\end{figure}

Theorem~\ref{theorem:AchieveSocialOptimum} proves that for any rating update error $\varepsilon\in[0,0.5)$, we can design an asymptotically optimal
rating mechanism. The design of the asymptotically optimal rating mechanism consists of two parts. The first part is to design the rating update
rule. First, we should give incentives for the users to provide high-quality service, by setting $\beta_{\theta}^+$, the probability that the rating
goes up when the service quality is not lower than the recommended quality, to be larger than $1-\beta_{\theta}^-$, the probability that the rating
goes up when the service quality is lower than the recommended quality. Second, for a user with rating 1, the expected probability that its rating
goes up when it complies should be larger than the threshold specified in Condition~2 ($x_{s_1}^+>x_1^+$ implies that $x_{s_1}^+$ is larger than the
threshold, too). This gives users incentives to obtain rating $1$. Meanwhile, for a user with rating 0, the expected probability that its rating goes
up when it complies, $x_0^+$, should be smaller than the threshold specified in Condition~3. This provides necessary punishment for a user with
rating 0. Note that Conditions~2 and 3 imply that $x_1^+>x_0^+$. In this way, a user will prefer to have rating 1.

The second part is to construct the equilibrium recommended strategy. Theorem~\ref{theorem:AchieveSocialOptimum} proves that for any feasible
discount factor $\delta$ no smaller than the lower-bound discount factor $\underline{\delta}(\varepsilon,\xi)$ defined in
Appendix~\ref{appendix:LowerBoundDiscountFactor}, we can construct the corresponding recommended strategy such that each user can achieve an
discounted average payoff of at least $b-c-\xi$. Now we show how to construct the recommended strategy. Note that determining the lower-bound
discount factor $\underline{\delta}(\varepsilon,\xi)$ analytically is important for constructing the equilibrium $(\pi_0,\pi_0\cdot\bm{1}_N)$,
because a feasible discount factor is needed to determine the strategy. In \cite{FLM1994} and \cite{HornerSugayaTakahashi}, the lower bound for the
discount factor cannot be obtained analytically. Hence, their results are not constructive.

\begin{table}
\centering
\renewcommand{\arraystretch}{1.0} \caption{Algorithm to construct recommended strategies.}{
\begin{tabular}{l}
\hline \hline
\textbf{Require:} $b$, $c$, $\varepsilon$, $\xi$; $\tau(\varepsilon)$, $\delta\geq\underline{\delta}(\varepsilon,\xi)$; $\bm{\theta}^0$ \\
\hline
\textbf{Initialization:} $t=0$, $\epsilon_0=\xi$, $\epsilon_1=\epsilon_0/(1+\frac{\kappa_2}{\kappa_1})$, $v^\theta=b-c-\epsilon_\theta$, $\bm{\theta}=\bm{\theta}^0$. \\
\hline
\textbf{repeat} \\
~~~~\textbf{if} $s_1(\bm{\theta})=0$ \textbf{then}  \\
~~~~~~~~\textbf{if} $(v^0,v^1)$ lies in region ``a'' of the horizontal line segment in Fig.~\ref{fig:AlgorithmIllustration}-(a) \\
~~~~~~~~~~~~choose recommended plan $\alpha^{\rm a}$ \\
~~~~~~~~\textbf{else} \\
~~~~~~~~~~~~choose recommended plan $\alpha^{\rm s}$ \\
~~~~~~~~\textbf{end} \\
~~~~\textbf{elseif} $s_1(\bm{\theta})=N$ \textbf{then}  \\
~~~~~~~~\textbf{if} $(v^0,v^1)$ lies in region ``a'' of the vertical line segment in Fig.~\ref{fig:AlgorithmIllustration}-(c) \\
~~~~~~~~~~~~choose recommended plan $\alpha^{\rm a}$ \\
~~~~~~~~\textbf{else} \\
~~~~~~~~~~~~choose recommended plan $\alpha^{\rm s}$ \\
~~~~~~~~\textbf{end} \\
~~~~\textbf{else}  \\
~~~~~~~~\textbf{if} $(v^0,v^1)$ lies in region ``a'' of the triangle in Fig.~\ref{fig:AlgorithmIllustration}-(b) \\
~~~~~~~~~~~~choose recommended plan $\alpha^{\rm a}$ \\
~~~~~~~~\textbf{else} \\
~~~~~~~~~~~~choose recommended plan $\alpha^{\rm f}$ \\
~~~~~~~~\textbf{end} \\
~~~~\textbf{end} \\
~~~~determine the continuation payoffs $(v_0^\prime,v_1^\prime)$ according to \eqref{eqn:DecompositionPayoffSimplified_Equality} \\
~~~~$t\leftarrow t+1$, determine the rating profile $\bm{\theta}^t$, set $\bm{\theta}\leftarrow\bm{\theta}^t$, $(v_0,v_1)\leftarrow(v_0^\prime,v_1^\prime)$ \\
\textbf{until} $\varnothing$ \\
\hline \hline
\end{tabular}}
\label{table:EquilibriumStrategy}
\end{table}

The algorithm in Table~\ref{table:EquilibriumStrategy} that constructs the optimal recommended strategy works as follows. In each period, the
algorithm updates the continuation payoffs $(v_0,v_1)$, and determines the recommended plan based on the current rating distribution and the
continuation payoffs. In Fig.~\ref{fig:AlgorithmIllustration}, we illustrate which plan to recommend based on where the continuation payoffs locate
in the self-generating sets. Specifically, each set $\mathcal{W}(\bm{s})$ is partitioned into two parts (the partition lines can determined
analytically; see Appendix~\ref{EquilibriumStrategyComplete} for the analytical expressions). When all the users have rating 0 (or 1), we recommend
the altruistic plan $\alpha^{\rm a}$ if the continuation payoff $v_0$ (or $v_1$) is large, and the selfish plan $\alpha^{\rm s}$ otherwise. When the
users have different ratings, we recommend the altruistic plan $\alpha^{\rm a}$ when $(v_0,v_1)$ lies in the region marked by ``a'' in the triangle
in Fig.~\ref{fig:AlgorithmIllustration}-(b), and the fair plan $\alpha^{\rm f}$ otherwise. Note that the partition of $\mathcal{W}(\bm{s})$ is
different under different rating distributions (e.g., the region in which the altruistic plan is chosen is larger when more users have rating $1$).
Fig.~\ref{fig:AlgorithmIllustration} also illustrates why the strategy is nonstationary: the recommended plan depends on not only the current rating
distribution $\bm{s}$, but also which region of $\mathcal{V}(\bm{s})$ the continuation payoffs $(v_0,v_1)$ lie in.

\emph{\textbf{Complexity:}} Although the design of the optimal recommended strategy is complicated, the implementation is simple. The computational
complexity in each period comes from 1) identifying which region the continuation payoffs lie in, which is simple because the regions are divided by
a straight line that is analytically determined, and 2) updating the continuation payoffs $(v_0,v_1)$ by
\eqref{eqn:DecompositionPayoffSimplified_Equality}, which can be easily done by solving a set of two linear equations with two variables. The memory
complexity is also low: because we summarize the history of past states by the continuation payoffs $(v_0,v_1)$, the protocol does not need to store
all the past states.

\subsection{Whitewashing-Proofness}
An important issue in rating mechanisms is whitewashing, namely users with low ratings can register as a new user to clear its history of bad
behaviors. We say a rating mechanism is whitewashing-proof, if the cost of whitewashing (e.g. creating a new account) is higher than the benefit from
whitewashing. The benefit from whitewashing is determined by the difference between the current continuation payoff of a low-rating user and the
target payoff of a high-rating user. Since this difference is relatively small under the proposed rating mechanism, the proposed mechanism is robust
to whitewashing.

\begin{proposition}\label{theorem:WhitewashingProof}
Given the performance loss tolerance $\xi>0$, the proposed rating mechanism is whitewashing-proof if the cost of whitewashing is larger than
$\left(1-\frac{1}{\kappa_1}-\frac{1}{\kappa_2}\right)\cdot\xi$.
\end{proposition}
\begin{proof}
We illustrate the proof using Fig.~\ref{fig:WhitewashingProof}. In Fig.~\ref{fig:WhitewashingProof}, we show the self-generating set again, and point
out the target payoff of a rating-1 user and the lowest continuation payoff of a rating-0 user. The difference between these two payoffs is the
highest benefit that a rating-0 user can get by whitewashing. Simple calculation tells us that the difference is
$\left(1-\frac{1}{\kappa_1}-\frac{1}{\kappa_2}\right)\cdot\xi$, which completes the proof of Proposition~\ref{theorem:WhitewashingProof}.
\begin{figure}
\centering
\includegraphics[width =3.0in]{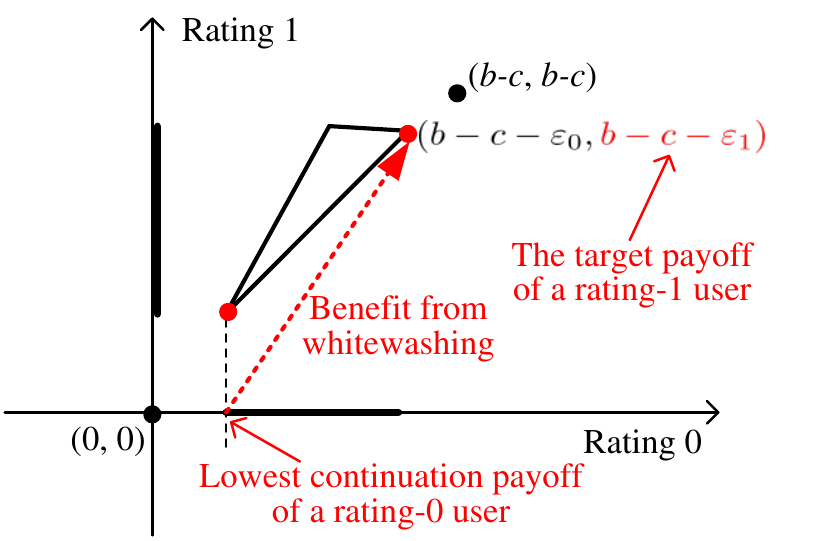}
\caption{Illustration of the target payoff of a rating-1 user and the lowest continuation payoff of a rating-0 user.} \label{fig:WhitewashingProof}
\end{figure}
\end{proof}

\section{Simulation Results}\label{sec:Simulation}

We compare against the rating mechanism with threshold-based stationary recommended strategies. In particular, we focus on threshold-based stationary
recommended strategies that use two plans. In other words, one plan is recommended when the number of rating-1 users is no smaller than the
threshold, and the other plan is recommended otherwise. In particular, we consider threshold-based stationary recommended strategies restricted on
$A^{\rm af}$, $A^{\rm as}$, and $A^{\rm fs}$, and call them ``Threshold AF'', ``Threshold AS'', and ``Threshold FS'', respectively. We focus on
threshold-based strategies because it is difficult to find the optimal stationary strategy in general when the number of users is large (the number
of stationary strategies grows exponentially with the number of users). In our experiments, we fix the following system parameters: $N=10, b=3, c=1$.

In Fig.~\ref{fig:Evolution}, we first illustrate the evolution of the states and the recommended plans taken under the proposed rating mechanism and
the rating mechanism with the Threshold AF strategy. The threshold is set to be 5. Hence, it recommends the altruistic plan when at least half of the
users have rating 1, and recommends the fair plan otherwise. We can see that in the proposed strategy, the plans taken can be different at the same
state. In particular, in ``bad'' states (6,4) and (7,3) at time slot 3 and 5, respectively, the proposed rating mechanism may recommend the fair plan
(as a punishment) and the altruistic plan (i.e. do not punish because the punishment happens in time slot 3), while the stationary mechanism always
recommends the fair plan to punish the low-rating users.

\begin{figure}
\centering
\includegraphics[width =4.5in]{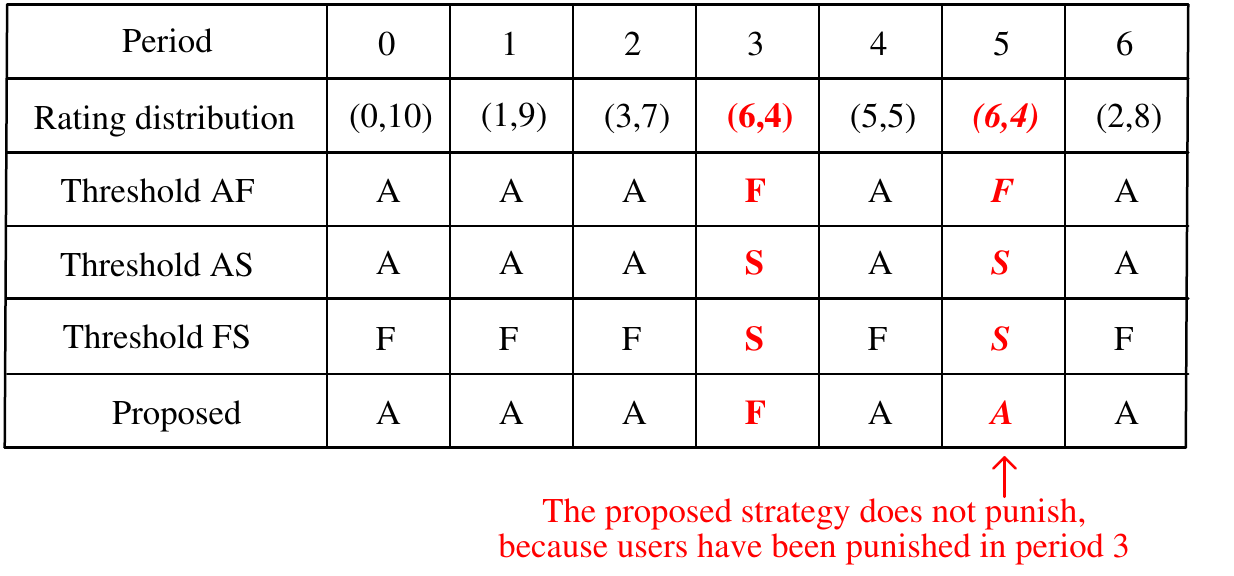}
\caption{Evolution of states and recommended plans taken in different rating mechanisms.} \label{fig:Evolution}
\end{figure}

\begin{figure}
\centering
\includegraphics[width =2.5in]{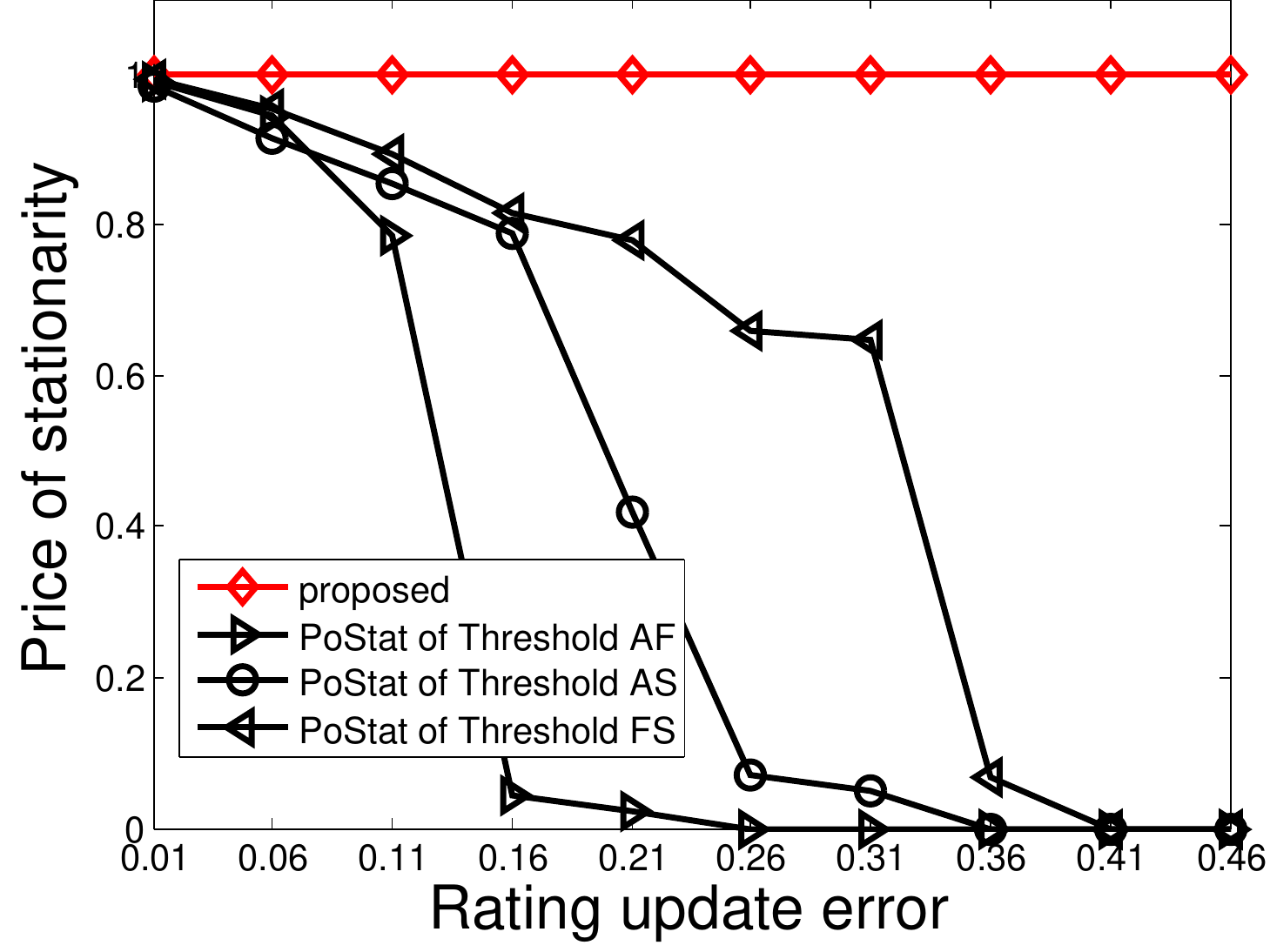}
\caption{Price of stationarity of different stationary rating mechanisms under different rating update errors.} \label{fig:PoStat}
\end{figure}

Then in Fig.~\ref{fig:PoStat}, we show the price of stationarity of three representative stationary rating mechanisms: the one with the optimal
Threshold AF strategy, the one with the optimal Threshold AS strategy, and the one with the optimal Threshold FS strategy. We can see from
Fig.~\ref{fig:PoStat} that as the rating update error increases, the efficiency of stationary rating mechanisms decreases, and drops to $0$ when the
error probability is large (e.g. when $\varepsilon>0.4$). In contrast, the proposed rating mechanism can achieve arbitrarily close to the social
optimum.

\begin{figure}
\centering
\includegraphics[width =3.5in]{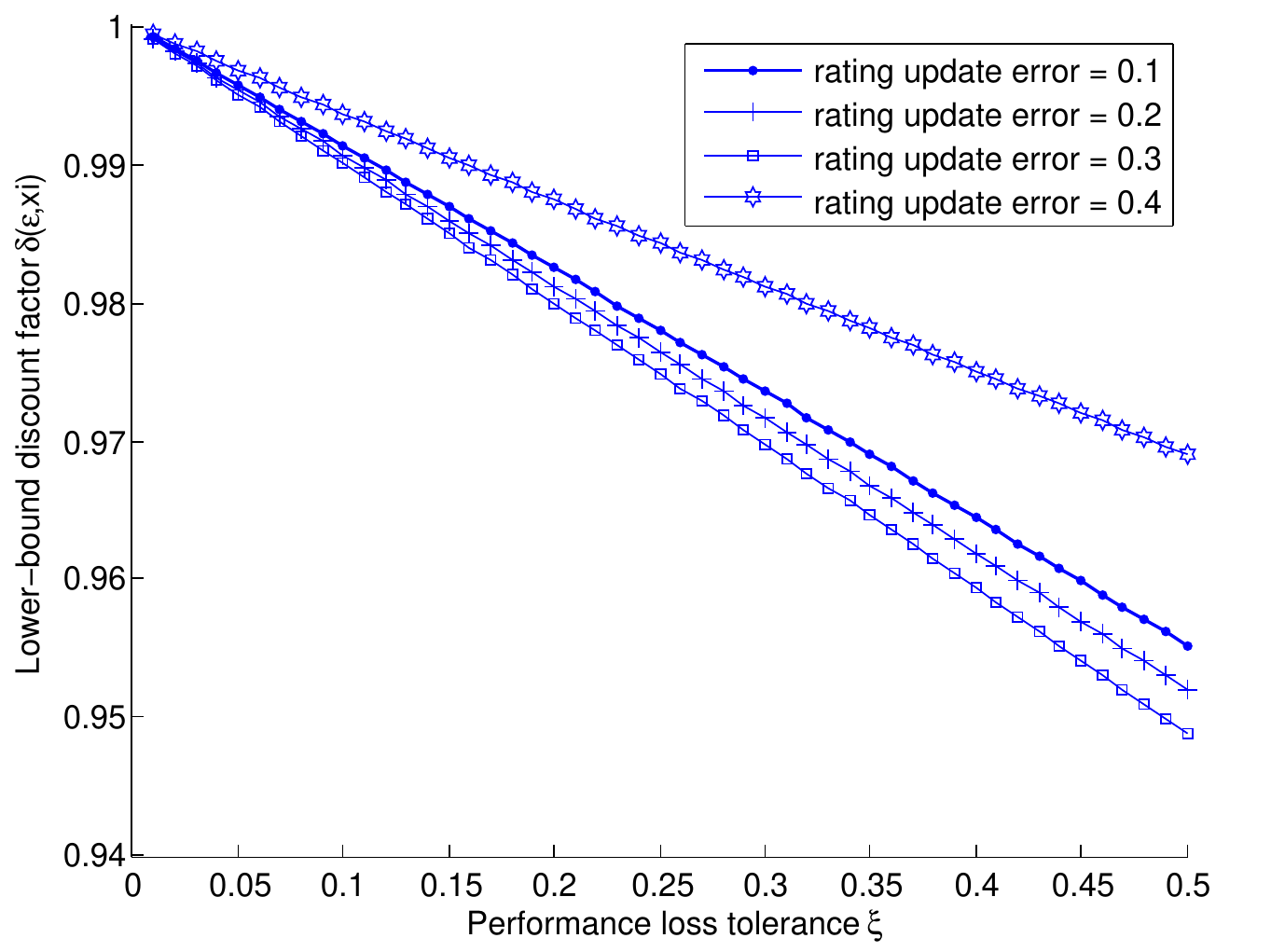}
\caption{Lower-bound discount factors under different performance loss tolerances and rating update errors.} \label{fig:tradeoffs}
\end{figure}

In Fig.~\ref{fig:tradeoffs}, we illustrate the lower-bound discount factors under different performance loss tolerances and rating update errors. As
expected, when the performance loss tolerance becomes larger, the lower-bound discount factor becomes smaller. What is unexpected is how the
lower-bound discount factor changes with the rating update error. Specifically, the lower-bound discount factor decreases initially with the increase
of the error, and then increases with the error. It is intuitive to see the discount factor increases with the rating update error, because the users
need to be more patient when the rating update is more erroneous. The initial decrease of the discount factor in the error can be explained as
follows. If the rating update error is extremely small, the punishment for the rating-$0$ users in the optimal rating update rule needs to be very
severe (i.e. a smaller $\beta_0^+$ and a larger $\beta_0^-$). Hence, once a user is assigned with rating $0$, it needs to be more patient to carry
out the severe punishment (i.e. weigh the future payoffs more).

Finally, we illustrate the robustness of the proposed mechanisms with respect to the estimation of rating update errors. Suppose that the rating
update error is $\varepsilon$. However, the designer cannot accurately measure this error. Under the estimated error $\hat{\varepsilon}$, the rating
mechanism will construct another recommended strategy. In Fig.~\ref{fig:robustness}, we illustrate the performance gain/loss in terms of social
welfare under the estimated error $\hat{\varepsilon}$, when the rating update error is $\varepsilon$. We can see that there is less than $5\%$
performance variance when the estimation inaccuracy is less than $50\%$. The performance variance is larger when the rating update error is larger.

\begin{figure}
\centering
\includegraphics[width =2.5in]{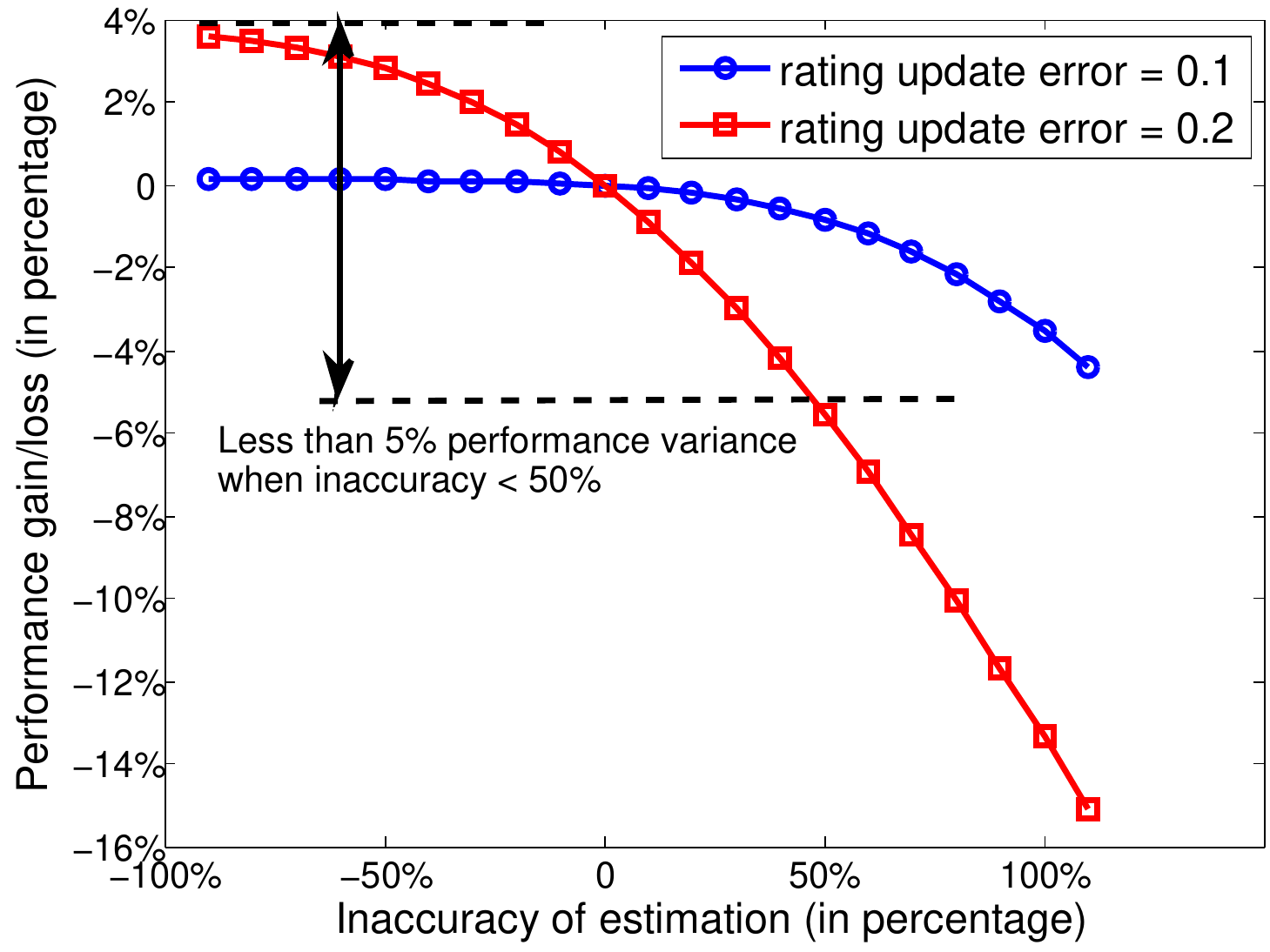}
\caption{Performance gain/loss (in percentage) under different inaccuracy of estimation (in percentage).} \label{fig:robustness}
\end{figure}

\section{Conclusion}\label{sec:Conclusion}
In this paper, we proposed a design framework for simple binary rating mechanisms that can achieve the social optimum in the presence of rating
update errors. We provided design guidelines for the optimal rating update rules, and an algorithm to construct the optimal nonstationary recommended
strategy. The key design principles that enable the rating mechanism to achieve the social optimum are the differential punishments, and the
nonstationary strategies that reduce the performance loss while providing enough incentives. We also reduced the complexity of computing the
recommended strategy by proving that using three recommended plans is enough to achieve the social optimum. The proposed rating mechanism is the
first one that can achieve the social optimum even when the rating update errors are large. Simulation results demonstrated the significant
performance gain of the proposed rating mechanism over the state-of-the-art mechanisms, especially when the rating update error is large.




\appendices

\medskip


\section{Proof of Proposition~\ref{proposition:RestrictedStrategies}}\label{Proof:RestrictedStrategies}
\subsection{The Claim to Prove}
In order to prove Proposition~\ref{proposition:RestrictedStrategies}, we quantify the performance loss of strategies restricted to $A^{\rm as}$. The
performance loss is determined in the following claim:

{\bf \emph{Claim}:} Starting from any initial rating profile $\bm{\theta}$, the maximum social welfare achievable at the PAE by
$(\pi_0,\pi\cdot\bm{1}_N)\in \Pi(A^{\rm as})\times\Pi^N(A^{\rm as})$ is at most
\begin{eqnarray}\label{eqn:BoundedAway_StronglySymmetric}
b-c-c\cdot \rho(\bm{\theta},\alpha_0^*,S_B^*) \sum_{\bm{s}^\prime\in S_B^*} q(\bm{s}^\prime|\bm{\theta},\alpha_0^*,\alpha^{\rm a}\cdot \bm{1}_N),
\end{eqnarray}
where $\alpha_0^*$, the optimal recommended plan, and $S_B^*$, the optimal subset of rating distributions, are the solutions to the following
optimization problem:
\begin{eqnarray}\label{eqn:StronglySymmetric_Optimization}
& \min_{\alpha_0} \min_{S_B\subset S} & \left\{\rho(\bm{\theta},\alpha_0,S_B) \sum_{\bm{s}^\prime\in S_B}
q(\bm{s}^\prime|\bm{\theta},\alpha_0,\alpha^{\rm a}\cdot \bm{1}_N)\right\} \\
& s.t. & \!\!\!\!\sum_{\bm{s}^\prime\in S\setminus S_B} q(\bm{s}^\prime|\bm{\theta},\alpha_0,\alpha^{\rm a}\cdot
\bm{1}_N) > \sum_{\bm{s}^\prime\in S\setminus S_B} q(\bm{s}^\prime|\bm{\theta},\alpha_0,\alpha_i=\alpha^{0},\alpha^{\rm a}\cdot \bm{1}_N),~\forall i\in\mathcal{N}, \nonumber \\
& & \!\!\!\!\sum_{\bm{s}^\prime\in S\setminus S_B} q(\bm{s}^\prime|\bm{\theta},\alpha_0,\alpha^{\rm a}\cdot
\bm{1}_N) > \sum_{\bm{s}^\prime\in S\setminus S_B} q(\bm{s}^\prime|\bm{\theta},\alpha_0,\alpha_i=\alpha^{1},\alpha^{\rm a}\cdot \bm{1}_N),~\forall i\in\mathcal{N}, \nonumber \\
& & \!\!\!\!\sum_{\bm{s}^\prime\in S\setminus S_B} q(\bm{s}^\prime|\bm{\theta},\alpha_0,\alpha^{\rm a}\cdot \bm{1}_N) > \sum_{\bm{s}^\prime\in
S\setminus S_B} q(\bm{s}^\prime|\bm{\theta},\alpha_0,\alpha_i=\alpha^{01},\alpha^{\rm a}\cdot \bm{1}_N),~\forall i\in\mathcal{N}, \nonumber
\end{eqnarray}
where $\rho(\bm{\theta},\alpha_0,S_B)$ is defined as
\begin{eqnarray}\label{eqn:StronglySymmetric_rho}
\rho(\bm{\theta},\alpha_0,S_B) \triangleq \max_{i\in\mathcal{N}} \max\left\{\frac{\frac{s_{\theta_i}-1}{N-1}}{\sum_{\bm{s}^\prime\in S\setminus S_B}
q(\bm{s}^\prime|\bm{\theta},\alpha_0,\alpha^{\rm a}\cdot \bm{1}_N)-q(\bm{s}^\prime|\bm{\theta},\alpha_0,\alpha_i=\alpha^{0},\alpha^{\rm a}\cdot
\bm{1}_N)}\right., \\
\frac{\frac{s_{1-\theta_i}}{N-1}}{\sum_{\bm{s}^\prime\in S\setminus S_B} q(\bm{s}^\prime|\bm{\theta},\alpha_0,\alpha^{\rm a}\cdot
\bm{1}_N)-q(\bm{s}^\prime|\bm{\theta},\alpha_0,\alpha_i=\alpha^{1},\alpha^{\rm a}\cdot \bm{1}_N)}, \nonumber\\
\left.\frac{1}{\sum_{\bm{s}^\prime\in S\setminus S_B} q(\bm{s}^\prime|\bm{\theta},\alpha_0,\alpha^{\rm a}\cdot
\bm{1}_N)-q(\bm{s}^\prime|\bm{\theta},\alpha_0,\alpha_i=\alpha^{01},\alpha^{\rm a}\cdot \bm{1}_N)}\right\}, \nonumber
\end{eqnarray}
where $\alpha^0$ (resp. $\alpha^1$) is the plan in which the user does not serve rating-0 (resp. rating-1) users, and $\alpha^{01}$ is the plan in
which the user does not serve anyone.

The above claim shows that
$$
W(\varepsilon,\delta,A^{\rm as})\leq b-c-c\cdot \rho(\bm{\theta},\alpha_0^*,S_B^*) \sum_{\bm{s}^\prime\in S_B^*}
q(\bm{s}^\prime|\bm{\theta},\alpha_0^*,\alpha^{\rm a}\cdot \bm{1}_N)
$$
for any $\delta$. By defining
$$\zeta(\varepsilon)\triangleq c\cdot
\rho(\bm{\theta},\alpha_0^*,S_B^*) \sum_{\bm{s}^\prime\in S_B^*} q(\bm{s}^\prime|\bm{\theta},\alpha_0^*,\alpha^{\rm a}\cdot \bm{1}_N),
$$
we obtain the result in Proposition~1, namely $\lim_{\delta\rightarrow1} W(\varepsilon,\delta,A^{\rm as})\leq b-c-\zeta(\varepsilon)$. Note that
$\zeta(\varepsilon)$ is indeed a function of the rating update error $\varepsilon$, because $\varepsilon$ determines the state transition function
$q(\bm{s}^\prime|\bm{\theta},\alpha_0^*,\alpha^{\rm a}\cdot \bm{1}_N)$, and thus affects $\rho(\bm{\theta},\alpha_0,S_B)$. Note also that
$\zeta(\varepsilon)$ is independent of the discount factor $\delta$.

In the expression of $\zeta(\varepsilon)$, $\rho(\bm{\theta},\alpha_0,S_B)$ represents the normalized benefit from deviation (normalized by $b-c$).
The numerator of $\rho(\bm{\theta},\alpha_0,S_B)$ is the probability of a player matched to the type of clients whom it deviates to not serve. The
higher this probability, the larger benefit from deviation a player can get. The denominator of $\rho(\bm{\theta},\alpha_0,S_B)$ is the difference
between the two state transition probabilities when the player does and does not deviate, respectively. When the above two transition probabilities
are closer, it is less likely to detect the deviation, which results in a larger $\rho(\bm{\theta},\alpha_0,S_B)$. Hence, we can expect that a larger
$\rho(\bm{\theta},\alpha_0,S_B)$ (i.e. a larger benefit from deviation) will result in a larger performance loss, which is indeed true as will be
proved later.

We can also see that $\zeta(\varepsilon)>0$ as long as $\varepsilon>0$. The reason is as follows. Suppose that $\varepsilon>0$. First, from
\eqref{eqn:StronglySymmetric_rho}, we know that $\rho(\bm{\theta},\alpha_0,S_B)>0$ for any $\bm{\theta}$, $\alpha$, and $S_B\neq\emptyset$. Second,
we can see that $\sum_{\bm{s}^\prime\in S_B^*} q(\bm{s}^\prime|\bm{\theta},\alpha_0^*,\alpha^{\rm a}\cdot \bm{1}_N)>0$ as long as
$S_B^*\neq\emptyset$. Since $S_B^*=\emptyset$ cannot be the solution to the optimization problem \eqref{eqn:StronglySymmetric_Optimization} (because
$S_B^*=\emptyset$ violates the constraints), we know that $\zeta(\varepsilon)>0$.

\subsection{Proof of the Claim}
We prove that for any self-generating set $(\mathcal{W}^{\bm{\theta}})_{\bm{\theta}\in\Theta^N}$, the maximum payoff in
$(\mathcal{W}^{\bm{\theta}})_{\bm{\theta}\in\Theta^N}$, namely $\max_{\bm{\theta}\in\Theta^N}\max_{\bm{v}\in\mathcal{W}^{\bm{\theta}}}
\max_{i\in\mathcal{N}} v_i$, is bounded away from the social optimum $b-c$, regardless of the discount factor. In this way, we can prove that any
equilibrium payoff is bounded away from the social optimum. In addition, we analytically quantify the efficiency loss, which is independent of the
discount factor.

Since the strategies are restricted on the subset of plans $A^{\rm as}$, in each period, all the users will receive the same stage-game payoff,
either $(b-c)$ or $0$, regardless of the matching rule and the rating profile. Hence, the expected discounted average payoff for each user is the
same. More precisely, at any given history $\bm{h}^t=(\bm{\theta}^0,\ldots,\bm{\theta}^t)$, we have
\begin{eqnarray}
U_i(\bm{\theta}^t,\pi_0|_{\bm{h}^t},\pi|_{\bm{h}^t}\cdot\bm{1}_N) = U_j(\bm{\theta}^t,\pi_0|_{\bm{h}^t},\pi|_{\bm{h}^t}\cdot\bm{1}_N),~\forall
i,j\in\mathcal{N},
\end{eqnarray}
for any $(\pi_0,\pi\cdot\bm{1}_N)\in\Pi(A^{\rm as})\times\Pi^N(A^{\rm as})$. As a result, when we restrict to the plan set $A^{\rm as}$, the
self-generating set $(\mathcal{W}^{\bm{\theta}})_{\bm{\theta}\in\Theta^N}$ satisfies for any $\bm{\theta}$ and any
$\bm{v}\in\mathcal{W}^{\bm{\theta}}$
\begin{eqnarray}
v_i=v_j,~\forall i,j\in\mathcal{N}.
\end{eqnarray}

Given any self-generating set $(\mathcal{W}^{\bm{\theta}})_{\bm{\theta}\in\Theta^N}$, define the maximum payoff $\bar{v}$ as
\begin{eqnarray}
\bar{v}\triangleq\max_{\bm{\theta}\in\Theta^N}\max_{\bm{v}\in\mathcal{W}^{\bm{\theta}}} \max_{i\in\mathcal{N}} v_i.
\end{eqnarray}
Now we derive the upper bound of $\bar{v}$ by looking at the decomposability constraints.

To decompose the payoff profile $\bar{v}\cdot\bm{1}_N$, we must find a recommended plan $\alpha_0\in A^{\rm as}$, a plan profile
$\alpha\cdot\bm{1}_N$ with $\alpha\in A^{\rm as}$, and a continuation payoff function $\bm{\gamma}:\Theta^N\rightarrow
\cup_{\bm{\theta}^\prime\in\Theta^N} \mathcal{W}^{\bm{\theta}^\prime}$ with $\bm{\gamma}(\bm{\theta}^\prime)\in\mathcal{W}^{\bm{\theta}^\prime}$,
such that for all $i\in\mathcal{N}$ and for all $\alpha_i\in A$,
\begin{eqnarray}
\bar{v} &=& (1-\delta)u_i(\bm{\theta},\alpha_0,\alpha\cdot \bm{1}_N) + \delta \sum_{\bm{\theta}^\prime} \gamma_i(\bm{\theta}^\prime)
q(\bm{\theta}^\prime|\bm{\theta},\alpha_0,\alpha\cdot \bm{1}_N) \\
&\geq& (1-\delta)u_i(\bm{\theta},\alpha_0,\alpha_i,\alpha\cdot \bm{1}_{N-1}) + \delta \sum_{\bm{\theta}^\prime} \gamma_i(\bm{\theta}^\prime)
q(\bm{\theta}^\prime|\bm{\theta},\alpha_0,\alpha_i,\alpha\cdot \bm{1}_{N-1}). \nonumber
\end{eqnarray}
Note that we do \emph{not} require the users' plan $\alpha$ to be the same as the recommended plan $\alpha_0$, and that we also do \emph{not} require
the continuation payoff function $\bm{\gamma}$ to be a simple continuation payoff function.

First, the payoff profile $\bar{v}\cdot\bm{1}_N$ cannot be decomposed by a recommended plan $\alpha_0$ and the selfish plan $\alpha^{\rm s}$.
Otherwise, since $\bm{\gamma}(\bm{\theta}^\prime)\in\mathcal{W}^{\bm{\theta}^\prime}$, we have
\begin{eqnarray}
\bar{v} = (1-\delta)\cdot 0 + \delta \sum_{\bm{\theta}^\prime} \gamma_i(\bm{\theta}^\prime) q(\bm{\theta}^\prime|\bm{\theta},\alpha_0,\alpha^{\rm
a}\cdot \bm{1}_N) \leq \delta \sum_{\bm{\theta}^\prime} \bar{v}_i \cdot q(\bm{\theta}^\prime|\bm{\theta},\alpha_0,\alpha^{\rm a}\cdot \bm{1}_N)
=\delta \cdot \bar{v} < \bar{v}, \nonumber
\end{eqnarray}
which is a contradiction.


Since we must use a recommended plan $\alpha_0$ and the altruistic plan $\alpha^{\rm a}$ to decompose $\bar{v}\cdot\bm{1}_N$, we can rewrite the
decomposability constraint as
\begin{eqnarray}
\bar{v} &=& (1-\delta)(b-c) + \delta \sum_{\bm{\theta}^\prime} \gamma_i(\bm{\theta}^\prime)
q(\bm{\theta}^\prime|\bm{\theta},\alpha_0,\alpha^{\rm a}\cdot \bm{1}_N) \\
&\geq& (1-\delta)u_i(\bm{\theta},\alpha_0,\alpha_i,\alpha^{\rm a}\cdot \bm{1}_{N-1}) + \delta \sum_{\bm{\theta}^\prime} \gamma_i(\bm{\theta}^\prime)
q(\bm{\theta}^\prime|\bm{\theta},\alpha_0,\alpha_i,\alpha^{\rm a}\cdot \bm{1}_{N-1}). \nonumber
\end{eqnarray}

Since the continuation payoffs under different rating profiles $\bm{\theta},\bm{\theta}^\prime$ that have the same rating distribution
$\bm{s}(\bm{\theta})=\bm{s}(\bm{\theta}^\prime)$ are the same, namely $\bm{\gamma}(\bm{\theta})=\bm{\gamma}(\bm{\theta}^\prime)$, the continuation
payoff depends only on the rating distribution. For notational simplicity, with some abuse of notation, we write $\bm{\gamma}(\bm{s})$ as the
continuation payoff when the rating distribution is $\bm{s}$, write $q(\bm{s}^\prime|\bm{\theta},\alpha_0,\alpha_i,\alpha^{\rm a}\cdot \bm{1}_{N-1})$
as the probability that the next state has a rating distribution $\bm{s}^\prime$, and write $u_i(\bm{s},\alpha^{\rm a},\alpha_i,\alpha^{\rm a}\cdot
\bm{1}_{N-1})$ as the stage-game payoff when the next state has a rating distribution $\bm{s}$. Then the decomposability constraint can be rewritten
as
\begin{eqnarray}
\bar{v} &=& (1-\delta)(b-c) + \delta \sum_{\bm{s}^\prime} \gamma_i(\bm{s}^\prime)
q(\bm{s}^\prime|\bm{\theta},\alpha_0,\alpha^{\rm a}\cdot \bm{1}_N) \\
&\geq& (1-\delta)u_i(\bm{s},\alpha_0,\alpha_i,\alpha^{\rm a}\cdot \bm{1}_{N-1}) + \delta \sum_{\bm{s}^\prime} \gamma_i(\bm{s}^\prime)
q(\bm{s}^\prime|\bm{\theta},\alpha_0,\alpha_i,\alpha^{\rm a}\cdot \bm{1}_{N-1}). \nonumber
\end{eqnarray}

Now we focus on a subclass of continuation payoff functions, and derive the maximum payoff $\bar{v}$ achievable under this subclass of continuation
payoff functions. Later, we will prove that we cannot increase $\bar{v}$ by choosing other continuation payoff functions. Specifically, we focus on a
subclass of continuation payoff functions that satisfy
\begin{eqnarray}
\gamma_i(\bm{s}) &=& x_A,~\forall i\in\mathcal{N},~\forall \bm{s}\in S_A \subset S, \\
\gamma_i(\bm{s}) &=& x_B,~\forall i\in\mathcal{N},~\forall \bm{s}\in S_B \subset S,
\end{eqnarray}
where $S_A$ and $S_B$ are subsets of the set of rating distributions $S$ that have no intersection, namely $S_A\cap S_B=\emptyset$. In other words,
we assign the two continuation payoff values to two subsets of rating distributions, respectively. Without loss of generality, we assume $x_A\geq
x_B$.

Now we derive the incentive compatibility constraints. There are three plans to deviate to, the plan $\alpha^0$ in which the user does not serve
users with rating $0$, the plan $\alpha^1$ in which the user does not serve users with rating $1$, and the plan $\alpha^{01}$ in which the user does
not serve anyone. The corresponding incentive compatibility constraints for a user $i$ with rating $\theta_i=1$ are
\begin{eqnarray}
\left[\sum_{\bm{s}^\prime\in S_A} q(\bm{s}^\prime|\bm{\theta},\alpha_0,\alpha^{\rm a}\cdot
\bm{1}_N)-q(\bm{s}^\prime|\bm{\theta},\alpha_0,\alpha_i=\alpha^{0},\alpha^{\rm a}\cdot \bm{1}_N)\right](x_A-x_B) &\geq& \frac{1-\delta}{\delta} \frac{s_0}{N-1} c, \nonumber \\
\left[\sum_{\bm{s}^\prime\in S_A} q(\bm{s}^\prime|\bm{\theta},\alpha_0,\alpha^{\rm a}\cdot
\bm{1}_N)-q(\bm{s}^\prime|\bm{\theta},\alpha_0,\alpha_i=\alpha^{1},\alpha^{\rm a}\cdot \bm{1}_N)\right](x_A-x_B) &\geq& \frac{1-\delta}{\delta} \frac{s_1-1}{N-1} c, \nonumber \\
\left[\sum_{\bm{s}^\prime\in S_A} q(\bm{s}^\prime|\bm{\theta},\alpha_0,\alpha^{\rm a}\cdot
\bm{1}_N)-q(\bm{s}^\prime|\bm{\theta},\alpha_0,\alpha_i=\alpha^{01},\alpha^{\rm a}\cdot \bm{1}_N)\right](x_A-x_B) &\geq& \frac{1-\delta}{\delta} c.
\end{eqnarray}

Similarly, the corresponding incentive compatibility constraints for a user $j$ with rating $\theta_j=0$ are
\begin{eqnarray}
\left[\sum_{\bm{s}^\prime\in S_A} q(\bm{s}^\prime|\bm{\theta},\alpha_0,\alpha^{\rm a}\cdot
\bm{1}_N)-q(\bm{s}^\prime|\bm{\theta},\alpha_0,\alpha_j=\alpha^{0},\alpha^{\rm a}\cdot \bm{1}_N)\right](x_A-x_B) &\geq& \frac{1-\delta}{\delta} \frac{s_0-1}{N-1} c, \nonumber \\
\left[\sum_{\bm{s}^\prime\in S_A} q(\bm{s}^\prime|\bm{\theta},\alpha_0,\alpha^{\rm a}\cdot
\bm{1}_N)-q(\bm{s}^\prime|\bm{\theta},\alpha_0,\alpha_j=\alpha^{1},\alpha^{\rm a}\cdot \bm{1}_N)\right](x_A-x_B) &\geq& \frac{1-\delta}{\delta} \frac{s_1}{N-1} c, \nonumber \\
\left[\sum_{\bm{s}^\prime\in S_A} q(\bm{s}^\prime|\bm{\theta},\alpha_0,\alpha^{\rm a}\cdot
\bm{1}_N)-q(\bm{s}^\prime|\bm{\theta},\alpha_0,\alpha_j=\alpha^{01},\alpha^{\rm a}\cdot \bm{1}_N)\right](x_A-x_B) &\geq& \frac{1-\delta}{\delta} c.
\end{eqnarray}

We can summarize the above incentive compatibility constraints as
\begin{eqnarray}
x_A-x_B \geq \frac{1-\delta}{\delta} c \cdot \rho(\bm{\theta},\alpha_0,S_A),
\end{eqnarray}
where
\begin{eqnarray}
\rho(\bm{\theta},\alpha_0,S_B) \triangleq \max_{i\in\mathcal{N}} \max\left\{\frac{\frac{s_{\theta_i}-1}{N-1}}{\sum_{\bm{s}^\prime\in S\setminus S_B}
q(\bm{s}^\prime|\bm{\theta},\alpha_0,\alpha^{\rm a}\cdot \bm{1}_N)-q(\bm{s}^\prime|\bm{\theta},\alpha_0,\alpha_i=\alpha^{0},\alpha^{\rm a}\cdot
\bm{1}_N)}\right., \\
\frac{\frac{s_{1-\theta_i}}{N-1}}{\sum_{\bm{s}^\prime\in S\setminus S_B} q(\bm{s}^\prime|\bm{\theta},\alpha_0,\alpha^{\rm a}\cdot
\bm{1}_N)-q(\bm{s}^\prime|\bm{\theta},\alpha_0,\alpha_i=\alpha^{1},\alpha^{\rm a}\cdot \bm{1}_N)}, \\
\left.\frac{1}{\sum_{\bm{s}^\prime\in S\setminus S_B} q(\bm{s}^\prime|\bm{\theta},\alpha_0,\alpha^{\rm a}\cdot
\bm{1}_N)-q(\bm{s}^\prime|\bm{\theta},\alpha_0,\alpha_i=\alpha^{01},\alpha^{\rm a}\cdot \bm{1}_N)}\right\}.
\end{eqnarray}

Since the maximum payoff $\bar{v}$ satisfies
\begin{eqnarray}
\bar{v} = (1-\delta)(b-c) + \delta \left( x_A \sum_{\bm{s}^\prime\in S\setminus S_B} q(\bm{s}^\prime|\bm{\theta},\alpha_0,\alpha^{\rm a}\cdot
\bm{1}_N) + x_B \sum_{\bm{s}^\prime\in S_B} q(\bm{s}^\prime|\bm{\theta},\alpha_0,\alpha^{\rm a}\cdot \bm{1}_N) \right),
\end{eqnarray}
to maximize $\bar{v}$, we choose $x_B=x_A - \frac{1-\delta}{\delta} c \cdot \rho(\bm{\theta},\alpha_0,S_B)$. Since $x_A\geq\bar{v}$, we have
\begin{eqnarray}
\bar{v} &=& (1-\delta)(b-c) + \delta \left( x_A  - \frac{1-\delta}{\delta} c \cdot \rho(\bm{\theta},\alpha_0,S_B) \sum_{\bm{s}^\prime\in S_B}
q(\bm{s}^\prime|\bm{\theta},\alpha_0,\alpha^{\rm a}\cdot \bm{1}_N) \right) \\
&\leq& (1-\delta)(b-c) + \delta \left( \bar{v}  - \frac{1-\delta}{\delta} c \cdot \rho(\bm{\theta},\alpha_0,S_B) \sum_{\bm{s}^\prime\in S_B}
q(\bm{s}^\prime|\bm{\theta},\alpha_0,\alpha^{\rm a}\cdot \bm{1}_N) \right),
\end{eqnarray}
which leads to
\begin{eqnarray}
\bar{v} \leq b-c - c  \cdot \rho(\bm{\theta},\alpha_0,S_B) \sum_{\bm{s}^\prime\in S_B} q(\bm{s}^\prime|\bm{\theta},\alpha_0,\alpha^{\rm a}\cdot
\bm{1}_N).
\end{eqnarray}

Hence, the maximum payoff $\bar{v}$ satisfies
\begin{eqnarray}
\bar{v} \leq b-c - c  \cdot \min_{S_B\subset S} \left\{\rho(\bm{\theta},\alpha_0,S_B) \sum_{\bm{s}^\prime\in S_B}
q(\bm{s}^\prime|\bm{\theta},\alpha_0,\alpha^{\rm a}\cdot \bm{1}_N)\right\},
\end{eqnarray}
where $S_B$ satisfies for all $i\in\mathcal{N}$,
\begin{eqnarray}
\sum_{\bm{s}^\prime\in S\setminus S_B} q(\bm{s}^\prime|\bm{\theta},\alpha_0,\alpha^{\rm a}\cdot
\bm{1}_N) &>& \sum_{\bm{s}^\prime\in S\setminus S_B} q(\bm{s}^\prime|\bm{\theta},\alpha_0,\alpha_i=\alpha^{0},\alpha^{\rm a}\cdot \bm{1}_N), \nonumber \\
\sum_{\bm{s}^\prime\in S\setminus S_B} q(\bm{s}^\prime|\bm{\theta},\alpha_0,\alpha^{\rm a}\cdot
\bm{1}_N) &>& \sum_{\bm{s}^\prime\in S\setminus S_B} q(\bm{s}^\prime|\bm{\theta},\alpha_0,\alpha_i=\alpha^{1},\alpha^{\rm a}\cdot \bm{1}_N), \nonumber \\
\sum_{\bm{s}^\prime\in S\setminus S_B} q(\bm{s}^\prime|\bm{\theta},\alpha_0,\alpha^{\rm a}\cdot \bm{1}_N) &>& \sum_{\bm{s}^\prime\in S\setminus S_B}
q(\bm{s}^\prime|\bm{\theta},\alpha_0,\alpha_i=\alpha^{01},\alpha^{\rm a}\cdot \bm{1}_N).
\end{eqnarray}

Following the same logic as in the proof of Proposition~6 in \cite{Dellarocas}, we can prove that we cannot achieve a higher maximum payoff by other
continuation payoff functions.

\section{Analytical Expression of $\underline{\delta}(\varepsilon,\xi)$}\label{appendix:LowerBoundDiscountFactor}
The lower-bound discount factor $\underline{\delta}(\varepsilon,\xi)$ is the maximum of three critical discount factors, namely
$\underline{\delta}(\varepsilon,\xi)\triangleq\max\{\delta_1(\varepsilon,\xi),\delta_2(\varepsilon,\xi),\delta_3(\varepsilon,\xi)\}$, where
\begin{eqnarray}
\delta_1(\varepsilon,\xi) \triangleq \max_{\theta\in\{0,1\}}
\frac{c}{c+(1-2\varepsilon)(\beta_{\theta}^+-(1-\beta_{\theta}^-))(\xi\frac{\kappa_2}{\kappa_1+\kappa_2})}, \nonumber
\end{eqnarray}
\begin{eqnarray}
\delta_2(\varepsilon,\xi) \triangleq \max_{s_1\in\{1,\ldots,N-1\}: \frac{s_1}{N-1}b+\frac{N-s_1}{N-1}c>\xi\frac{\kappa_2}{\kappa_1+\kappa_2}}
\left\{\frac{\xi\frac{\kappa_2}{\kappa_1+\kappa_2}-\left(\frac{s_1}{N-1}b+\frac{N-s_1}{N-1}c\right)}{(\xi\frac{\kappa_2}{\kappa_1+\kappa_2})\left(x_{s_1}^+-x_0^+\right)-\left(\frac{s_1}{N-1}b+\frac{N-s_1}{N-1}c\right)}\right\},\nonumber
\end{eqnarray}
and
\begin{eqnarray}
\delta_3(\varepsilon,\xi) \triangleq \max_{\theta\in\{0,1\}}
\frac{b-c+c\frac{x_{\theta}^+}{(1-2\varepsilon)\left[\beta_{\theta}^+-(1-\beta_{\theta}^-)\right]}}{b-c+\frac{c\cdot
x_{\theta}^+}{(1-2\varepsilon)\left[\beta_{\theta}^+-(1-\beta_{\theta}^-)\right]} -
\frac{(1+\kappa_1)(\xi\frac{\kappa_2}{\kappa_1+\kappa_2})-z_2}{\kappa_1} -z_3},
\end{eqnarray}
where $z_2 \triangleq -\kappa_1(b-c)+\kappa_1(1-1/\kappa_2)\xi+\xi\frac{\kappa_1}{\kappa_2}/(1+\frac{\kappa_2}{\kappa_1})$, and $z_3 \triangleq
z_2/(\kappa_1+\kappa_2)$. Note that $\frac{(1+\kappa_1)(\xi\frac{\kappa_2}{\kappa_1+\kappa_2})-z_2}{\kappa_1} + z_3<0$. We can see from the above
expressions that $\delta_1(\varepsilon,\xi)<1$ and $\delta_2(\varepsilon,\xi)<1$ as long as $\xi>0$. For $\delta_3(\varepsilon,\xi)$, simple
calculations tell us that $\xi$ appears in the denominator in the form of $-\frac{(2\kappa_1+\kappa_2)\kappa_2}{(\kappa_1+\kappa_2)^2 \kappa_1}\cdot
\xi$. Since $\kappa_1>0$ and $\kappa_2>0$, we know that $-\frac{(2\kappa_1+\kappa_2)\kappa_2}{(\kappa_1+\kappa_2)^2 \kappa_1}<0$. Hence,
$\delta_3(\varepsilon,\xi)$ is increasing in $\xi$. As a result, $\delta_3(\varepsilon,\xi)<1$ as long as $\xi$ is small enough.

Note that all the critical discount factors can be calculated analytically. Specifically, $\delta_1(\varepsilon,\xi)$ and $\delta_3(\varepsilon,\xi)$
are the maximum of two analytically-computed numbers, and $\delta_2(\varepsilon,\xi)$ is the maximum of at most $N-1$ analytically-computed numbers.

\section{Proof of Theorem~\ref{theorem:AchieveSocialOptimum}}\label{Proof:AchieveSocialOptimum}
\subsection{Outline of the proof}

We derive the conditions under which the set $(\mathcal{W}^{\bm{\theta}})_{\bm{\theta}\in\Theta^N}$ is a self-generating set. Specifically, we derive
the conditions under which any payoff profile $\bm{v}\in\mathcal{W}^{\bm{\theta}}$ is decomposable on
$(\mathcal{W}^{\bm{\theta}^\prime})_{\bm{\theta}^\prime\in\Theta^N}$ given $\bm{\theta}$, for all $\bm{\theta}\in\Theta^N$.

\subsection{When users have different ratings}
\subsubsection{Preliminaries}
We first focus on the states $\bm{\theta}$ with $1\leq s_1(\bm{\theta})\leq N-1$, and derive the conditions under which any payoff profile
$\bm{v}\in\mathcal{W}^{\bm{\theta}}$ can be decomposed by $(\alpha_0=\alpha^{\rm a},\alpha^{\rm a}\cdot \bm{1}_N)$ or $(\alpha_0=\alpha^{\rm
f},\alpha^{\rm f}\cdot \bm{1}_N)$. First, $\bm{v}$ could be decomposed by $(\alpha^{\rm a},\alpha^{\rm a}\cdot \bm{1}_N)$, if there exists a
continuation payoff function $\bm{\gamma}:\Theta^N\rightarrow \cup_{\bm{\theta}^\prime\in\Theta^N} \mathcal{W}^{\bm{\theta}^\prime}$ with
$\bm{\gamma}(\bm{\theta}^\prime)\in\mathcal{W}^{\bm{\theta}^\prime}$, such that for all $i\in\mathcal{N}$ and for all $\alpha_i\in A$,
\begin{eqnarray}
v_i &=& (1-\delta)u_i(\bm{\theta},\alpha^{\rm a},\alpha^{\rm a}\cdot \bm{1}_N) + \delta \sum_{\bm{\theta}^\prime} \gamma_i(\bm{\theta}^\prime)
q(\bm{\theta}^\prime|\bm{\theta},\alpha^{\rm a},\alpha^{\rm a}\cdot \bm{1}_N) \\
&\geq& (1-\delta)u_i(\bm{\theta},\alpha^{\rm a},\alpha_i,\alpha^{\rm a}\cdot \bm{1}_{N-1}) + \delta \sum_{\bm{\theta}^\prime}
\gamma_i(\bm{\theta}^\prime) q(\bm{\theta}^\prime|\bm{\theta},\alpha^{\rm a},\alpha_i,\alpha^{\rm a}\cdot \bm{1}_{N-1}). \nonumber
\end{eqnarray}
Since we focus on simple continuation payoff functions, all the users with the same future rating will have the same continuation payoff regardless
of the recommended plan $\alpha_0$, the plan profile $(\alpha_i,\alpha\cdot \bm{1}_{N-1})$, and the future state $\bm{\theta}^\prime$. Hence, we
write the continuation payoffs for the users with future rating $1$ and $0$ as $\gamma^1$ and $\gamma^0$, respectively. Consequently, the above
conditions on decomposability can be simplified to
\begin{eqnarray}\label{eqn:Decomposability_Altruistic_DifferentReputations}
v_i &=& (1-\delta) \cdot u_i(\bm{\theta},\alpha^{\rm a},\alpha^{\rm a}\cdot \bm{1}_N) \\
&+& \delta \left( \gamma^1 \sum_{\bm{\theta}^\prime:\theta_i^\prime=1} q(\bm{\theta}^\prime|\bm{\theta},\alpha^{\rm a},\alpha^{\rm a}\cdot \bm{1}_N)
+ \gamma^0 \sum_{\bm{\theta}^\prime:\theta_i^\prime=0}
q(\bm{\theta}^\prime|\bm{\theta},\alpha^{\rm a},\alpha^{\rm a}\cdot \bm{1}_N) \right) \nonumber \\
&\geq& (1-\delta) \cdot u_i(\bm{\theta},\alpha^{\rm a},\alpha_i,\alpha^{\rm a}\cdot \bm{1}_{N-1}) \nonumber \\
&+& \delta \left( \gamma^1 \sum_{\bm{\theta}^\prime:\theta_i^\prime=1} q(\bm{\theta}^\prime|\bm{\theta},\alpha^{\rm a},\alpha_i,\alpha^{\rm a}\cdot
\bm{1}_{N-1}) + \gamma^0 \sum_{\bm{\theta}^\prime:\theta_i^\prime=0} q(\bm{\theta}^\prime|\bm{\theta},\alpha^{\rm a},\alpha_i,\alpha^{\rm a}\cdot
\bm{1}_{N-1}) \right). \nonumber
\end{eqnarray}

First, consider the case when user $i$ has rating 1 (i.e. $\theta_i=1$). Based on \eqref{eqn:StageGamePayoff}, we can calculate the stage-game payoff
as $u_i(\bm{\theta},\alpha^{\rm a},\alpha^{\rm a}\cdot \bm{1}_N)=b-c$. The term $\sum_{\bm{\theta}^\prime:\theta_i^\prime=1}
q(\bm{\theta}^\prime|\bm{\theta},\alpha^{\rm a},\alpha^{\rm a}\cdot \bm{1}_N)$ is the probability that user $i$ has rating 1 in the next state. Since
user $i$'s rating update is independent of the other users' rating update, we can calculate this probability as
\begin{eqnarray}
\sum_{\bm{\theta}^\prime:\theta_i^\prime=1} q(\bm{\theta}^\prime|\bm{\theta},\alpha^{\rm a},\alpha^{\rm a}\cdot \bm{1}_N) &=&
[(1-\varepsilon)\beta_1^+ + \varepsilon(1-\beta_1^-)] \sum_{m\in M:\theta_{m(i)}=1} \mu(m) \\
&+& [(1-\varepsilon)\beta_1^+ + \varepsilon(1-\beta_1^-)] \sum_{m\in M:\theta_{m(i)}=0} \mu(m) \\
&=& (1-\varepsilon)\beta_1^+ + \varepsilon(1-\beta_1^-) = x_1^+.
\end{eqnarray}
Similarly, we can calculate $\sum_{\bm{\theta}^\prime:\theta_i^\prime=0} q(\bm{\theta}^\prime|\bm{\theta},\alpha^{\rm a},\alpha^{\rm a}\cdot
\bm{1}_N)$, the probability that user $i$ has rating 0 in the next state, as
\begin{eqnarray}
\sum_{\bm{\theta}^\prime:\theta_i^\prime=0} q(\bm{\theta}^\prime|\bm{\theta},\alpha^{\rm a},\alpha^{\rm a}\cdot \bm{1}_N) &=&
[(1-\varepsilon)(1-\beta_1^+) + \varepsilon\beta_1^-] \sum_{m\in M:\theta_{m(i)}=1} \mu(m) \\
&+& [(1-\varepsilon)(1-\beta_1^+) + \varepsilon\beta_1^-] \sum_{m\in M:\theta_{m(i)}=0} \mu(m) \\
&=& (1-\varepsilon)(1-\beta_1^+) + \varepsilon\beta_1^- = 1-x_1^+.
\end{eqnarray}
Now we discuss what happens if user $i$ deviates. Since the recommended plan $\alpha^{\rm a}$ is to exert high effort for all the users, user $i$ can
deviate to the other three plans, namely ``exert high effort for rating-1 users only'', ``exert high effort for rating-0 users only'', ``exert low
effort for all the users''. We can calculate the corresponding stage-game payoff and state transition probabilities under each deviation.
\begin{itemize}
\item ``exert high effort for rating-1 users only'' ($\alpha_i(1,\theta_i)=1,\alpha_i(0,\theta_i)=0$):\\
\begin{eqnarray}
u_i(\bm{\theta},\alpha^{\rm a},\alpha_i,\alpha^{\rm a}\cdot \bm{1}_{N-1}) &=& b - c\cdot \sum_{m\in M:\theta_{m(i)}=1} \mu(m) = b - c\cdot
\frac{s_1(\bm{\theta})-1}{N-1}
\end{eqnarray}
\begin{eqnarray}
&&\sum_{\bm{\theta}^\prime:\theta_i^\prime=1} q(\bm{\theta}^\prime|\bm{\theta},\alpha^{\rm a},\alpha_i,\alpha^{\rm a}\cdot \bm{1}_{N-1}) \\
&=& [(1-\varepsilon)\beta_1^+ + \varepsilon(1-\beta_1^-)] \sum_{m\in M:\theta_{m(i)}=1} \mu(m) + [(1-\varepsilon)(1-\beta_1^-) + \varepsilon \beta_1^+] \sum_{m\in M:\theta_{m(i)}=0} \mu(m) \nonumber \\
&=& [(1-\varepsilon)\beta_1^+ + \varepsilon(1-\beta_1^-)] \frac{s_1(\bm{\theta})-1}{N-1} + [(1-\varepsilon)(1-\beta_1^-) + \varepsilon \beta_1^+]
\frac{s_0(\bm{\theta})}{N-1} \nonumber.
\end{eqnarray}
\begin{eqnarray}
&&\sum_{\bm{\theta}^\prime:\theta_i^\prime=0} q(\bm{\theta}^\prime|\bm{\theta},\alpha^{\rm a},\alpha_i,\alpha^{\rm a}\cdot \bm{1}_{N-1}) \\
&=& [(1-\varepsilon)(1-\beta_1^+) + \varepsilon \beta_1^-] \sum_{m\in M:\theta_{m(i)}=1} \mu(m) + [(1-\varepsilon) \beta_1^- + \varepsilon (1-\beta_1^+)] \sum_{m\in M:\theta_{m(i)}=0} \mu(m) \nonumber \\
&=& [(1-\varepsilon)(1-\beta_1^+) + \varepsilon \beta_1^-] \frac{s_1(\bm{\theta})-1}{N-1} + [(1-\varepsilon) \beta_1^- + \varepsilon (1-\beta_1^+)]
\frac{s_0(\bm{\theta})}{N-1} \nonumber.
\end{eqnarray}
\item ``exert high effort for rating-0 users only'' ($\alpha_i(1,\theta_i)=0,\alpha_i(0,\theta_i)=1$):\\
\begin{eqnarray}
u_i(\bm{\theta},\alpha^{\rm a},\alpha_i,\alpha^{\rm a}\cdot \bm{1}_{N-1}) &=& b - c\cdot \sum_{m\in M:\theta_{m(i)}=0} \mu(m) = b - c\cdot
\frac{s_0(\bm{\theta})}{N-1}
\end{eqnarray}
\begin{eqnarray}
&&\sum_{\bm{\theta}^\prime:\theta_i^\prime=1} q(\bm{\theta}^\prime|\bm{\theta},\alpha^{\rm a},\alpha_i,\alpha^{\rm a}\cdot \bm{1}_{N-1}) \\
&=& [(1-\varepsilon)(1-\beta_1^-) + \varepsilon \beta_1^+] \sum_{m\in M:\theta_{m(i)}=1} \mu(m) + [(1-\varepsilon) \beta_1^+ + \varepsilon (1-\beta_1^-)] \sum_{m\in M:\theta_{m(i)}=0} \mu(m) \nonumber \\
&=& [(1-\varepsilon)(1-\beta_1^-) + \varepsilon \beta_1^+] \frac{s_1(\bm{\theta})-1}{N-1} + [(1-\varepsilon) \beta_1^+ + \varepsilon (1-\beta_1^-)]
\frac{s_0(\bm{\theta})}{N-1} \nonumber.
\end{eqnarray}
\begin{eqnarray}
&&\sum_{\bm{\theta}^\prime:\theta_i^\prime=0} q(\bm{\theta}^\prime|\bm{\theta},\alpha^{\rm a},\alpha_i,\alpha^{\rm a}\cdot \bm{1}_{N-1}) \\
&=& [(1-\varepsilon) \beta_1^- + \varepsilon (1-\beta_1^+)] \sum_{m\in M:\theta_{m(i)}=1} \mu(m) + [(1-\varepsilon) (1-\beta_1^+) + \varepsilon \beta_1^-] \sum_{m\in M:\theta_{m(i)}=0} \mu(m) \nonumber \\
&=& [(1-\varepsilon) \beta_1^- + \varepsilon (1-\beta_1^+)] \frac{s_1(\bm{\theta})-1}{N-1} + [(1-\varepsilon) (1-\beta_1^+) + \varepsilon \beta_1^-]
\frac{s_0(\bm{\theta})}{N-1} \nonumber.
\end{eqnarray}
\item ``exert low effort for all the users'' ($\alpha_i(1,\theta_i)=0,\alpha_i(0,\theta_i)=0$):\\
\begin{eqnarray}
u_i(\bm{\theta},\alpha^{\rm a},\alpha_i,\alpha^{\rm a}\cdot \bm{1}_{N-1}) &=& b
\end{eqnarray}
\begin{eqnarray}
&&\sum_{\bm{\theta}^\prime:\theta_i^\prime=1} q(\bm{\theta}^\prime|\bm{\theta},\alpha^{\rm a},\alpha_i,\alpha^{\rm a}\cdot \bm{1}_{N-1}) \\
&=& [(1-\varepsilon)(1-\beta_1^-) + \varepsilon \beta_1^+] \sum_{m\in M:\theta_{m(i)}=1} \mu(m) + [(1-\varepsilon)(1-\beta_1^-) + \varepsilon \beta_1^+] \sum_{m\in M:\theta_{m(i)}=0} \mu(m) \nonumber \\
&=& (1-\varepsilon)(1-\beta_1^-) + \varepsilon \beta_1^+ \nonumber.
\end{eqnarray}
\begin{eqnarray}
&&\sum_{\bm{\theta}^\prime:\theta_i^\prime=0} q(\bm{\theta}^\prime|\bm{\theta},\alpha^{\rm a},\alpha_i,\alpha^{\rm a}\cdot \bm{1}_{N-1}) \\
&=& [(1-\varepsilon) \beta_1^- + \varepsilon (1-\beta_1^+)] \sum_{m\in M:\theta_{m(i)}=1} \mu(m) + [(1-\varepsilon) \beta_1^- + \varepsilon (1-\beta_1^+)] \sum_{m\in M:\theta_{m(i)}=0} \mu(m) \nonumber \\
&=& (1-\varepsilon) \beta_1^- + \varepsilon (1-\beta_1^+) \nonumber.
\end{eqnarray}
\end{itemize}

Plugging the above expressions into \eqref{eqn:Decomposability_Altruistic_DifferentReputations}, we can simplify the incentive compatibility
constraints (i.e. the inequality constraints) to
\begin{eqnarray}
(1-2\varepsilon) \left[\beta_1^+-(1-\beta_1^-)\right] (\gamma^1-\gamma^0) \geq \frac{1-\delta}{\delta} \cdot c,
\end{eqnarray}
under all three deviating plans.

Hence, if user $i$ has rating $1$, the decomposability constraints \eqref{eqn:Decomposability_Altruistic_DifferentReputations} reduces to
\begin{eqnarray}
v^1 = (1-\delta) \cdot (b-c) + \delta \cdot \left[x_1^+ \gamma^1 + (1-x_1^+) \gamma^0\right],
\end{eqnarray}
where $v^1$ is the payoff of the users with rating $1$, and
\begin{eqnarray}
(1-2\varepsilon) \left[\beta_1^+-(1-\beta_1^-)\right] (\gamma^1-\gamma^0) \geq \frac{1-\delta}{\delta} \cdot c.
\end{eqnarray}

Similarly, if user $i$ has rating $0$, we can reduce the decomposability constraints \eqref{eqn:Decomposability_Altruistic_DifferentReputations} to
\begin{eqnarray}
v^0 = (1-\delta) \cdot (b-c) + \delta \cdot \left[x_0^+ \gamma^1 + (1-x_0^+) \gamma^0\right],
\end{eqnarray}
and
\begin{eqnarray}
(1-2\varepsilon) \left[\beta_0^+-(1-\beta_0^-)\right] (\gamma^1-\gamma^0) \geq \frac{1-\delta}{\delta} \cdot c.
\end{eqnarray}

For the above incentive compatibility constraints (the above two inequalities) to hold, we need to have $\beta_1^+-(1-\beta_1^-)>0$ and
$\beta_0^+-(1-\beta_0^-)>0$, which are part of Condition 1 and Condition 2. Now we will derive the rest of the sufficient conditions in
Theorem~\ref{theorem:AchieveSocialOptimum}.

The above two equalities determine the continuation payoff $\gamma^1$ and $\gamma^0$ as below
\begin{eqnarray}\label{eqn:ContinuationPayoff_Altruistic_DifferentReputations}
\left\{\begin{array}{l}\gamma^1 = \frac{1}{\delta} \cdot \frac{(1-x_0^+)v^1-(1-x_1^+)v^0}{x_1^+ - x_0^+}-\frac{1-\delta}{\delta} \cdot (b-c) \\
\gamma^0 = \frac{1}{\delta} \cdot \frac{x_1^+v^0-x_0^+v^1}{x_1^+ - x_0^+}-\frac{1-\delta}{\delta} \cdot (b-c)\end{array}.\right.
\end{eqnarray}

Now we consider the decomposability constraints if we want to decompose a payoff profile $\bm{v}\in\mathcal{W}^{\bm{\theta}}$ using the fair plan
$\alpha^{\rm f}$. Since we focus on decomposition by simple continuation payoff functions, we write the decomposition constraints as
\begin{eqnarray}\label{eqn:Decomposability_Fair_DifferentReputations}
v_i &=& (1-\delta) \cdot u_i(\bm{\theta},\alpha^{\rm f},\alpha^{\rm f}\cdot \bm{1}_N) \\
&+& \delta \left( \gamma^1 \sum_{\bm{\theta}^\prime:\theta_i^\prime=1} q(\bm{\theta}^\prime|\bm{\theta},\alpha^{\rm f},\alpha^{\rm f}\cdot \bm{1}_N)
+ \gamma^0 \sum_{\bm{\theta}^\prime:\theta_i^\prime=0}
q(\bm{\theta}^\prime|\bm{\theta},\alpha^{\rm f},\alpha^{\rm f}\cdot \bm{1}_N) \right) \nonumber \\
&\geq& (1-\delta) \cdot u_i(\bm{\theta},\alpha^{\rm f},\alpha_i,\alpha^{\rm f}\cdot \bm{1}_{N-1}) \nonumber \\
&+& \delta \left( \gamma^1 \sum_{\bm{\theta}^\prime:\theta_i^\prime=1} q(\bm{\theta}^\prime|\bm{\theta},\alpha^{\rm f},\alpha_i,\alpha^{\rm f}\cdot
\bm{1}_{N-1}) + \gamma^0 \sum_{\bm{\theta}^\prime:\theta_i^\prime=0} q(\bm{\theta}^\prime|\bm{\theta},\alpha^{\rm f},\alpha_i,\alpha^{\rm f}\cdot
\bm{1}_{N-1}) \right). \nonumber
\end{eqnarray}

Due to space limitation, we omit the details and directly give the simplification of the above decomposability constraints as follows. First, the
incentive compatibility constraints (i.e. the inequality constraints) are simplified to
\begin{eqnarray}
(1-2\varepsilon) \left[\beta_1^+-(1-\beta_1^-)\right] (\gamma^1-\gamma^0) \geq \frac{1-\delta}{\delta} \cdot c,
\end{eqnarray}
and
\begin{eqnarray}
(1-2\varepsilon) \left[\beta_0^+-(1-\beta_0^-)\right] (\gamma^1-\gamma^0) \geq \frac{1-\delta}{\delta} \cdot c,
\end{eqnarray}
under all three deviating plans. Note that the above incentive compatibility constraints are the same as the ones when we want to decompose the
payoffs using the altruistic plan $\alpha^{\rm a}$.

Then, the equality constraints in \eqref{eqn:Decomposability_Fair_DifferentReputations} can be simplified as follows. For the users with rating $1$,
we have
\begin{eqnarray}
v^1 = (1-\delta) \cdot \left(b-\frac{s_1(\bm{\theta})-1}{N-1}c\right) + \delta \cdot \left[x_{s_1(\bm{\theta})}^+ \cdot \gamma^1 +
(1-x_{s_1(\bm{\theta})}^+) \cdot \gamma^0\right],
\end{eqnarray}
where
\begin{eqnarray}
x_{s_1(\bm{\theta})}\triangleq \left[(1-\varepsilon)\frac{s_1(\bm{\theta})-1}{N-1}+\frac{s_0(\bm{\theta})}{N-1}\right] \beta_1^+ +
\left(\varepsilon\frac{s_1(\bm{\theta})-1}{N-1}\right) (1-\beta_1^-).
\end{eqnarray}
For the users with rating $0$, we have
\begin{eqnarray}
v^0 = (1-\delta) \cdot \left(\frac{s_0(\bm{\theta})-1}{N-1}b-c\right) + \delta \cdot \left[x_0^+ \gamma^1 + (1-x_0^+) \gamma^0\right].
\end{eqnarray}

The above two equalities determine the continuation payoff $\gamma^1$ and $\gamma^0$ as below
\begin{eqnarray}\label{eqn:ContinuationPayoff_Fair_DifferentReputations}
\left\{\begin{array}{l}\gamma^1 = \frac{1}{\delta} \cdot \frac{(1-x_0^+)v^1-(1-x_{s_1(\bm{\theta})}^+)v^0}{x_{s_1(\bm{\theta})}^+ - x_0^+}-\frac{1-\delta}{\delta} \cdot \frac{\left(b-\frac{s_1(\bm{\theta})-1}{N-1}c\right)(1-x_0^+)-\left(\frac{s_0(\bm{\theta})-1}{N-1}b-c\right)(1-x_{s_1(\bm{\theta})}^+)}{x_{s_1(\bm{\theta})}^+ - x_0^+} \\
\gamma^0 = \frac{1}{\delta} \cdot \frac{x_{s_1(\bm{\theta})}^+ v^0-x_0^+ v^1}{x_{s_1(\bm{\theta})}^+ - x_0^+}-\frac{1-\delta}{\delta} \cdot \frac{\left(b-\frac{s_1(\bm{\theta})-1}{N-1}c\right) x_0^+-\left(\frac{s_0(\bm{\theta})-1}{N-1}b-c\right) x_{s_1(\bm{\theta})}^+}{x_{s_1(\bm{\theta})}^+ - x_0^+} \\
\end{array}.\right.
\end{eqnarray}

\subsubsection{Sufficient conditions}
Now we derive the sufficient conditions under which any payoff profile $\bm{v}\in\mathcal{W}^{\bm{\theta}}$ can be decomposed by
$(\alpha_0=\alpha^{\rm a},\alpha^{\rm a}\cdot \bm{1}_N)$ or $(\alpha_0=\alpha^{\rm f},\alpha^{\rm f}\cdot \bm{1}_N)$. Specifically, we will derive
the conditions such that for any payoff profile $\bm{v}\in\mathcal{W}^{\bm{\theta}}$, at least one of the two decomposability constraints
\eqref{eqn:Decomposability_Altruistic_DifferentReputations} and \eqref{eqn:Decomposability_Fair_DifferentReputations} is satisfied. From the
preliminaries, we know that the incentive compatibility constraints in \eqref{eqn:Decomposability_Altruistic_DifferentReputations} and
\eqref{eqn:Decomposability_Fair_DifferentReputations} can be simplified into the same constraints:
\begin{eqnarray}
(1-2\varepsilon) \left[\beta_1^+-(1-\beta_1^-)\right] (\gamma^1-\gamma^0) \geq \frac{1-\delta}{\delta} \cdot c,
\end{eqnarray}
and
\begin{eqnarray}
(1-2\varepsilon) \left[\beta_0^+-(1-\beta_0^-)\right] (\gamma^1-\gamma^0) \geq \frac{1-\delta}{\delta} \cdot c.
\end{eqnarray}
The above constraints impose the constraint on the discount factor, namely
\begin{eqnarray}
\delta\geq \max_{\theta\in\Theta} \frac{c}{c+(1-2\varepsilon) \left[\beta_{\theta}^+-(1-\beta_{\theta}^-)\right] (\gamma^1-\gamma^0)}.
\end{eqnarray}
Since $\gamma_1$ and $\gamma_0$ should satisfy $\gamma^1 - \gamma^0 \geq \epsilon_0-\epsilon_1$, the above constraints can be rewritten as
\begin{eqnarray}
\delta\geq \max_{\theta\in\Theta} \frac{c}{c+(1-2\varepsilon) \left[\beta_{\theta}^+-(1-\beta_{\theta}^-)\right] (\epsilon_0-\epsilon_1)},
\end{eqnarray}
where is part of Condition 3 in Theorem~\ref{theorem:AchieveSocialOptimum}.

In addition, the continuation payoffs $\gamma_1$ and $\gamma_0$ should satisfy the constraints of the self-generating set, namely
\begin{eqnarray}
\gamma^1 - \gamma^0 &\geq& \epsilon_0-\epsilon_1, \\
\gamma^1 + \frac{c}{(N-1)b} \cdot \gamma^0 &\leq& z_2 \triangleq (1+\frac{c}{(N-1)b})(b-c)-\frac{c}{(N-1)b}\epsilon_0-\epsilon_1, \\
\gamma^1 - \frac{b}{\frac{N-2}{N-1}b-c}\cdot \gamma^0 &\leq& z_3 \triangleq -\frac{\frac{b}{\frac{N-2}{N-1}b-c}-1}{1+\frac{c}{(N-1)b}}\cdot z_2.
\end{eqnarray}

We can plug the expressions of the continuation payoffs $\gamma_1$ and $\gamma_0$ in \eqref{eqn:ContinuationPayoff_Altruistic_DifferentReputations}
and \eqref{eqn:ContinuationPayoff_Fair_DifferentReputations} into the above constraints. Specifically, if a payoff profile $\bm{v}$ is decomposed by
the altruistic plan, the following constraints should be satisfied for the continuation payoff profile to be in the self-generating set: (for
notational simplicity, we define $\kappa_1\triangleq\frac{b}{\frac{N-2}{N-1}b-c}-1$ and $\kappa_2\triangleq 1+\frac{c}{(N-1)b}$)
\begin{align}
&\frac{1}{\delta}\cdot\frac{v^1-v^0}{x_1^+ - x_0^+} \geq \epsilon_0-\epsilon_1, \tag{$\alpha^{\rm a}$-1}\label{eqn:Altruistic_1} \\
&\frac{1}{\delta}\cdot\left\{\frac{(1-\kappa_2 x_0^+)v^1 - (1-\kappa_2 x_1^+)v^0}{x_1^+-x_0^+} - \kappa_2 \cdot (b-c)\right\} \leq z_2 - \kappa_2
\cdot (b-c), \tag{$\alpha^{\rm a}$-2}\label{eqn:Altruistic_2} \\
&\frac{1}{\delta}\cdot\left\{\frac{(1+\kappa_1 x_0^+)v^1 - (1+\kappa_1 x_1^+)v^0}{x_1^+-x_0^+} + \kappa_1 \cdot (b-c)\right\} \leq z_3 + \kappa_1
\cdot (b-c). \tag{$\alpha^{\rm a}$-3}\label{eqn:Altruistic_3}
\end{align}

The constraint \eqref{eqn:Altruistic_1} is satisfied for all $v^1$ and $v^0$ as long as $x_1^+>x_0^+$, because $v^1-v^0>\epsilon_0-\epsilon_1$,
$|x_1^+>x_0^+|<1$, and $\delta<1$.

Since both the left-hand side (LHS) and the right-hand side (RHS) of \eqref{eqn:Altruistic_2} are smaller than $0$, we have
\begin{eqnarray}
\eqref{eqn:Altruistic_2} \Leftrightarrow \delta\leq \frac{\frac{(1-\kappa_2 x_0^+)v^1 - (1-\kappa_2 x_1^+)v^0}{x_1^+-x_0^+} - \kappa_2 \cdot
(b-c)}{z_2 - \kappa_2 \cdot (b-c)}
\end{eqnarray}

The RHS of \eqref{eqn:Altruistic_3} is larger than $0$. Hence, we have
\begin{eqnarray}
\eqref{eqn:Altruistic_3} \Leftrightarrow \delta\geq \frac{\frac{(1+\kappa_1 x_0^+)v^1 - (1+\kappa_1 x_1^+)v^0}{x_1^+-x_0^+} + \kappa_1 \cdot
(b-c)}{z_3 + \kappa_1 \cdot (b-c)}.
\end{eqnarray}

If a payoff profile $\bm{v}$ is decomposed by the fair plan, the following constraints should be satisfied for the continuation payoff profile to be
in the self-generating set:
\begin{align}
&\frac{1}{\delta}\cdot\left\{\frac{v^1-v^0}{x_{s_1(\bm{\theta})}^+ - x_0^+}-\frac{\frac{s_1(\bm{\theta})}{N-1}b+\frac{s_0(\bm{\theta})}{N-1}c}{x_{s_1(\bm{\theta})}^+ - x_0^+}\right\} \geq \epsilon_0-\epsilon_1-\frac{\frac{s_1(\bm{\theta})}{N-1}b+\frac{s_0(\bm{\theta})}{N-1}c}{x_{s_1(\bm{\theta})}^+ - x_0^+}, \tag{$\alpha^{\rm f}$-1}\label{eqn:Fair_1} \\
&\frac{1}{\delta}\cdot\left\{\frac{(1-\kappa_2 x_0^+)v^1 - (1-\kappa_2 x_{s_1(\bm{\theta})}^+)v^0}{x_{s_1(\bm{\theta})}^+-x_0^+} - \frac{(1-\kappa_2
x_0^+)\left(b-\frac{s_1(\bm{\theta})-1}{N-1}c\right) - (1-\kappa_2
x_{s_1(\bm{\theta})}^+)\left(\frac{s_0(\bm{\theta})-1}{N-1}b-c\right)}{x_{s_1(\bm{\theta})}^+-x_0^+}\right\} \nonumber \\
&\leq z_2 - \frac{(1-\kappa_2 x_0^+)\left(b-\frac{s_1(\bm{\theta})-1}{N-1}c\right) - (1-\kappa_2
x_{s_1(\bm{\theta})}^+)\left(\frac{s_0(\bm{\theta})-1}{N-1}b-c\right)}{x_{s_1(\bm{\theta})}^+-x_0^+}, \tag{$\alpha^{\rm f}$-2}\label{eqn:Fair_2} \\
&\frac{1}{\delta}\cdot\left\{\frac{(1+\kappa_1 x_0^+)v^1 - (1+\kappa_1 x_{s_1(\bm{\theta})}^+)v^0}{x_{s_1(\bm{\theta})}^+-x_0^+} - \frac{(1+\kappa_1
x_0^+)\left(b-\frac{s_1(\bm{\theta})-1}{N-1}c\right) - (1+\kappa_1
x_{s_1(\bm{\theta})}^+)\left(\frac{s_0(\bm{\theta})-1}{N-1}b-c\right)}{x_{s_1(\bm{\theta})}^+-x_0^+}\right\} \nonumber \\
&\leq z_3 - \frac{(1+\kappa_1 x_0^+)\left(b-\frac{s_1(\bm{\theta})-1}{N-1}c\right) - (1+\kappa_1
x_{s_1(\bm{\theta})}^+)\left(\frac{s_0(\bm{\theta})-1}{N-1}b-c\right)}{x_{s_1(\bm{\theta})}^+-x_0^+}. \tag{$\alpha^{\rm f}$-3}\label{eqn:Fair_3}
\end{align}

Since $\frac{v^1-v^0}{x_{s_1(\bm{\theta})}^+ - x_0^+}>\epsilon_0-\epsilon_1$, the constraint \eqref{eqn:Fair_1} is satisfied for all $v^1$ and $v^0$
if $v^1-v^0\geq\frac{s_1(\bm{\theta})}{N-1}b+\frac{s_0(\bm{\theta})}{N-1}c$. Hence, the constraint \eqref{eqn:Fair_1} is equivalent to
\begin{eqnarray}
\delta\geq \frac{v^1-v^0-\left(\frac{s_1(\bm{\theta})}{N-1}b+\frac{s_0(\bm{\theta})}{N-1}c\right)}{(\epsilon^0-\epsilon^1)(x_{s_1(\bm{\theta})}^+ -
x_0^+)-\left(\frac{s_1(\bm{\theta})}{N-1}b+\frac{s_0(\bm{\theta})}{N-1}c\right)},~{\rm
for}~\bm{\theta}~s.t.~\frac{s_1(\bm{\theta})}{N-1}b+\frac{s_0(\bm{\theta})}{N-1}c\geq v^1-v^0.
\end{eqnarray}

For \eqref{eqn:Fair_2}, we want to make the RHS have the same (minus) sign under any state $\bm{\theta}$, which is true if
\begin{eqnarray}
1-\kappa_2 x_0^+>0,~1-\kappa_2 x_{s_1(\bm{\theta})}^+<0,~\frac{1-\kappa_2 x_{s_1(\bm{\theta})}^+}{1-\kappa_2 x_0^+}\geq
-(\kappa_2-1),~s_1(\bm{\theta})=1,\ldots,N-1,
\end{eqnarray}
which leads to
\begin{eqnarray}
&&x_{s_1(\bm{\theta})}^+>\frac{1}{\kappa_2},~x_0^+<\frac{1}{\kappa_2},~x_0^+<\frac{1-x_{s_1(\bm{\theta})}^+}{1-\kappa_2},~s_1(\bm{\theta})=1,\ldots,N-1,\\
&\Leftrightarrow& \frac{N-2}{N-1}x_1^+ +
\frac{1}{N-1}\beta_1^+>\frac{1}{\kappa_2},~x_0^+<\min\left\{\frac{1}{\kappa_2},~\frac{1-\beta_1^+}{1-\kappa_2}\right\}.
\end{eqnarray}
Since the RHS of \eqref{eqn:Fair_2} is smaller than $0$, we have
\begin{eqnarray}
\eqref{eqn:Fair_2} \Leftrightarrow \delta\leq \frac{\frac{(1-\kappa_2 x_0^+)v^1 - (1-\kappa_2
x_{s_1(\bm{\theta})}^+)v^0}{x_{s_1(\bm{\theta})}^+-x_0^+} - \frac{(1-\kappa_2 x_0^+)\left(b-\frac{s_1(\bm{\theta})-1}{N-1}c\right) - (1-\kappa_2
x_{s_1(\bm{\theta})}^+)\left(\frac{s_0(\bm{\theta})-1}{N-1}b-c\right)}{x_{s_1(\bm{\theta})}^+-x_0^+}}{z_2 - \frac{(1-\kappa_2
x_0^+)\left(b-\frac{s_1(\bm{\theta})-1}{N-1}c\right) - (1-\kappa_2
x_{s_1(\bm{\theta})}^+)\left(\frac{s_0(\bm{\theta})-1}{N-1}b-c\right)}{x_{s_1(\bm{\theta})}^+-x_0^+}}.
\end{eqnarray}

For \eqref{eqn:Fair_3}, since $\frac{1+\kappa_1 x_{s_1(\bm{\theta})}^+}{1+\kappa_1 x_0^+}<1+\kappa_1$, the RHS is always smaller than $0$. Hence, we
have
\begin{eqnarray}
\eqref{eqn:Fair_3} \Leftrightarrow \delta\leq \frac{\frac{(1+\kappa_1 x_0^+)v^1 - (1+\kappa_1
x_{s_1(\bm{\theta})}^+)v^0}{x_{s_1(\bm{\theta})}^+-x_0^+} - \frac{(1+\kappa_1 x_0^+)\left(b-\frac{s_1(\bm{\theta})-1}{N-1}c\right) - (1+\kappa_1
x_{s_1(\bm{\theta})}^+)\left(\frac{s_0(\bm{\theta})-1}{N-1}b-c\right)}{x_{s_1(\bm{\theta})}^+-x_0^+}}{z_3 - \frac{(1+\kappa_1
x_0^+)\left(b-\frac{s_1(\bm{\theta})-1}{N-1}c\right) - (1+\kappa_1
x_{s_1(\bm{\theta})}^+)\left(\frac{s_0(\bm{\theta})-1}{N-1}b-c\right)}{x_{s_1(\bm{\theta})}^+-x_0^+}}.
\end{eqnarray}

We briefly summarize what requirements on $\delta$ we have obtained now. To make the continuation payoff profile in the self-generating under the
decomposition of $\alpha^{\rm a}$, we have one upper bound on $\delta$ resulting from \eqref{eqn:Altruistic_2} and one lower bound on $\delta$
resulting from \eqref{eqn:Altruistic_3}. To make the continuation payoff profile in the self-generating under the decomposition of $\alpha^{\rm f}$,
we have two upper bounds on $\delta$ resulting from \eqref{eqn:Fair_2} and \eqref{eqn:Fair_3}, and one lower bound on $\delta$ resulting from
\eqref{eqn:Fair_1}. First, we want to eliminate the upper bounds, namely make the upper bounds larger than $1$, such that $\delta$ can be arbitrarily
close to $1$.

To eliminate the following upper bound resulting from \eqref{eqn:Altruistic_2}
\begin{eqnarray}
\delta\leq \frac{\frac{(1-\kappa_2 x_0^+)v^1 - (1-\kappa_2 x_1^+)v^0}{x_1^+-x_0^+} - \kappa_2 \cdot (b-c)}{z_2 - \kappa_2 \cdot (b-c)},
\end{eqnarray}
we need to have (since $z_2 - \kappa_2 \cdot (b-c)<0$)
\begin{eqnarray}
\frac{(1-\kappa_2 x_0^+)v^1 - (1-\kappa_2 x_1^+)v^0}{x_1^+-x_0^+} \leq z_2,~\forall v^1,v^0.
\end{eqnarray}
The LHS of the above inequality is maximized when $v^0=\frac{z_2-z_3}{\kappa_1+\kappa_2}$ and $v^1=v^0+\frac{\kappa_1 z_2+\kappa_2
z_3}{\kappa_1+\kappa_2}$. Hence, the above inequality is satisfied if
\begin{eqnarray}
&& \frac{(1-\kappa_2 x_0^+)\left(\frac{z_2-z_3}{\kappa_1+\kappa_2}+\frac{\kappa_1 z_2+\kappa_2 z_3}{\kappa_1+\kappa_2}\right) - (1-\kappa_2
x_1^+)\frac{z_2-z_3}{\kappa_1+\kappa_2}}{x_1^+-x_0^+} \leq z_2 \\
&\Leftrightarrow& \left(\frac{1-x_1^+ - x_0^+(\kappa_2-1)}{x_1^+ - x_0^+}\frac{\kappa_1}{\kappa_1+\kappa_2}\right) z_2 \leq -\frac{1-x_1^+ -
x_0^+(\kappa_2-1)}{x_1^+ - x_0^+}\frac{\kappa_2}{\kappa_1+\kappa_2} z_3.
\end{eqnarray}
Since $x_0^+<\frac{1-\beta_1^+}{1-\kappa_2}<\frac{1-x_1^+}{1-\kappa_2}$, we have
\begin{eqnarray}
z_2 \leq -\frac{\kappa_2}{\kappa_1} z_3.
\end{eqnarray}

To eliminate the following upper bound resulting from \eqref{eqn:Fair_2}
\begin{eqnarray}
\delta\leq \frac{\frac{(1-\kappa_2 x_0^+)v^1 - (1-\kappa_2 x_{s_1(\bm{\theta})}^+)v^0}{x_{s_1(\bm{\theta})}^+-x_0^+} - \frac{(1-\kappa_2
x_0^+)\left(b-\frac{s_1(\bm{\theta})-1}{N-1}c\right) - (1-\kappa_2
x_{s_1(\bm{\theta})}^+)\left(\frac{s_0(\bm{\theta})-1}{N-1}b-c\right)}{x_{s_1(\bm{\theta})}^+-x_0^+}}{z_2 - \frac{(1-\kappa_2
x_0^+)\left(b-\frac{s_1(\bm{\theta})-1}{N-1}c\right) - (1-\kappa_2
x_{s_1(\bm{\theta})}^+)\left(\frac{s_0(\bm{\theta})-1}{N-1}b-c\right)}{x_{s_1(\bm{\theta})}^+-x_0^+}},
\end{eqnarray}
we need to have (since $z_2 - \frac{(1-\kappa_2 x_0^+)\left(b-\frac{s_1(\bm{\theta})-1}{N-1}c\right) - (1-\kappa_2
x_{s_1(\bm{\theta})}^+)\left(\frac{s_0(\bm{\theta})-1}{N-1}b-c\right)}{x_{s_1(\bm{\theta})}^+-x_0^+}<0$)
\begin{eqnarray}
\frac{(1-\kappa_2 x_0^+)v^1 - (1-\kappa_2 x_{s_1(\bm{\theta})}^+)v^0}{x_{s_1(\bm{\theta})}^+-x_0^+} \leq z_2,~\forall v^1,v^0.
\end{eqnarray}
Similarly, the LHS of the above inequality is maximized when $v^0=\frac{z_2-z_3}{\kappa_1+\kappa_2}$ and $v^1=v^0+\frac{\kappa_1 z_2+\kappa_2
z_3}{\kappa_1+\kappa_2}$. Hence, the above inequality is satisfied if
\begin{eqnarray}
&& \frac{(1-\kappa_2 x_0^+)\left(\frac{z_2-z_3}{\kappa_1+\kappa_2}+\frac{\kappa_1 z_2+\kappa_2 z_3}{\kappa_1+\kappa_2}\right) - (1-\kappa_2
x_{s_1(\bm{\theta})}^+)\frac{z_2-z_3}{\kappa_1+\kappa_2}}{x_{s_1(\bm{\theta})}^+-x_0^+} \leq z_2 \\
&\Leftrightarrow& \left(\frac{1-x_{s_1(\bm{\theta})}^+ - x_0^+(\kappa_2-1)}{x_{s_1(\bm{\theta})}^+ - x_0^+}\frac{\kappa_1}{\kappa_1+\kappa_2}\right)
z_2 \leq -\frac{1-x_{s_1(\bm{\theta})}^+ - x_0^+(\kappa_2-1)}{x_{s_1(\bm{\theta})}^+ - x_0^+}\frac{\kappa_2}{\kappa_1+\kappa_2} z_3.
\end{eqnarray}
Since $x_0^+<\frac{1-\beta_1^+}{1-\kappa_2}<\frac{1-x_{s_1(\bm{\theta})}^+}{1-\kappa_2}$, we have
\begin{eqnarray}
z_2 \leq -\frac{\kappa_2}{\kappa_1} z_3.
\end{eqnarray}

To eliminate the following upper bound resulting from \eqref{eqn:Fair_3}
\begin{eqnarray}
\delta\leq \frac{\frac{(1+\kappa_1 x_0^+)v^1 - (1+\kappa_1 x_{s_1(\bm{\theta})}^+)v^0}{x_{s_1(\bm{\theta})}^+-x_0^+} - \frac{(1+\kappa_1
x_0^+)\left(b-\frac{s_1(\bm{\theta})-1}{N-1}c\right) - (1+\kappa_1
x_{s_1(\bm{\theta})}^+)\left(\frac{s_0(\bm{\theta})-1}{N-1}b-c\right)}{x_{s_1(\bm{\theta})}^+-x_0^+}}{z_3 - \frac{(1+\kappa_1
x_0^+)\left(b-\frac{s_1(\bm{\theta})-1}{N-1}c\right) - (1+\kappa_1
x_{s_1(\bm{\theta})}^+)\left(\frac{s_0(\bm{\theta})-1}{N-1}b-c\right)}{x_{s_1(\bm{\theta})}^+-x_0^+}},
\end{eqnarray}
we need to have (since $z_3 - \frac{(1+\kappa_1 x_0^+)\left(b-\frac{s_1(\bm{\theta})-1}{N-1}c\right) - (1+\kappa_1
x_{s_1(\bm{\theta})}^+)\left(\frac{s_0(\bm{\theta})-1}{N-1}b-c\right)}{x_{s_1(\bm{\theta})}^+-x_0^+}<0$)
\begin{eqnarray}
\frac{(1+\kappa_1 x_0^+)v^1 - (1+\kappa_1 x_{s_1(\bm{\theta})}^+)v^0}{x_{s_1(\bm{\theta})}^+-x_0^+} \leq z_3,~\forall v^1,v^0.
\end{eqnarray}
Again, the LHS of the above inequality is maximized when $v^0=\frac{z_2-z_3}{\kappa_1+\kappa_2}$ and $v^1=v^0+\frac{\kappa_1 z_2+\kappa_2
z_3}{\kappa_1+\kappa_2}$. Hence, the above inequality is satisfied if
\begin{eqnarray}
&& \frac{(1+\kappa_1 x_0^+)\left(\frac{z_2-z_3}{\kappa_1+\kappa_2}+\frac{\kappa_1 z_2+\kappa_2 z_3}{\kappa_1+\kappa_2}\right) - (1+\kappa_1 x_{s_1(\bm{\theta})}^+)\frac{z_2-z_3}{\kappa_1+\kappa_2}}{x_{s_1(\bm{\theta})}^+-x_0^+} \leq z_3 \\
&\Leftrightarrow& \left(\frac{1-x_{s_1(\bm{\theta})}^+ + x_0^+(\kappa_1+1)}{x_{s_1(\bm{\theta})}^+ - x_0^+}\frac{\kappa_1}{\kappa_1+\kappa_2}\right)
z_2 \leq -\frac{1-x_{s_1(\bm{\theta})}^+ + x_0^+(\kappa_1+1)}{x_{s_1(\bm{\theta})}^+ - x_0^+}\frac{\kappa_2}{\kappa_1+\kappa_2} z_3.
\end{eqnarray}
Since $1-x_{s_1(\bm{\theta})}^+ + x_0^+(\kappa_1+1)>0$, we have
\begin{eqnarray}
z_2 \leq -\frac{\kappa_2}{\kappa_1} z_3.
\end{eqnarray}

In summary, to eliminate the upper bounds on $\delta$, we only need to have $z_2 \leq -\frac{\kappa_2}{\kappa_1} z_3$, which is satisfied since we
define $z_3 \triangleq -\frac{\kappa_1}{\kappa_2} z_2$.

Now we derive the analytical lower bound on $\delta$ based on the lower bounds resulting from \eqref{eqn:Altruistic_3} and \eqref{eqn:Fair_1}:
\begin{eqnarray}
\eqref{eqn:Altruistic_3}\Leftrightarrow\delta\geq \frac{\frac{(1+\kappa_1 x_0^+)v^1 - (1+\kappa_1 x_1^+)v^0}{x_1^+-x_0^+} + \kappa_1 \cdot (b-c)}{z_3
+ \kappa_1 \cdot (b-c)},
\end{eqnarray}
and
\begin{eqnarray}
\delta\geq \frac{v^1-v^0-\left(\frac{s_1(\bm{\theta})}{N-1}b+\frac{s_0(\bm{\theta})}{N-1}c\right)}{(\epsilon^0-\epsilon^1)(x_{s_1(\bm{\theta})}^+ -
x_0^+)-\left(\frac{s_1(\bm{\theta})}{N-1}b+\frac{s_0(\bm{\theta})}{N-1}c\right)},~{\rm
for}~\bm{\theta}~s.t.~\frac{s_1(\bm{\theta})}{N-1}b+\frac{s_0(\bm{\theta})}{N-1}c\geq v^1-v^0.
\end{eqnarray}
We define an intermediate lower bound based on the latter inequality along with the inequalities resulting from the incentive compatibility
constraints:
\begin{eqnarray}
\underline{\delta}^\prime=\max\left\{\max_{s_1\in\{1,\ldots,N-1\}: \frac{s_1}{N-1}b+\frac{N-s_1}{N-1}c>\epsilon_0-\epsilon_1}\frac{\epsilon_0-\epsilon_1-\left(\frac{s_1}{N-1}b+\frac{N-s_1}{N-1}c\right)}{(\epsilon_0-\epsilon_1)\left(\frac{N-s_1}{N-1}\beta_1^++\frac{s_1-1}{N-1}x_1^+\right)-\left(\frac{s_1}{N-1}b+\frac{N-s_1}{N-1}c\right)},\right. \nonumber \\
\left.\max_{\theta\in\{0,1\}}\frac{c}{c+(1-2\varepsilon)(\beta_{\theta}^+-(1-\beta_{\theta}^-))(\epsilon_0-\epsilon_1)}\right\}.
\end{eqnarray}
Then the lower bound can be written as $\underline{\delta}=\max\left\{\underline{\delta}^\prime,\underline{\delta}^{\prime\prime}\right\}$, where
$\underline{\delta}^{\prime\prime}$ is the lower bound that we will derive for the case when the users have the same rating. If the payoffs $v^1$ and
$v^0$ satisfy the constraint resulting from \eqref{eqn:Altruistic_3}, namely satisfy
\begin{eqnarray}
\frac{(1+\kappa_1 x_0^+)v^1 - (1+\kappa_1 x_1^+)v^0}{x_1^+-x_0^+} \leq \underline{\delta} z_3 - (1-\underline{\delta}) \kappa_1 \cdot (b-c),
\end{eqnarray}
then we use $\alpha^{\rm a}$ to decompose $v^1$ and $v^0$. Otherwise, we use $\alpha^{\rm f}$ to decompose $v^1$ and $v^0$

\subsection{When the users have the same rating}
Now we derive the conditions under which any payoff profile in $\mathcal{W}^{\bm{1}_N}$ and $\mathcal{W}^{\bm{0}_N}$ can be decomposed.

If all the users have rating $1$, namely $\bm{\theta}=\bm{1}_N$, to decompose $\bm{v}\in\mathcal{W}^{\bm{1}_N}$, we need to find a recommended plan
$\alpha_0$ and a simple continuation payoff function $\bm{\gamma}$ such that for all $i\in\mathcal{N}$ and for all $\alpha_i\in A$,
\begin{eqnarray}
v_i &=& (1-\delta)u_i(\bm{\theta},\alpha_0,\alpha_0\cdot \bm{1}_N) + \delta \sum_{\bm{\theta}^\prime} \gamma_i(\bm{\theta}^\prime)
q(\bm{\theta}^\prime|\bm{\theta},\alpha_0,\alpha_0\cdot \bm{1}_N) \\
&\geq& (1-\delta)u_i(\bm{\theta},\alpha_0,\alpha_i,\alpha_0\cdot \bm{1}_{N-1}) + \delta \sum_{\bm{\theta}^\prime} \gamma_i(\bm{\theta}^\prime)
q(\bm{\theta}^\prime|\bm{\theta},\alpha_0,\alpha_i,\alpha_0\cdot \bm{1}_{N-1}). \nonumber
\end{eqnarray}
When all the users have the same rating, the altruistic plan $\alpha^{\rm a}$ is equivalent to the fair plan $\alpha^{\rm f}$. Hence, we use the
altruistic plan and the selfish plan to decompose the payoff profiles.

If we use the altruistic plan $\alpha^{\rm a}$ to decompose a payoff profile $\bm{v}$, we have
\begin{eqnarray}
v^1 = (1-\delta)(b-c) + \delta (x_1^+ \gamma^1 + (1-x_1^+) \gamma^0),
\end{eqnarray}
and the incentive compatibility constraint
\begin{eqnarray}
(1-2\varepsilon)\left[\beta_1^+-(1-\beta_1^-)\right](\gamma^1-\gamma^0)\geq \frac{1-\delta}{\delta} c.
\end{eqnarray}
Setting $\gamma^1=\gamma^0+\frac{1-\delta}{\delta} \frac{c}{(1-2\varepsilon)\left[\beta_1^+-(1-\beta_1^-)\right]}$ and noticing that
$\gamma^0\in\left[\frac{(1+\kappa_1)(\epsilon_0-\epsilon_1)-z_3}{\kappa_1}, \frac{\kappa_1 z_2+(\kappa_2-1)z_3}{\kappa_1+\kappa_2}\right]$, we get an
lower bound on $v^1$ that can be decomposed by $\alpha^{\rm a}$
\begin{eqnarray}
v^1 &=& (1-\delta)(b-c) + \delta \left(\gamma^0+x_1^+\frac{1-\delta}{\delta} \frac{c}{(1-2\varepsilon)\left[\beta_1^+-(1-\beta_1^-)\right]}\right) \\
&\geq& (1-\delta)\left(b-c+c\frac{x_1^+}{(1-2\varepsilon)\left[\beta_1^+-(1-\beta_1^-)\right]}\right) + \delta
\frac{(1+\kappa_1)(\epsilon_0-\epsilon_1)-z_3}{\kappa_1}
\end{eqnarray}

If we use the selfish plan $\alpha^{\rm s}$ to decompose a payoff profile $\bm{v}$, we have
\begin{eqnarray}
v^1 = \delta (x_1^+ \gamma^1 + (1-x_1^+) \gamma^0).
\end{eqnarray}
Since the selfish plan is NE of the stage game, the incentive compatibility constraint is satisfied as long as we set $\gamma^1=\gamma^0$. Hence, we
have $v^1=\delta \gamma^0$. Again, noticing that $\gamma^0\in\left[\frac{(1+\kappa_1)(\epsilon_0-\epsilon_1)-z_3}{\kappa_1}, \frac{\kappa_1
z_2+(\kappa_2-1)z_3}{\kappa_1+\kappa_2}\right]$, we get an upper bound on $v^1$ that can be decomposed by $\alpha^{\rm s}$
\begin{eqnarray}
v^1 = \delta \gamma^0 \leq \delta \frac{\kappa_1 z_2+(\kappa_2-1)z_3}{\kappa_1+\kappa_2}.
\end{eqnarray}

In order to decompose any payoff profile $\bm{v}\in\mathcal{W}^{\bm{1}_N}$, the lower bound on $v^1$ that can be decomposed by $\alpha^{\rm a}$ must
be smaller than the upper bound on $v^1$ that can be decomposed by $\alpha^{\rm s}$, which leads to
\begin{eqnarray}
&(1-\delta)\left(b-c+c\frac{x_1^+}{(1-2\varepsilon)\left[\beta_1^+-(1-\beta_1^-)\right]}\right) + \delta
\frac{(1+\kappa_1)(\epsilon_0-\epsilon_1)-z_3}{\kappa_1} \leq \delta \frac{\kappa_1 z_2+(\kappa_2-1)z_3}{\kappa_1+\kappa_2} \nonumber \\
&\Rightarrow \delta \geq
\frac{b-c+c\frac{x_1^+}{(1-2\varepsilon)\left[\beta_1^+-(1-\beta_1^-)\right]}}{b-c+c\frac{x_1^+}{(1-2\varepsilon)\left[\beta_1^+-(1-\beta_1^-)\right]}
+ \frac{\kappa_1 z_2+(\kappa_2-1)z_3}{\kappa_1+\kappa_2} - \frac{(1+\kappa_1)(\epsilon_0-\epsilon_1)-z_3}{\kappa_1}}.
\end{eqnarray}

Finally, following the same procedure, we derive the lower bound on $\delta$ when all the users have rating $0$, namely $\bm{\theta}=\bm{0}_N$.
Similarly, in this case, the altruistic plan $\alpha^{\rm a}$ is equivalent to the fair plan $\alpha^{\rm f}$. Hence, we use the altruistic plan and
the selfish plan to decompose the payoff profiles.

If we use the altruistic plan $\alpha^{\rm a}$ to decompose a payoff profile $\bm{v}$, we have
\begin{eqnarray}
v^0 = (1-\delta)(b-c) + \delta (x_0^+ \gamma^1 + (1-x_0^+) \gamma^0),
\end{eqnarray}
and the incentive compatibility constraint
\begin{eqnarray}
(1-2\varepsilon)\left[\beta_0^+-(1-\beta_0^-)\right](\gamma^1-\gamma^0)\geq \frac{1-\delta}{\delta} c.
\end{eqnarray}

If we use the selfish plan $\alpha^{\rm s}$ to decompose a payoff profile $\bm{v}$, we have
\begin{eqnarray}
v^1 = \delta (x_0^+ \gamma^1 + (1-x_0^+) \gamma^0).
\end{eqnarray}

Note that when $\bm{\theta}=\bm{0}_N$, if we substitute $\beta_0^+$, $\beta_0^-$, $x_0^-$ with $\beta_1^+$, $\beta_1^-$, $x_1^-$, respectively, the
decomposability constraints become the same as those when $\bm{\theta}=\bm{1}_N$. Hence, we derive a similar lower bound on $\delta$
\begin{eqnarray}
\delta \geq
\frac{b-c+c\frac{x_0^+}{(1-2\varepsilon)\left[\beta_0^+-(1-\beta_0^-)\right]}}{b-c+c\frac{x_0^+}{(1-2\varepsilon)\left[\beta_0^+-(1-\beta_0^-)\right]}
+ \frac{\kappa_1 z_2+(\kappa_2-1)z_3}{\kappa_1+\kappa_2} - \frac{(1+\kappa_1)(\epsilon_0-\epsilon_1)-z_3}{\kappa_1}}.
\end{eqnarray}

Finally, we can obtain the lower bound on $\delta$ when the users have the same rating as
\begin{eqnarray}
\underline{\delta}^{\prime\prime} = \max_{\theta\in\{0,1\}}
\frac{b-c+c\frac{x_{\theta}^+}{(1-2\varepsilon)\left[\beta_{\theta}^+-(1-\beta_{\theta}^-)\right]}}{b-c+c\frac{x_{\theta}^+}{(1-2\varepsilon)\left[\beta_{\theta}^+-(1-\beta_{\theta}^-)\right]}
+ \frac{\kappa_1 z_2+(\kappa_2-1)z_3}{\kappa_1+\kappa_2} - \frac{(1+\kappa_1)(\epsilon_0-\epsilon_1)-z_3}{\kappa_1}}.
\end{eqnarray}

Together with the lower bound $\underline{\delta}^\prime$ derived for the case when the users have different ratings, we can get the lower bound
$\underline{\delta}$ specified in Condition~3 of Theorem~\ref{theorem:AchieveSocialOptimum}.

\section{Complete Description of the Algorithm}\label{EquilibriumStrategyComplete}

\begin{table}
\renewcommand{\arraystretch}{1.1}
\caption{The algorithm of constructing the equilibrium strategy by the rating mechanism.}{
\begin{tabular}{l}
\hline \hline
\textbf{Require:} $b$, $c$, $\varepsilon$, $\xi$; $\tau(\varepsilon)$, $\delta\geq\underline{\delta}(\varepsilon,\xi)$; $\bm{\theta}^0$ \hfill \emph{(inputs to the algorithm)}\\
\hline
\textbf{Initialization:} $t=0$, $\epsilon_0=\xi$, $\epsilon_1=\epsilon_0/(1+\frac{\kappa_2}{\kappa_1})$, $v^\theta=b-c-\epsilon_\theta$, $\bm{\theta}=\bm{\theta}^0$. \hfill \emph{(set the target payoffs)}\\
\hline
\textbf{repeat} \\
~~~~\textbf{if} $s_1(\bm{\theta})=0$ \textbf{then}  \\
~~~~~~~~\textbf{if} $v^0\geq (1-\delta)\left[b-c+\frac{(1-\varepsilon)\beta_0^++\varepsilon(1-\beta_0^-)}{(1-2\varepsilon)(\beta_0^+-(1-\beta_0^-)}c\right]+\delta\frac{\epsilon_0-\epsilon_1-z_3}{\kappa_1}$ \textbf{then} \\
~~~~~~~~~~~~$\alpha_0^t=\alpha^{\rm a}$ \hfill \emph{(determine the recommended plan)}\\
~~~~~~~~~~~~$v^0\leftarrow \frac{v^0}{\delta}-\frac{1-\delta}{\delta}\left[b-c+\frac{(1-\varepsilon)\beta_0^++\varepsilon(1-\beta_0^-)}{(1-2\varepsilon)(\beta_0^+-(1-\beta_0^-)}c\right]$,$v^1\leftarrow v^0+\frac{1-\delta}{\delta}\left[\frac{1}{(1-2\varepsilon)(\beta_0^+-(1-\beta_0^-)}c\right]$ \hfill \emph{(update the continuation payoff)}\\
~~~~~~~~\textbf{else} \\
~~~~~~~~~~~~$\alpha_0^t=\alpha^{\rm s}$ \hfill \emph{(determine the recommended plan)} \\
~~~~~~~~~~~~$v^0\leftarrow \frac{v^0}{\delta},~v^1\leftarrow v^0$ \hfill \emph{(update the continuation payoff)} \\
~~~~~~~~\textbf{end} \\
~~~~\textbf{elseif} $s_1(\bm{\theta})=N$ \textbf{then}  \\
~~~~~~~~\textbf{if} $v^1\geq (1-\delta)\left[b-c+\frac{(1-\varepsilon)\beta_1^++\varepsilon(1-\beta_1^-)}{(1-2\varepsilon)(\beta_1^+-(1-\beta_1^-)}c\right]+\delta\frac{\epsilon_0-\epsilon_1-z_3}{\kappa_1}$ \textbf{then} \\
~~~~~~~~~~~~$\alpha_0^t=\alpha^{\rm a}$ \hfill \emph{(determine the recommended plan)} \\
~~~~~~~~~~~~$v^1\leftarrow \frac{v^1}{\delta}-\frac{1-\delta}{\delta}\left[b-c+\frac{(1-\varepsilon)\beta_1^++\varepsilon(1-\beta_1^-)}{(1-2\varepsilon)(\beta_1^+-(1-\beta_1^-)}c\right]$, $v^0\leftarrow v^1-\frac{1-\delta}{\delta}\left[\frac{1}{(1-2\varepsilon)(\beta_1^+-(1-\beta_1^-)}c\right]$ \hfill \emph{(update the continuation payoff)} \\
~~~~~~~~\textbf{else} \\
~~~~~~~~~~~~$\alpha_0^t=\alpha^{\rm s}$ \hfill \emph{(determine the recommended plan)} \\
~~~~~~~~~~~~$v^1\leftarrow \frac{v^1}{\delta},~v^0\leftarrow v^1$ \hfill \emph{(update the continuation payoff)} \\
~~~~~~~~\textbf{end} \\
~~~~\textbf{else}  \\
~~~~~~~~\textbf{if} $\frac{1+\kappa_1 x_0^+}{x_1^+-x_0^+}v^1-\frac{1+\kappa_1 x_1^+}{x_1^+-x_0^+}v^0 \leq \delta z_3 - (1-\delta)\kappa_1(b-c)$ \textbf{then} \\
~~~~~~~~~~~~$\alpha_0^t=\alpha^{\rm a}$ \hfill \emph{(determine the recommended plan)} \\
~~~~~~~~~~~~$v^{1\prime}\leftarrow \frac{1}{\delta}\frac{(1-x_0^+)v^1-(1-x_1^+)v^0}{x_1^+ - x_0^+}-\frac{1-\delta}{\delta}(b-c)$, $v^{0\prime}\leftarrow \frac{1}{\delta}\frac{x_1^+v^0-x_0^+v^1}{x_1^+ - x_0^+}-\frac{1-\delta}{\delta}(b-c)$ \hfill \emph{(update the continuation payoff)} \\
~~~~~~~~~~~~$v^1\leftarrow v^{1\prime},~v^0\leftarrow v^{0\prime}$ \\
~~~~~~~~\textbf{else} \\
~~~~~~~~~~~~$\alpha_0^t=\alpha^{\rm f}$ \hfill \emph{(determine the recommended plan)} \\
~~~~~~~~~~~~$v^{1\prime}\leftarrow \frac{1}{\delta}\frac{(1-x_0^+)v^1-(1-x_{s_1(\bm{\theta})}^+)v^0}{x_{s_1(\bm{\theta})}^+-x_0^+}-\frac{1-\delta}{\delta}\frac{(b-\frac{s_1(\bm{\theta})-1}{N-1}c)(1-x_0^+)-(\frac{s_0(\bm{\theta})-1}{N-1}b-c)(1-x_{s_1(\bm{\theta})}^+)}{x_{s_1(\bm{\theta})}^+-x_0^+}$ \hfill \emph{(update the continuation payoff)} \\
~~~~~~~~~~~~$v^{1\prime}\leftarrow \frac{1}{\delta}\frac{x_0^+v^1-x_{s_1(\bm{\theta})}^+v^0}{x_0^+-x_{s_1(\bm{\theta})}^+}-\frac{1-\delta}{\delta}\frac{(b-\frac{s_1(\bm{\theta})-1}{N-1}c)x_0^+-(\frac{s_0(\bm{\theta})-1}{N-1}b-c)x_{s_1(\bm{\theta})}^+}{x_0^+-x_{s_1(\bm{\theta})}^+}$ \\
~~~~~~~~~~~~$v^1\leftarrow v^{1\prime},~v^0\leftarrow v^{0\prime}$ \\
~~~~~~~~\textbf{end} \\
~~~~\textbf{end} \\
~~~~$t\leftarrow t+1$, determine the rating profile $\bm{\theta}^t$, set $\bm{\theta}\leftarrow\bm{\theta}^t$ \\
\textbf{until} $\varnothing$ \\
\hline \hline
\end{tabular}}
\label{table:EquilibriumStrategyComplete}
\end{table}

%

\end{document}